\documentclass[aps,prx,twocolumn,superscriptaddress,nofootinbib,10pt]{revtex4-2}
\usepackage[utf8]{inputenc}
\usepackage[english]{babel}
\usepackage[T1]{fontenc}
\usepackage{amsmath}
\usepackage{amsfonts}
\usepackage{bbm}
\usepackage{amssymb}
\usepackage{amsthm}
\usepackage{braket}
\usepackage[dvipsnames, table]{xcolor}
\usepackage{hyperref}
\hypersetup{
  colorlinks  = true, 
  urlcolor    = Blue, 
  linkcolor   = BrickRed, 
  citecolor   = PineGreen 
}
\usepackage{graphicx}
\usepackage{diagbox}
\usepackage{tikz}
\usetikzlibrary{tikzmark}

\usepackage{enumitem}

\usepackage{chngcntr} 

\usepackage{array}
\newcolumntype{x}[1]{>{\centering\let\newline\\\arraybackslash\hspace{0pt}}p{#1}}

\newcommand{\tsf}[1]{\texttt{#1}}

\newcommand{\plus}{\text{+}}

\newcommand{\id}{\textsf{id}}
\newcommand{\X}{{\textsf{X}}}
\newcommand{\Z}{{\textsf{Z}}}
\newcommand{\Os}{{\textsf{O}}}
\def\paddedtext#1#2{\leavevmode\hbox to#1{#2\hss}\ignorespaces}

\newcommand{\magic}[9]{
  \begin{tabular}{x{3mm}|x{3mm}|x{3mm}}
    {#1} & {#2} & {#3} \\
    \hline
    {#4} & {#5} & {#6} \\
    \hline
    {#7} & {#8} & {#9}
  \end{tabular}
}
\newcommand{\magicc}[9]{
  \begin{tabular}{c|c|c}
    {#1} & {#2} & {#3} \\
    \hline
    {#4} & {#5} & {#6} \\
    \hline
    {#7} & {#8} & {#9}
  \end{tabular}
}
\newcommand{\magicp}[9]{
  \begin{tabular}{c|c|c}
    \paddedtext{3mm}{{#1}} & \paddedtext{3mm}{{#2}} & \paddedtext{3mm}{{#3}} \\
    \hline
    \paddedtext{3mm}{{#4}} & \paddedtext{3mm}{{#5}} & \paddedtext{3mm}{{#6}} \\
    \hline
    \paddedtext{3mm}{{#7}} & \paddedtext{3mm}{{#8}} & \paddedtext{3mm}{{#9}}
  \end{tabular}
}
\newcommand{\magicsf}[9]{
  \begin{tabular}{c|c|c}
    \tsf{#1} & \tsf{#2} & \tsf{#3} \\
    \hline
    \tsf{#4} & \tsf{#5} & \tsf{#6} \\
    \hline
    \tsf{#7} & \tsf{#8} & \tsf{#9}
  \end{tabular}
}

\newtheorem{theorem}{Theorem}
\newtheorem{definition}{Definition}

\newtheorem{lemma}[theorem]{Lemma}

\begin{document}

\title{Competing automorphisms and disordered Floquet codes}

\author{Cory T. Aitchison}
\affiliation{DAMTP, University of Cambridge, Wilberforce Road, Cambridge, CB3
0WA, UK}
\author{Benjamin Béri}
\affiliation{DAMTP, University of Cambridge, Wilberforce Road, Cambridge, CB3
0WA, UK}
\affiliation{T.C.M. Group, Cavendish Laboratory, University of Cambridge, J.J. Thomson Avenue, Cambridge, CB3 0HE, UK\looseness=-1}

\begin{abstract}
  Topological order is a promising basis for quantum error correction, a key
  milestone towards large-scale quantum computing. Floquet codes provide a
  dynamical scheme for this while also exhibiting Floquet-enriched topological
  order (FET) where anyons periodically undergo a measurement-induced
  automorphism that acts uniformly in space. We study disordered Floquet codes
  where automorphisms have a spatiotemporally heterogeneous distribution---the
  automorphisms ``compete''. We characterize the effect of this competition,
  showing how key features of the purification dynamics of mixed codestates can
  be inferred from anyon and automorphism properties for any Abelian topological
  order. This perspective can explain the protection or measurement of logical
  information in a dynamic automorphism (DA) code when subjected to a noise
  model of missing measurements. We demonstrate this using a DA color code with
  perturbed measurement sequences. The framework of competing automorphisms
  captures essential features of Floquet codes and robustness to noise, and may
  elucidate key mechanisms involving topological order, automorphisms, and
  fault-tolerance.
\end{abstract}

\maketitle

\section{Introduction}
Quantum error-correcting (QEC) codes are crucial for the effective, scalable
operation of current and future quantum computers \cite{steaneError1996,
gottesmanStabilizer1997, kitaevQuantum1997, kitaevQuantum1997a, knillTheory1997,
preskillFaultTolerant1998, shorFaulttolerant1997, knillTheory2000,
schusterPolynomialtime2024}. These codes envision a smaller number of protected
logical qubits encoded within a larger Hilbert space of physical qubits.
Similarly, quantum systems with long-range entanglement and topological order
(TO) display robustness to local perturbations, host excitations with fractional
statistics known as anyons, and support a topology-dependent ground-state
degeneracy \cite{wenVacuum1989, wittenQuantum1989}. Because of their inherent
robustness, TOs are promising candidates for QEC codes.

These topological QEC codes have historically been static: the code properties,
such as the stabilizer group in stabilizer codes \cite{gottesmanStabilizer1997,
poulinStabilizer2005}, are fixed through time. A recent class of codes,
so-called dynamical codes, forgo this notion. They are inherently time-evolving
and this can improve their characteristics or enable novel behaviors
\cite{fuError2024, davydovaFloquet2023, zhuQubit2023, zhuNishimoris2023}. The
first such example was the honeycomb code \cite{hastingsDynamically2021}, a form
of dynamical code called a Floquet code due to its time-periodic evolution. By
measuring only two-qubit Pauli operators in a particular sequence, a stabilizer
group emerges that enables the detection and correction of errors
\cite{vuillotPlanar2021, haahBoundaries2022, ellisonFloquet2023}. Each stage in
the sequence generates a TO equivalent to a static toric code ($TC$)
\cite{bravyiQuantum1998, dennisTopological2002, kitaevFaulttolerant2003,
kitaevAnyons2006}. 

The close relationship between TOs and QEC has informed our understanding of
both quantum matter and error-correction protocols. A prominent example is that
of topological defects such as transparent domain walls, which permute the
labels of anyon worldlines crossing the wall \cite{kesselringBoundaries2018,
kesselringAnyon2024}. These permutations, or ``automorphisms'', are elements of
the symmetry group of the anyons that preserves fusion and braiding data. In
$(2+1)$D these domain walls terminate in point-like defects known as twists,
which exhibit exotic behavior including anyon localization and enable
topological quantum computation through their non-Abelian fusion and braiding
\cite{bombinTopological2010a, brownTopological2013}. 

A key feature of the honeycomb Floquet code is that its measurement sequence
implements a global automorphism every period that exchanges two of the $TC$
anyons. The time-periodic nature of these automorphisms extend the TO to a
time-crystalline-like phase \cite{aasenAdiabatic2022}. This Floquet-enriched TO
(FET) arises in dynamical codes through measurements, while FETs were originally
proposed in unitarily driven $(2+1)$D TO systems \cite{potterDynamically2017,
poRadical2017}. Dynamic automorphism (DA) codes extend this notion further by
using measurement-induced automorphisms to construct quantum logic gates
\cite{davydovaQuantum2024}.

An important question is whether these FETs or codes are robust against
perturbations or disorder. For the honeycomb code, recent results suggest
competitive levels of tolerance to fabrication defects in realistic physical
devices \cite{aasenFaultTolerant2023, mclauchlanAccommodating2024}, and that its
FET persists amid random modifications to the measurement sequence, such as
omitting or splitting up the two-qubit Pauli operators \cite{vuStable2024,
sriramTopology2023}. Characterizing the effects of disorder in more general
Floquet codes is, however, an open question. Previous results have largely
focused on homogeneous cases concerning FETs with automorphisms acting globally
\cite{kesselringAnyon2024, davydovaFloquet2023}. Analyses of topological defects
in Floquet codes have primarily considered spatial domain walls
\cite{ellisonFloquet2023, haahBoundaries2022}. Perturbations to temporal domain
walls---fundamentally modifying the Floquet evolution---were analyzed in terms
of microscopic measurements in the honeycomb code~\cite{vuStable2024,
sriramTopology2023}. However, a framework capturing general aspects that follow
from anyon data, i.e., a topological quantum field theory (TQFT) description of
the dynamics, is missing.

Here we study Abelian TOs with measurement-induced automorphisms, and ask how do
they evolve when perturbations cause the temporal domain walls to be
spatiotemporally heterogeneous. These multiple domain walls, or ``competing
automorphisms'', nontrivially influence the codespace. We find that their
evolution and purification dynamics are fundamentally decided by the homology
classes of the domain walls' boundary segments, and the properties of anyons
that localize at the boundaries' spatiotemporal twist defects or are invariant
under the corresponding transition maps (that we shall define in
Section~\ref{sec:background}). We show that such perturbed measurement-induced
FETs can have rich phenomenology and form a new part of the landscape of quantum
matter emerging from disordered measurements \cite{liQuantum2018,
  liMeasurementdriven2019, choiQuantum2020, jianMeasurementinduced2020,
  ippolitiEntanglement2021, fisherRandom2023, skinnerMeasurementInduced2019,
  baoTheory2020, ippolitiPostselectionFree2021, liCross2023, liStatistical2021,
  gullansDynamical2020, nahumEntanglement2020, nahumMeasurement2021,
zabaloCritical2020, sommersDynamically2024, behrendsSurface2024,
sriramTopology2023, lavasaniMeasurementinduced2021}.

We also show how the framework of competing automorphisms can be used to
characterize disordered dynamical QEC codes. To illustrate this, we explore how
perturbations to the circuit protocol in the form of a noise model with missing
measurements affects dynamic automorphism (DA) codes.  Under this noise model,
we show that the code's ability to protect logical information can be readily
understood via its competing automorphisms. We formulate necessary and
sufficient conditions for different implementations of the DA color code to be
unaffected by missing measurements, i.e., for a logical Hilbert subspace to
remain unmeasured despite the disorder. Parameterizing the strength of the
perturbations, we identify connections between the automorphisms of the code,
and characterize the topological space of FETs in this exotic dynamical quantum
matter. Our approach is generalizable to other dynamical codes and TOs
understandable through their automorphisms.

The rest of the paper is organized as follows: in Section~\ref{sec:background},
we provide background on FETs and the DA color code---we use this code as a
motivating example and a concrete illustration for measurement-induced
automorphisms. In Section~\ref{sec:competing}, we introduce competing
automorphisms in dynamical quantum systems. In Section~\ref{sec:disordermodel}
we then introduce a disordered DA color code and analyze its behavior using
competing automorphisms. In Section~\ref{sec:conclusion} we conclude and discuss
some future directions. Appendix~\ref{app:additional} includes additional
background material, while Appendix~\ref{app:disordered} discusses in more
detail the disordered DA color code. Finally, Appendix~\ref{app:critical}
analytically and numerically examines the phases and critical behavior of the DA
color code under our noise models.

\section{\label{sec:background}Background}

\begin{figure}
  \begin{center}
    \def\svgwidth{0.7\columnwidth}
    \begingroup \makeatletter \providecommand\color[2][]{\errmessage{(Inkscape) Color is used for the text in Inkscape, but the package 'color.sty' is not loaded}\renewcommand\color[2][]{}}\providecommand\transparent[1]{\errmessage{(Inkscape) Transparency is used (non-zero) for the text in Inkscape, but the package 'transparent.sty' is not loaded}\renewcommand\transparent[1]{}}\providecommand\rotatebox[2]{#2}\newcommand*\fsize{\dimexpr\f@size pt\relax}\newcommand*\lineheight[1]{\fontsize{\fsize}{#1\fsize}\selectfont}\ifx\svgwidth\undefined \setlength{\unitlength}{767.99990845bp}\ifx\svgscale\undefined \relax \else \setlength{\unitlength}{\unitlength * \real{\svgscale}}\fi \else \setlength{\unitlength}{\svgwidth}\fi \global\let\svgwidth\undefined \global\let\svgscale\undefined \makeatother \begin{picture}(1,1)\lineheight{1}\setlength\tabcolsep{0pt}\put(0,0){\includegraphics[width=\unitlength,page=1]{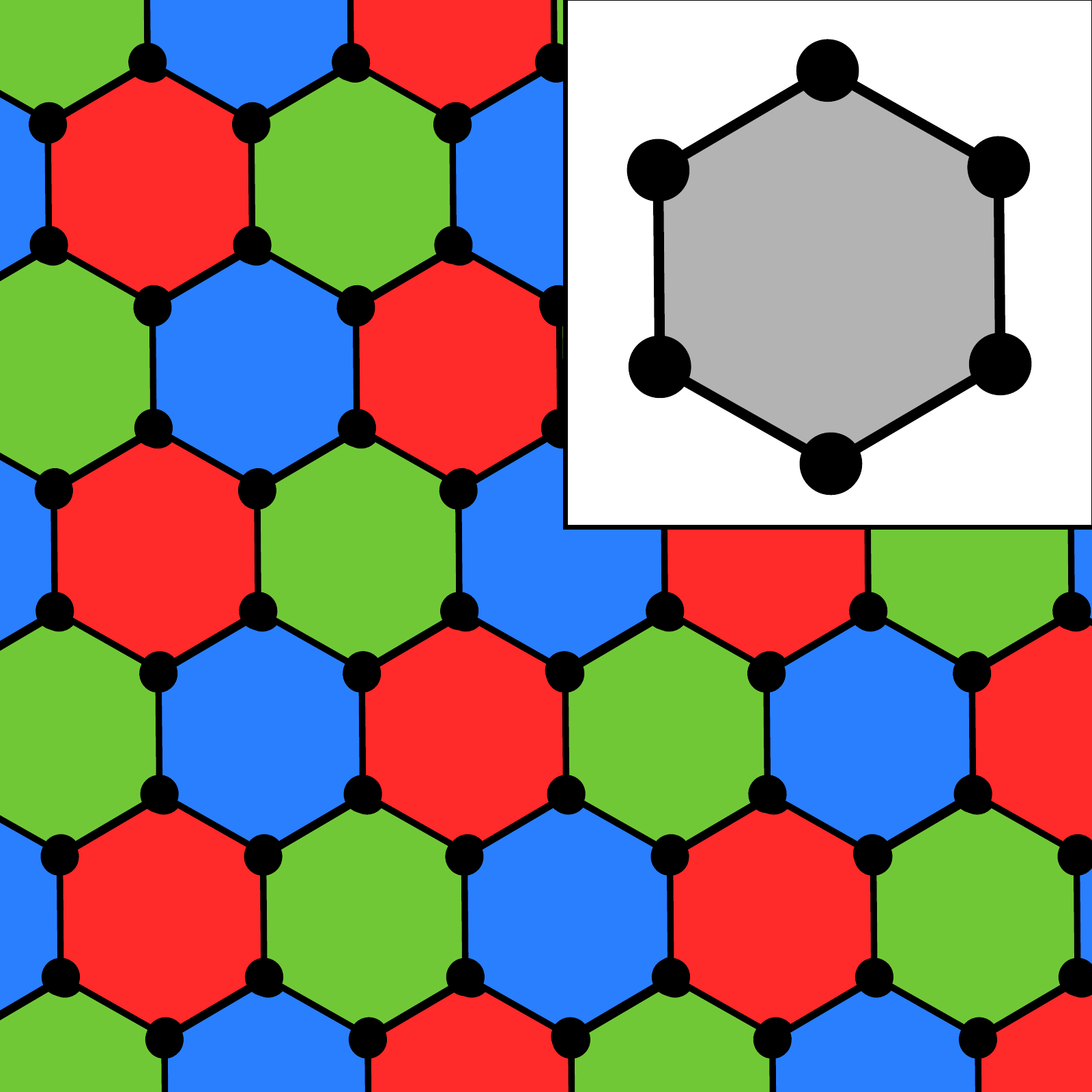}}\put(0.5877707,0.64663914){\color[rgb]{1,1,1}\makebox(0,0)[lt]{\lineheight{1.25}\smash{\begin{tabular}[t]{l}\X\end{tabular}}}}\put(0.533094,0.94847361){\color[rgb]{0,0,0}\makebox(0,0)[lt]{\lineheight{1.25}\smash{\begin{tabular}[t]{l}(a)\end{tabular}}}}\put(0.74130996,0.92013356){\color[rgb]{1,1,1}\makebox(0,0)[lt]{\lineheight{1.25}\smash{\begin{tabular}[t]{l}\X\end{tabular}}}}\put(0.90040365,0.64891064){\color[rgb]{1,1,1}\makebox(0,0)[lt]{\lineheight{1.25}\smash{\begin{tabular}[t]{l}\X\end{tabular}}}}\put(0.89845051,0.82941849){\color[rgb]{1,1,1}\makebox(0,0)[lt]{\lineheight{1.25}\smash{\begin{tabular}[t]{l}\X\end{tabular}}}}\put(0.58678201,0.82746536){\color[rgb]{1,1,1}\makebox(0,0)[lt]{\lineheight{1.25}\smash{\begin{tabular}[t]{l}\X\end{tabular}}}}\put(0.74584148,0.55787033){\color[rgb]{1,1,1}\makebox(0,0)[lt]{\lineheight{1.25}\smash{\begin{tabular}[t]{l}\X\end{tabular}}}}\put(0,0){\includegraphics[width=\unitlength,page=2]{honeycomb.pdf}}\put(0.22794977,0.24078585){\color[rgb]{1,1,1}\makebox(0,0)[lt]{\lineheight{1.25}\smash{\begin{tabular}[t]{l}\Z\end{tabular}}}}\put(0.30982503,0.29442002){\color[rgb]{1,1,1}\makebox(0,0)[lt]{\lineheight{1.25}\smash{\begin{tabular}[t]{l}\tsf{bz}\end{tabular}}}}\put(0.1223668,0.30520226){\color[rgb]{1,1,1}\makebox(0,0)[lt]{\lineheight{1.25}\smash{\begin{tabular}[t]{l}\tsf{gz}\end{tabular}}}}\put(0.21751911,0.1319036){\color[rgb]{1,1,1}\makebox(0,0)[lt]{\lineheight{1.25}\smash{\begin{tabular}[t]{l}\tsf{rz}\end{tabular}}}}\put(0.02066858,0.43350294){\color[rgb]{0,0,0}\makebox(0,0)[lt]{\lineheight{1.25}\smash{\begin{tabular}[t]{l}(b)\end{tabular}}}}\end{picture}\endgroup    \end{center}
  \caption{The $6$-$6$-$6$ honeycomb lattice formed by hexagonal plaquettes,
  each assigned a color: red ($r$), green ($g$), or blue ($b$). Links are
  colored by their terminating plaquettes; a link connecting two red plaquettes
  is also red, for example. A qubit (black circle) occupies each lattice site.
  {Inset (a).} The weight-$6$ operator $P_\X$ constructed from single-qubit
  Pauli $\X$ terms, used in the Hamiltonian. {Inset (b).} A single-qubit $\Z$
  excites $\tsf{bz}$, $\tsf{gz}$, and $\tsf{rz}$ anyons in neighboring
plaquettes.} 
  \label{fig:honeycomb}
\end{figure}

The color code topological order ($CC$) is defined on a $2$D three-colorable
lattice \cite{bombinTopological2006, bombinGauge2015, kesselringBoundaries2018,
yoshidaTopological2015}. In our work, we focus on the $6$-$6$-$6$ honeycomb
lattice formed by tessellating hexagonal plaquettes shown in
Fig.~\ref{fig:honeycomb}. We color each plaquette red ($r$), green ($g)$, or
blue ($b$) such that adjacent plaquettes are of different colors; a link that
connects two plaquettes of the same color by convention adopts that color too. A
qubit is placed on each lattice site and the Hamiltonian contains two species of
commuting operators, $P_\X$ and $P_\Z$, formed by the weight-$6$ Pauli $\X$ and
Pauli $\Z$ tensor products on the corners of each hexagonal plaquette, with the
identity elsewhere:

\begin{equation}
  H_{CC} = -\sum_{\text{plaquette }p} P_{\X, p} - \sum_{\text{plaquette }p}
  P_{\Z,p}.
  \label{eq:hamiltonian}
\end{equation}

In the ground state of this Hamiltonian, a single-qubit Pauli operator at a
lattice site excites three anyons in the three adjacent plaquettes, which we
label by $\tsf{c}\sigma$ indicating both the hosting plaquette color $c \in
\{r,g,b\}$ and the Pauli flavor $\sigma \in \{x,y,z\}$ that created it. A
single-qubit $\X$ therefore creates $\tsf{rx}$, $\tsf{gx}$, and $\tsf{bx}$
anyons, for example. An $\X$ applied to the two qubits of a red link creates a
pair of $\tsf{rx}$ anyons on the plaquettes at its endpoints.

These anyons and their behavior fundamentally define a TO. Self-statistics
capture the topological or spin-statistics of a particle, while the mutual
statistics encode the phase acquired upon braiding one anyon around another. In
particular, bosons have trivial self-exchange statistics, fermions a $-1$
statistic, while anyons can accumulate a different phase in $\text{U}(1)$.
Mutual-semions accumulate a phase of $-1$ when braided. The vacuum particle,
$\tsf 1$, is a boson and braids trivially with all other anyons. We have already
encountered the $9$ other, nontrivial, bosons of $CC$, formed by the combination
of $r,g,b$ colors and $x,y,z$ flavors; two of these bosons that share the same
color or flavor braid trivially, but otherwise are mutual-semions. There are
also $6$ fermions formed by pairs of mutual-semions; these are listed in
Appendix~\ref{app:fermions}. $CC$ is isomorphic to two copies of $TC$, written
as $CC \equiv TC \boxtimes TC$, with $\boxtimes$ denoting a product between two
TOs; any anyon in $CC$ can be written as a corresponding product of two anyons
from $TC$ \cite{bombinUniversal2012, kubicaUnfolding2015,
haahClassification2021}.\footnote{The toric code is characterized by four
species of anyons: the vacuum $\tsf{1}$, bosons $\tsf{e}$ and $\tsf{m}$, and
fermion $\tsf{f}$. There are $4^2$ anyon species in $CC$ and $4$ in $TC$.}

The anyon theory of a TO is also characterized by fusion rules (whereby two
anyons in proximity to each other are equivalent to a third)
\cite{simonTopological2023}. It is often convenient to represent the braiding
and fusion rules of the color code by the ``Mermin-Peres magic square''
(hereafter referred to as the magic square) \cite{merminSimple1990,
peresTwo1991}:
\begin{equation}
  \magicsf{rx}{ry}{rz}{gx}{gy}{gz}{bx}{by}{bz}
  \label{eq:magicsq}
\end{equation}
such that bosons in the same row or column braid trivially while those that are
not are mutual-semions. The fusion rules are such that two bosons in the same
row or column fuse to make the third, and two anyons of the same type fuse
(annihilate) to the vacuum. We write 
\begin{equation}
  \tsf{rx}\times\tsf{rz} = \tsf{ry},\quad
\tsf{rz}\times\tsf{gz}=\tsf{bz},\quad \tsf{gy}\times\tsf{gy} = \tsf{1},
\end{equation}
for example.

The automorphisms of a TO are maps between its anyon theories that preserve the
statistics and fusion rules. For $CC$, these form a $72$-element symmetry group
$\text{Aut}[CC]$ that is in correspondence to a subgroup of the permutations of
the magic square: $\text{Aut}[CC]$ can be decomposed\footnote{The $S_2$ subgroup
is not closed under conjugation, and hence we require the semidirect product.}
as $(S_3 \times S_3) \rtimes S_2$ such that we can write any automorphism as the
product of one of $6=3!$ row or color permutations (the symmetry group $S_3$),
$6$ column or flavor permutations ($S_3$), and $2$ color-flavor swapping
reflections about the diagonal of the magic square ($S_2$)
\cite{davydovaQuantum2024}. Appendix~\ref{app:theory} provides a more detailed
description of the relevant group-theoretic concepts of the automorphism group;
we summarize some key information here. We denote elements of $\text{Aut}[CC]$
using the cycle notation of the permutation group $S_6$, indicating the
transformation of the 6 anyon labels $r,g,b,$ and $x,y,z$. Since all bosons are
composed of one color and one flavor label, cycles must always be formed of
either disjoint color and flavor cycles, or alternating color and flavor.
$(rgx)(yz)$, for example, is not a member of $\text{Aut}[CC]$, but
$(rx)(gy)(bz)$ and $(rg)(xyz)$ are. We write the identity map as $\id$.
Composition of two elements of $\text{Aut}[CC]$ are written as
$\varphi_2\varphi_1$, evaluated as $(\varphi_2\varphi_1)(\tsf a) =
\varphi_2(\varphi_1(\tsf a))$ on some anyon $\tsf a$. The ``separation'' between
two automorphisms $\varphi_A, \varphi_B$ is quantified by the transition map
\begin{equation} 
  \tau_{BA} = \varphi_B\varphi_A^{-1}. 
\end{equation} 
It links $\varphi_B$ and $\varphi_A$ by $\varphi_B = \tau_{BA} \varphi_A$ and
satisfies $\tau_{AB} = \tau_{BA}^{-1}$. $\text{Aut}[CC]$ can be partitioned by
cycle type into $9$ conjugacy classes (sets that are linked by conjugation with
some element in $\text{Aut}[CC]$), which are subsets of the conjugacy classes of
$S_6$. For example, $(rg)(xyz)$ and $(rgb)(xy)$ both contain one
$2$-cycle\footnote{A $k$-cycle is a cycle with $k$ labels. It equivalently has
order $k$: if $\phi$ is a $k$-cycle then $k\geq1$ is minimal such that $\phi^k =
\id$. } and one $3$-cycle and thus both have cycle type $[2^13^1]$. We denote
their conjugacy class $\mathcal C\{(ccc)(\sigma\sigma)\}$. All $9$ conjugacy
classes are listed in Table~\ref{tbl:conjugacyclasses}. $\tau_{BA}$ and
$\tau_{AB}$ are in the same conjugacy class because permutations and their
inverses are always of the same cycle type.

\begin{table*}
  \caption{\label{tbl:conjugacyclasses}Conjugacy classes of the automorphism
  group, $\text{Aut}[CC]$ (see Appendix~\ref{app:theory}). Cycle type states the
number of $k$-cycles that form the automorphisms, with $[3^2]$ indicating two
$3$-cycles, for example. $\mathcal D^2$ is the square of the quantum dimension
of a twist, equal to the number of anyon species that localize at that twist
(see Section~\ref{sec:anyonlocalization}). IMS indicates the number of invariant
anyons that form mutual-semion pairs $\tsf a$ and $\tsf b$, such that
$\varphi(\tsf a) =\tsf a$, $\varphi(\tsf b) = \tsf b$, with $\tsf a$ and $\tsf
b$ having $-1$-mutual statistics, and trivial statistics with all other pairs.}
  \begin{ruledtabular}
    \begin{tabular}{ccccccc}
      Conjugacy Class & Cycle Type & Example & Parity on
      $S_3\times S_3$ & $\log_2\mathcal D^2$ & IMS &
      Number of Elements \\
      \hline
      $\mathcal C\{\id\}$ & $[1^6]$ & $\id$ & even & $0$ & $4$ & $1$ \\
      $\mathcal C\{(c\sigma)(c\sigma)(c\sigma)\}$ & $[2^3]$ & $(rx)(gy)(bz)$ &
      even & $1$ & $2$ & $6$ \\
      $\mathcal C\{(ccc)(\sigma\sigma\sigma)\}$ & $[3^2]$ & $(rgb)(xyz)$ & even
      & $2$ & $2$ & $4$ \\ 
      $\mathcal C\{(cc)(\sigma\sigma)\}$ & $[1^22^2]$ & $(rg)(xy)$ & even & $2$
      & $0$ & $9$
      \\ 
      $\mathcal C\{(c\sigma c\sigma c\sigma)\}$ & $[6^1]$ & $(rxgybz)$ & even &
      $3$ & $0$ & $12$ \\ 
      $\mathcal C\{(ccc)\}$ & $[1^33^1]$ & $(rgb)$ & even & $4$ & $0$ & $4$ \\ 
      $\mathcal C\{(cc)\}$ & $[1^42^1]$ & $(rg)$ & odd & $2$ & $0$ & $6$ \\ 
      $\mathcal C\{(c\sigma c\sigma)(c\sigma)\}$ & $[2^14^2]$ & $(rxgy)(bz)$ &
      odd & $3$ & $0$ & $18$ \\
      $\mathcal C\{(ccc)(\sigma\sigma)\}$ & $[1^12^13^1]$ & $(rgb)(xy)$ & odd &
      $4$ & $0$ & $12$ \\ 
      \hline 
      Total & & & & & & $72$
    \end{tabular}
  \end{ruledtabular}
\end{table*}

\subsection{\label{sec:anyonlocalization}Anyon Localization}
A useful characterization for anyons is localization at domain wall endpoints
(or ``twists'' \cite{bombinTopological2010a}): if anyon $\tsf c$ arises from
fusion as $\tsf c = \tsf b \times \varphi(\bar{\tsf b})$, where $\varphi$ is the
domain wall automorphism for an anyon encircling a twist anticlockwise, it is
said to localize at that twist by the process in Fig.~\ref{fig:localization}
\cite{kesselringBoundaries2018, barkeshliSymmetry2019}. At a
$\varphi_A$-$\varphi_B$ boundary between two domain walls, anyon localization of
$\tsf c$ occurs if (cf. Fig.~\ref{fig:localization2}) 
\begin{equation}\label{eq:anyonlocalization}
\tsf c = \tsf b \times \tau_{BA}(\bar{\tsf b}),\end{equation}
where $\tau_{BA} = \varphi_B\varphi_A^{-1}$ is their transition map.

\begin{figure}
  \vskip -0.5cm
  \begin{center}
    \def\svgwidth{\linewidth}
    \begingroup \makeatletter \providecommand\color[2][]{\errmessage{(Inkscape) Color is used for the text in Inkscape, but the package 'color.sty' is not loaded}\renewcommand\color[2][]{}}\providecommand\transparent[1]{\errmessage{(Inkscape) Transparency is used (non-zero) for the text in Inkscape, but the package 'transparent.sty' is not loaded}\renewcommand\transparent[1]{}}\providecommand\rotatebox[2]{#2}\newcommand*\fsize{\dimexpr\f@size pt\relax}\newcommand*\lineheight[1]{\fontsize{\fsize}{#1\fsize}\selectfont}\ifx\svgwidth\undefined \setlength{\unitlength}{600bp}\ifx\svgscale\undefined \relax \else \setlength{\unitlength}{\unitlength * \real{\svgscale}}\fi \else \setlength{\unitlength}{\svgwidth}\fi \global\let\svgwidth\undefined \global\let\svgscale\undefined \makeatother \begin{picture}(1,0.5)\lineheight{1}\setlength\tabcolsep{0pt}\put(0,0){\includegraphics[width=\unitlength,page=1]{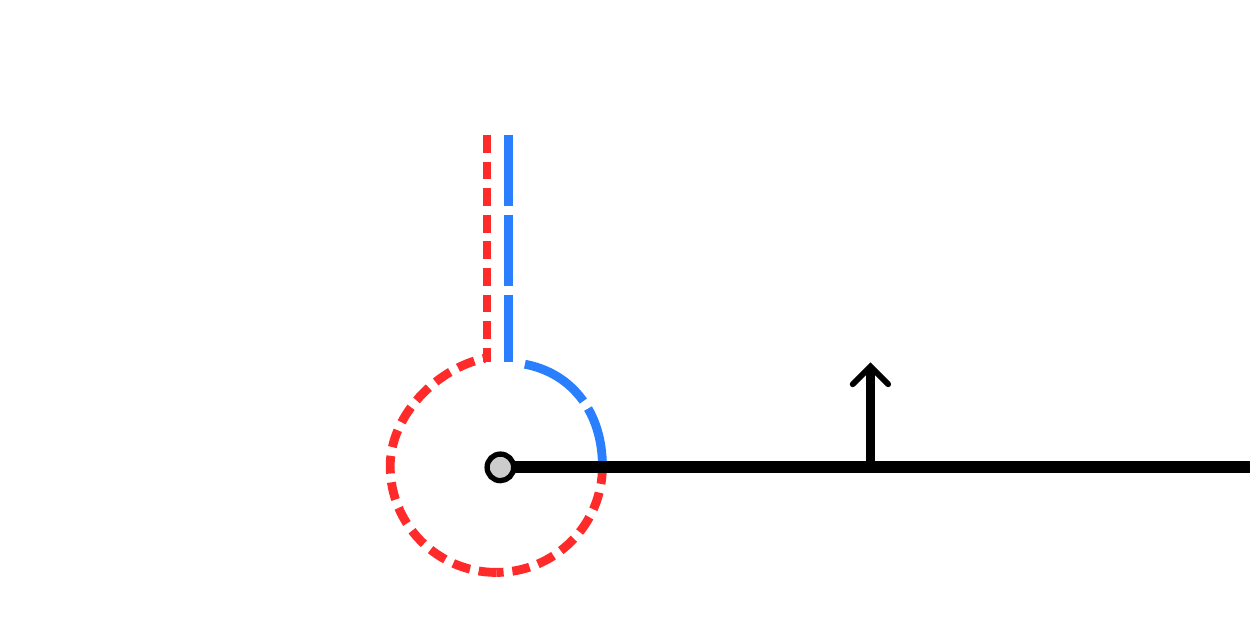}}\put(0.6843395,0.23052167){\color[rgb]{0,0,0}\makebox(0,0)[lt]{\lineheight{1.25}\smash{\begin{tabular}[t]{l}$\varphi$\end{tabular}}}}\put(0.38590616,0.40546535){\color[rgb]{0,0,0}\makebox(0,0)[lt]{\lineheight{1.25}\smash{\begin{tabular}[t]{l}$\tsf c$\end{tabular}}}}\put(0.47852341,0.19107359){\color[rgb]{0,0,0}\makebox(0,0)[lt]{\lineheight{1.25}\smash{\begin{tabular}[t]{l}$\varphi(\bar{\tsf b})$\end{tabular}}}}\put(0.30529486,0.19107359){\color[rgb]{0,0,0}\makebox(0,0)[lt]{\lineheight{1.25}\smash{\begin{tabular}[t]{l}$\tsf b$\end{tabular}}}}\put(0,0){\includegraphics[width=\unitlength,page=2]{twists.pdf}}\end{picture}\endgroup    \end{center}
  \vskip -0.5cm
  \caption{An anyon $\tsf c$ decomposes into two anyons $\tsf b$ and
  $\varphi(\bar{\tsf b})$ in the vicinity of the endpoint (twist) of a
  domain wall. Moving anticlockwise around the endpoint enacts the automorphism
  $\varphi$. Since $\tsf c$ can emanate or be absorbed at the twist, we say that
  $\tsf c$ localizes at the $\varphi$ twist. This picture applies to both
temporal and spatial domain walls, and hence a time arrow is not indicated.}
  \label{fig:localization}
\end{figure}

\begin{figure}
  \vskip -0.5cm
  \begin{center}
    \def\svgwidth{\linewidth}
    \begingroup \makeatletter \providecommand\color[2][]{\errmessage{(Inkscape) Color is used for the text in Inkscape, but the package 'color.sty' is not loaded}\renewcommand\color[2][]{}}\providecommand\transparent[1]{\errmessage{(Inkscape) Transparency is used (non-zero) for the text in Inkscape, but the package 'transparent.sty' is not loaded}\renewcommand\transparent[1]{}}\providecommand\rotatebox[2]{#2}\newcommand*\fsize{\dimexpr\f@size pt\relax}\newcommand*\lineheight[1]{\fontsize{\fsize}{#1\fsize}\selectfont}\ifx\svgwidth\undefined \setlength{\unitlength}{600bp}\ifx\svgscale\undefined \relax \else \setlength{\unitlength}{\unitlength * \real{\svgscale}}\fi \else \setlength{\unitlength}{\svgwidth}\fi \global\let\svgwidth\undefined \global\let\svgscale\undefined \makeatother \begin{picture}(1,0.5)\lineheight{1}\setlength\tabcolsep{0pt}\put(0,0){\includegraphics[width=\unitlength,page=1]{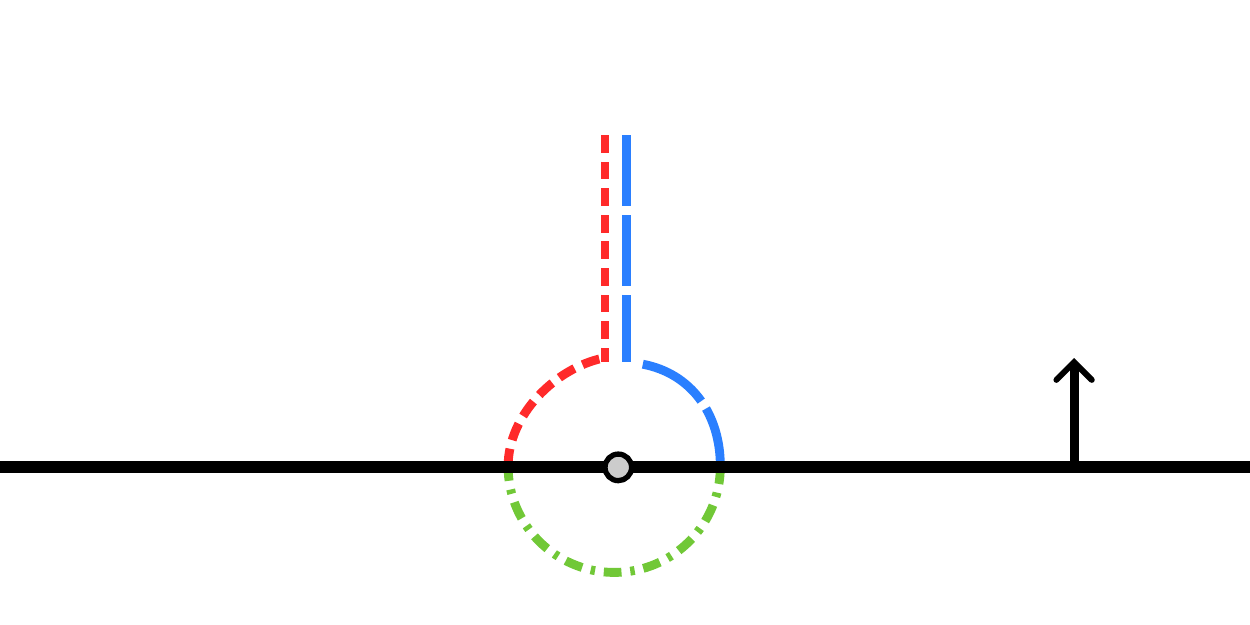}}\put(0.84556212,0.22366114){\color[rgb]{0,0,0}\makebox(0,0)[lt]{\lineheight{1.25}\smash{\begin{tabular}[t]{l}$\varphi_B$\end{tabular}}}}\put(0.48023854,0.40546535){\color[rgb]{0,0,0}\makebox(0,0)[lt]{\lineheight{1.25}\smash{\begin{tabular}[t]{l}$\tsf c$\end{tabular}}}}\put(0.57285579,0.19107359){\color[rgb]{0,0,0}\makebox(0,0)[lt]{\lineheight{1.25}\smash{\begin{tabular}[t]{l}$(\varphi_B\varphi_A^{-1})(\bar{\tsf b})$\end{tabular}}}}\put(0.39962724,0.19107359){\color[rgb]{0,0,0}\makebox(0,0)[lt]{\lineheight{1.25}\smash{\begin{tabular}[t]{l}$\tsf b$\end{tabular}}}}\put(0,0){\includegraphics[width=\unitlength,page=2]{twists2.pdf}}\put(0.13035119,0.22366114){\color[rgb]{0,0,0}\makebox(0,0)[lt]{\lineheight{1.25}\smash{\begin{tabular}[t]{l}$\varphi_A$\end{tabular}}}}\put(0.57457093,0.04700231){\color[rgb]{0,0,0}\makebox(0,0)[lt]{\lineheight{1.25}\smash{\begin{tabular}[t]{l}$\varphi_A^{-1}(\bar{\tsf b})$\end{tabular}}}}\put(0,0){\includegraphics[width=\unitlength,page=3]{twists2.pdf}}\end{picture}\endgroup    \end{center}
  \vskip -0.5cm
  \caption{A $\varphi_A$ and a $\varphi_B$ domain wall separated by a boundary
  (indicated by the point). A $\tsf c$ anyon that has fusion rule
  $\tsf c = \tsf b \times (\varphi_B\varphi_A^{-1})(\bar{\tsf b})$ localizes
  at this boundary.}
  \label{fig:localization2}
\end{figure}

The number of anyons that can localize at a twist equals $\mathcal D^2$, the
square of the quantum dimension $\mathcal D$ of the twist.\footnote{This
  relation applies only for Abelian anyon theories; in non-Abelian theories, one
  must also consider the quantum dimension of localizing anyons
\cite{bombinTopological2010, kesselringBoundaries2018, barkeshliSymmetry2019}.}
$\mathcal D$ tracks the increase in dimension of the Hilbert space when twists
are introduced. If $\tsf c$ localizes at $\tau_{BA}$, then for any automorphism
$\varphi$, $\varphi(\tsf c)$ localizes at $\varphi\tau_{BA}\varphi^{-1}$. Hence,
the quantum dimension of an automorphism's twist is a characteristic of its
conjugacy class. Table~\ref{tbl:conjugacyclasses} lists  $\log_2\mathcal D^2$
for each conjugacy class in $\text{Aut}[CC]$. 

In Section~\ref{sec:competing}, we examine how twists at the boundaries of
temporal domain walls affect the TO. In comparison to existing works such as
Refs.~\onlinecite{kesselringBoundaries2018, kesselringAnyon2024,
ellisonFloquet2023, davydovaFloquet2023}, which considered homogeneous or
spatial domain walls, we focus on heterogeneous temporal domain walls. We show
that anyon localization in this context fundamentally leads to logical
measurements of the degenerate codespace when the FET is operated as a QEC code.

\subsection{\label{sec:anyoncondensation}Anyon Condensation}
Anyon condensation is the process of relating two topologically-ordered systems
by identifying a set of condensed bosons in the ``parent'' TO with the vacuum
particle in the ``child'' TO \cite{baisCondensateinduced2009, huAnyon2022,
eliensDiagrammatics2014, eliensDiagrammatics2014, kesselringAnyon2024,
ellisonPauli2023, kongAnyon2014, baisTheory2009}; the process has similarities
to its namesake Bose-condensations in other physical systems \cite{huAnyon2022}.

Condensing a nontrivial boson, $\tsf c\sigma$, in a $CC$ parent phase realizes a
child theory equivalent to the toric code TO, denoted as $TC(\tsf c\sigma)$
\cite{kesselringAnyon2024}. By condensing different bosons at
different times, a system transitions between different child theories. Indeed,
the honeycomb Floquet code is equivalent to a dynamical transition through the
TOs of $TC(\tsf{rx})\rightarrow TC(\tsf{gy}) \rightarrow TC(\tsf{bz})\rightarrow
TC(\tsf{rx})\rightarrow\cdots$ \cite{hastingsDynamically2021}. We
diagrammatically represent this sequence using the magic square notation of
Eq.~\eqref{eq:magicsq}, 
\begin{equation*}
  \magic{$\bullet$}{}{}{}{}{}{}{}{} \rightarrow 
  \magic{}{}{}{}{$\bullet$}{}{}{}{} \rightarrow 
  \magic{}{}{}{}{}{}{}{}{$\bullet$} \rightarrow 
  \magic{$\bullet$}{}{}{}{}{}{}{}{} \rightarrow \cdots,
\end{equation*}
such that the $\bullet$ indicates the condensed boson.

Importantly, two child theories $TC(\tsf c\sigma_1)$ and $TC(\tsf c\sigma_2)$
are compatible if and only if $\tsf c\sigma_1$ and $\tsf c\sigma_2$ are
mutual-semions \cite{aasenMeasurement2023, burnellAnyon2018,
davydovaQuantum2024}. This ensures that two regions of $TC(\tsf c\sigma_1)$ and
$TC(\tsf c\sigma_2$) in spacetime share an invertible domain wall: if we start
with $TC(\tsf c\sigma_1)$ and condense the $\tsf c\sigma_2$ boson, the quantum
state of the TO is preserved. Any anyon in $TC(\tsf c\sigma_1)$ can move across
the domain wall and be mapped onto another anyon in $TC(c\sigma_2)$ without
modifying information about the particle (such as its statistics with other
anyons). This process of pairing up consecutive compatible child theories
$TC(\tsf c\sigma_1) \rightarrow TC(\tsf c\sigma_2)$ is called a ``reversible
transition''. In the magic square notation, reversible transitions require that
consecutive stages condense bosons that share neither the same row nor the same
column. 

It is also possible to construct a TO where anyon condensation results in a
child $CC$ theory. One such example is the ``dynamic automorphism'' (DA) color
code from Ref.~\onlinecite{davydovaQuantum2024}, using a parent $CC\boxtimes CC$
theory of two color code models; this can be envisaged as the honeycomb lattice
with two qubits at each lattice site (or equivalently, two layers of honeycomb
lattices) each hosting an independent $CC$ phase. The Hamiltonian is equivalent
to Eq.~\eqref{eq:hamiltonian}, except there are now two of each $P_\X$ and
$P_\Z$ that act only on the first or second layers. Condensing the anyons
$\tsf{rz}_1\tsf{rz}_2$, $\tsf{gz}_1\tsf{gz}_2$, and thus $\tsf{bz}_1\tsf{bz}_2$
(where the subscripts indicate the layer) produces a child theory
$\widetilde{CC}$ equivalent to the color code.\footnote{This is only one such
condensation choice; we could condense $x$-flavored bosons, for example, and
achieve an equivalent theory.} The anyons of this theory have representatives
\begin{equation}
  \magicc{$\tsf{rx}_1\tsf{rx}_2$}{$\tsf{ry}_1\tsf{rx}_2$}{$\tsf{rz}_1$}{$\tsf{gx}_1\tsf{gx}_2$}{$\tsf{gy}_1\tsf{gx}_2$}{$\tsf{gz}_1$}{$\tsf{bx}_1\tsf{bx}_2$}{$\tsf{by}_1\tsf{bx}_2$}{$\tsf{bz}_1$}
  \quad \sim \quad 
  \magicc{$\tsf{ry}_1\tsf{ry}_2$}{$\tsf{rx}_1\tsf{ry}_2$}{$\tsf{rz}_2$}{$\tsf{gy}_1\tsf{gy}_2$}{$\tsf{gx}_1\tsf{gy}_2$}{$\tsf{gz}_2$}{$\tsf{by}_1\tsf{by}_2$}{$\tsf{bx}_1\tsf{by}_2$}{$\tsf{bz}_2$}
\end{equation}
in correspondence to the $CC$ anyons of Eq.~\eqref{eq:magicsq}, with $\sim$
indicating equivalence up to fusion with the condensed bosons. For example,
$\tsf{rx}_1\tsf{rx}_2\times\tsf{rz}_1\tsf{rz}_2 = \tsf{ry}_1\tsf{ry}_2$. When
referring to anyons of $\widetilde{CC}$, if the subscripts are omitted then we
are denoting them by their equivalent sectors in $CC$; that is, $\tsf{rx}$
refers to $\tsf{rx}_1\tsf{rx}_2$ or $\tsf{ry}_1\tsf{ry}_2$.

Condensing an individual anyon from each $CC$ layer alternatively creates a
child theory of two decoupled toric codes, denoted $TC(\tsf{c}\sigma_1)\boxtimes
TC(\tsf{c}\sigma_2)$ with $\boxtimes$ again indicating the tensor product.
Reversible transitions now occur in two ways: (1) $TC\boxtimes TC
\leftrightarrow TC\boxtimes TC$ are reversible iff the individual $TC_1$ and
$TC_2$ transitions are reversible; and (2) $\widetilde{CC} \leftrightarrow
TC\boxtimes TC$ are reversible iff the two condensed anyons of the $TC\boxtimes
TC$ are of different colors and neither are $z$-flavored\footnote{This condition
arises due to the choice of $z$-flavored condensations leading to the child
theory $\widetilde{CC}$.} \cite{davydovaQuantum2024}. 

\begin{figure}
  \begin{center}
    \def\svgwidth{\columnwidth}
    \begingroup \makeatletter \providecommand\color[2][]{\errmessage{(Inkscape) Color is used for the text in Inkscape, but the package 'color.sty' is not loaded}\renewcommand\color[2][]{}}\providecommand\transparent[1]{\errmessage{(Inkscape) Transparency is used (non-zero) for the text in Inkscape, but the package 'transparent.sty' is not loaded}\renewcommand\transparent[1]{}}\providecommand\rotatebox[2]{#2}\newcommand*\fsize{\dimexpr\f@size pt\relax}\newcommand*\lineheight[1]{\fontsize{\fsize}{#1\fsize}\selectfont}\ifx\svgwidth\undefined \setlength{\unitlength}{900bp}\ifx\svgscale\undefined \relax \else \setlength{\unitlength}{\unitlength * \real{\svgscale}}\fi \else \setlength{\unitlength}{\svgwidth}\fi \global\let\svgwidth\undefined \global\let\svgscale\undefined \makeatother \begin{picture}(1,0.5)\lineheight{1}\setlength\tabcolsep{0pt}\put(0.4446177,0.23741795){\color[rgb]{0,0,0}\makebox(0,0)[lt]{\lineheight{1.25}\smash{\begin{tabular}[t]{l}$\widetilde{CC}$\end{tabular}}}}\put(0,0){\includegraphics[width=\unitlength,page=1]{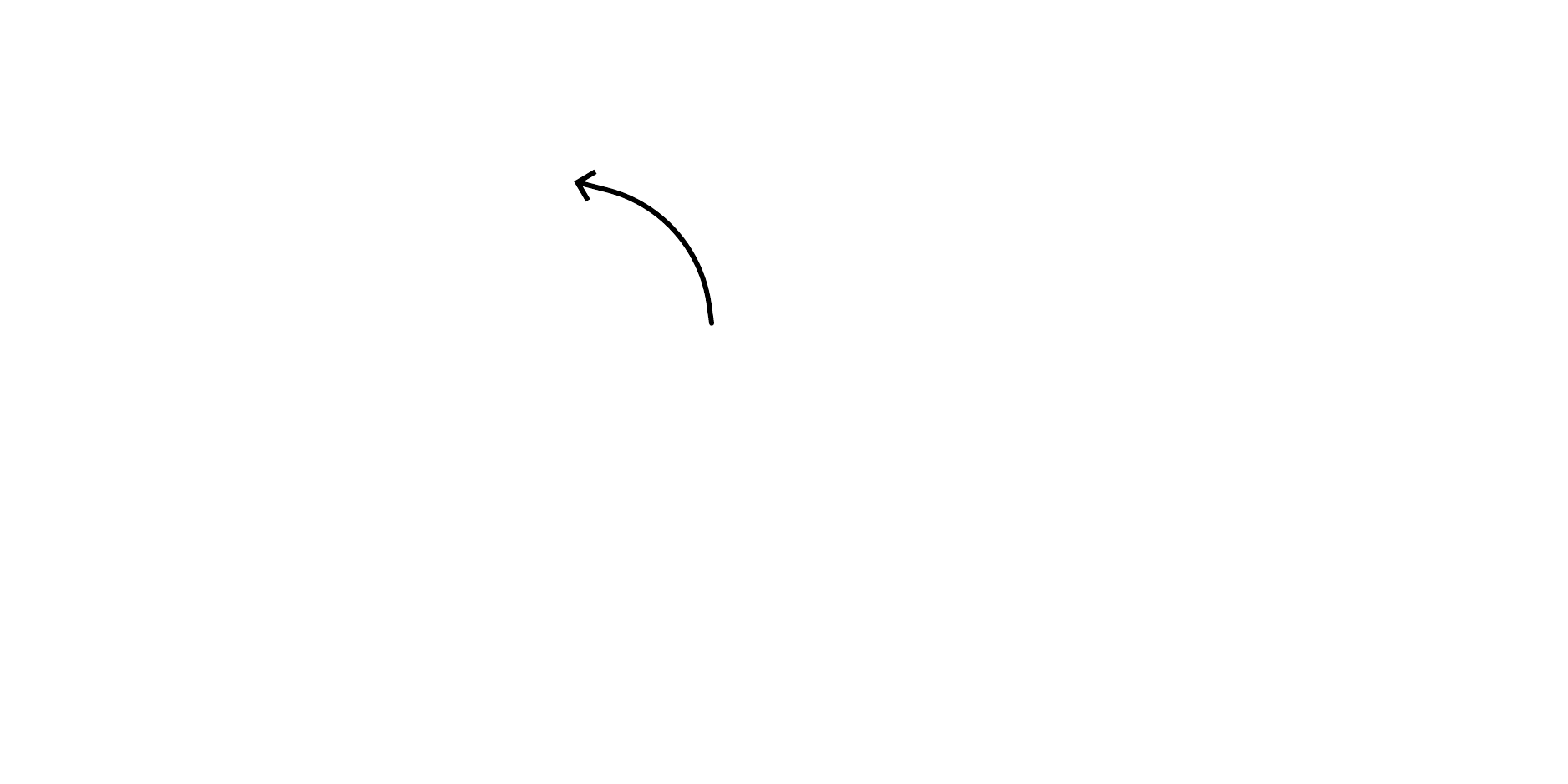}}\put(0.15564937,0.37561572){\color[rgb]{0,0,0}\makebox(0,0)[lt]{\lineheight{1.25}\smash{\begin{tabular}[t]{l}$(TC\boxtimes TC)_1$\end{tabular}}}}\put(0,0){\includegraphics[width=\unitlength,page=2]{automorphism.pdf}}\put(0.05097144,0.23741795){\color[rgb]{0,0,0}\makebox(0,0)[lt]{\lineheight{1.25}\smash{\begin{tabular}[t]{l}$\vdots$\end{tabular}}}}\put(0,0){\includegraphics[width=\unitlength,page=3]{automorphism.pdf}}\put(0.1539827,0.10952768){\color[rgb]{0,0,0}\makebox(0,0)[lt]{\lineheight{1.25}\smash{\begin{tabular}[t]{l}$(TC\boxtimes TC)_k$\end{tabular}}}}\put(0,0){\includegraphics[width=\unitlength,page=4]{automorphism.pdf}}\put(0.5949659,0.37561572){\color[rgb]{0,0,0}\makebox(0,0)[lt]{\lineheight{1.25}\smash{\begin{tabular}[t]{l}$(TC\boxtimes TC)_1$\end{tabular}}}}\put(0.87528111,0.19615522){\color[rgb]{0,0,0}\makebox(0,0)[lt]{\lineheight{1.25}\smash{\begin{tabular}[t]{l}$\vdots$\end{tabular}}}}\put(0,0){\includegraphics[width=\unitlength,page=5]{automorphism.pdf}}\put(0.57996589,0.10952768){\color[rgb]{0,0,0}\makebox(0,0)[lt]{\lineheight{1.25}\smash{\begin{tabular}[t]{l}$(TC\boxtimes TC)_{k'}$\end{tabular}}}}\put(0,0){\includegraphics[width=\unitlength,page=6]{automorphism.pdf}}\put(0.78494621,0.2991583){\color[rgb]{0,0,0}\makebox(0,0)[lt]{\lineheight{1.25}\smash{\begin{tabular}[t]{l}$(TC\boxtimes TC)_3$\end{tabular}}}}\put(0,0){\includegraphics[width=\unitlength,page=7]{automorphism.pdf}}\put(0.65320281,0.23969142){\color[rgb]{0,0,0}\makebox(0,0)[lt]{\lineheight{1.25}\smash{\begin{tabular}[t]{l}$\varphi_B$\end{tabular}}}}\put(0.2383102,0.23969142){\color[rgb]{0,0,0}\makebox(0,0)[lt]{\lineheight{1.25}\smash{\begin{tabular}[t]{l}$\varphi_A$\end{tabular}}}}\end{picture}\endgroup    \end{center}
  \vskip -1cm
  \caption{A color code topological order ($\widetilde{CC}$) can be mapped
    through a series of doubled toric code topological orders ($TC\boxtimes TC$)
    via anyon condensation. Upon returning to the original $\widetilde{CC}$, an
    automorphism $\varphi_A$ permutes the anyons of the model, creating a
    Floquet-enriched topological order (FET). Modifying the sequence, such as by
    skipping $(TC\boxtimes TC)_2$, may result in a different automorphism,
  $\varphi_B$.}
  \label{fig:automorphism}
\end{figure}

Automorphisms can be implemented by transitioning through a series of
condensations, cf. Fig.~\ref{fig:automorphism}. Davydova et al.
\cite{davydovaQuantum2024} showed how sequences starting and ending at
$\widetilde{CC}$ can transition the TO in such a way that any of the $72$
$\text{Aut}[CC]$ automorphisms can be enacted using at most $4$ intermediary
$TC\boxtimes TC$ condensations. For example, this is a
sequence that enacts an $(rgb)$ automorphism,
\begin{eqnarray}
    \widetilde{CC} \rightarrow
    \magic{1}{}{}{}{}{}{2}{}{} \rightarrow 
    \magic{}{}{}{}{}{2}{}{}{1} \rightarrow
    \magic{}{2}{}{}{1}{}{}{}{} \rightarrow
    \widetilde{CC}
    \label{eq:rgb}
\end{eqnarray}
where the $1$, $2$ labels indicate the condensed boson in the two $CC$ layers,
using the magic square notation from Eq.~\eqref{eq:magicsq}. That is, this
sequence of anyon condensation maps between the TOs of $\widetilde{CC}
\rightarrow TC(\tsf{rx}_1)\boxtimes TC(\tsf{bx}_2) \rightarrow$
$TC(\tsf{bz}_1)\boxtimes TC(\tsf{gz}_2) \rightarrow TC(\tsf{gy}_1)\boxtimes
TC(\tsf{ry}_2) \rightarrow \widetilde{CC}$. 
Appendix~\ref{app:automorphisms} explains how to compute the automorphism from
any given sequence, or construct a sequence to realize any given automorphism.
By repeatedly cycling through a sequence of condensates such as
Eq.~\eqref{eq:rgb} similarly to a driven quantum system, we thus create an
evolving phase that exhibits time-periodic, Floquet-enriched topological order
(FET) \cite{potterDynamically2017}. Anyons present will periodically have their
labels permuted. Multiple measurement sequences can realize the same
automorphism each Floquet period, and therefore multiple different systems can
exhibit the same FET.

\subsection{Dynamical Codes}
We may also interpret these FETs as dynamical QEC codes capable of encoding and
storing quantum information. The first dynamical code was the honeycomb Floquet
code from Ref.~\onlinecite{hastingsDynamically2021}. Other dynamical codes have
since been proposed, such as the CSS honeycomb code \cite{kesselringAnyon2024,
davydovaFloquet2023}, the automorphism code \cite{aasenAdiabatic2022}, the
dynamic automorphism color code \cite{davydovaQuantum2024}, the x$+$y Floquet
code \cite{bauerFloquet2024}, the XYZ ruby code \cite{delafuenteXYZ2024},
$(3+1)$D Floquet codes \cite{duaEngineering2024}, or the hyperbolic Floquet code
\cite{higgottConstructions2023}. Theoretical studies of dynamical codes include
perspectives such as subsystem codes \cite{davydovaFloquet2023,
bombinTopological2010}, quantum cellular automata \cite{aasenMeasurement2023},
twist-defect networks \cite{sullivanFloquet2023}, adiabatic paths of
Hamiltonians \cite{aasenAdiabatic2022}, or fixed-point path integrals
\cite{bauerTopological2024, bauerLowoverhead2024,
townsend-teagueFloquetifying2023, bombinUnifying2024}, and aspects of Floquet
code phenomenology were also linked to symmetry topological field theory
\cite{motamarriSymTFT2024}. 

To describe such codes, we use here the stabilizer formalism
\cite{shorFaulttolerant1997, gottesmanStabilizer1997, gottesmanHeisenberg1998,
poulinStabilizer2005}: taking the $P_\X$ and $P_\Z$ plaquette Hamiltonian terms
from Eq.~\eqref{eq:hamiltonian}, we promote them to generators of an Abelian
``stabilizer'' group $\mathcal S$.\footnote{A stabilizer group is specifically
any subgroup of the $n$-qubit Pauli group that does not contain $-1$ and where
all the elements commute. It ``stabilizes'' a code in the sense that any element
of the stabilizer group acts trivially on the logical subspace.} The
simultaneous $+1$-eigenspace of $\mathcal S$ defines the codespace $C$ such that
$\forall S \in \mathcal S$ and $\ket\psi \in C$ we have $S\ket\psi = \ket\psi$.
$C$ coincides with the ground space of our Hamiltonian. Excited states
$\ket{\psi'}$, for which $S\ket{\psi'} = -\ket{\psi'}$ for some $S \in \mathcal
S$, are not in the codespace; the ``excited'' stabilizers indicate that the
system has suffered an error. These excited plaquettes are equivalent to the
locations of anyon excitations in the Hamiltonian picture. Logical operators map
between states within the degenerate codespace $C$; these are denoted as
$\bar\X, \bar\Z$ and act with the same algebra on the (logical) Hilbert space of
$C$ as $\X$, $\Z$ act on single qubits. The centralizer of $\mathcal S$ is the
group $C(\mathcal S)$ of Pauli operators that commute with every stabilizer; the
nontrivial logicals are those in $C(\mathcal S)$ but not in (a phase times)
$\mathcal S$ itself, forming the set $\mathcal L = C(\mathcal S) \setminus
\langle i\openone,\mathcal S\rangle$.

To perform anyon condensation, we use projective measurements of the hopping
operators for the condensed boson; these are (typically) weight-$2$ Pauli
operators that correspond to moving an excitation through the lattice
\cite{kesselringAnyon2024}. For a $\tsf{rx}$ in the color code, for example, it
is the $2$-qubit $\X$ operator on the ends of a red link. Measuring these
throughout a system in the codespace of $CC$ causes a $CC \rightarrow
TC(\tsf{rx})$ transition \cite{kesselringAnyon2024}. In doing so, the stabilizer
group updates as measured operators are added and anticommuting operators
are combined or removed; since this group is constantly changing in a dynamical
code, we refer to the current state as the instantaneous stabilizer group (ISG)
\cite{hastingsDynamically2021}. The ISG updates during each condensation stage.

In our work, we focus on the $CC\boxtimes CC$ model of the DA color code. In
this context, the external tensor product denotes that the stabilizer group of
$CC \boxtimes CC$ can be factored (up to a unitary transformation) into two
independent copies of $CC$ stabilizer groups \cite{davydovaQuantum2024}. The
initial stabilizer group is thus comprised of $P_\X$ and $P_\Z$ plaquettes on
both layers of honeycomb lattices. Forming the child theory $\widetilde{CC}$
requires projective measurements of $\Z_1\Z_2$ on each lattice site, where the
subscripts indicate the two layers. $TC\boxtimes TC$ child theories are
analogous to the previous discussions, with hopping operators measured on each
layer separately for the respective condensed anyons.

\begin{figure}
  \begin{center}
    \def\svgwidth{0.75\columnwidth}
    \begingroup \makeatletter \providecommand\color[2][]{\errmessage{(Inkscape) Color is used for the text in Inkscape, but the package 'color.sty' is not loaded}\renewcommand\color[2][]{}}\providecommand\transparent[1]{\errmessage{(Inkscape) Transparency is used (non-zero) for the text in Inkscape, but the package 'transparent.sty' is not loaded}\renewcommand\transparent[1]{}}\providecommand\rotatebox[2]{#2}\newcommand*\fsize{\dimexpr\f@size pt\relax}\newcommand*\lineheight[1]{\fontsize{\fsize}{#1\fsize}\selectfont}\ifx\svgwidth\undefined \setlength{\unitlength}{767.99990845bp}\ifx\svgscale\undefined \relax \else \setlength{\unitlength}{\unitlength * \real{\svgscale}}\fi \else \setlength{\unitlength}{\svgwidth}\fi \global\let\svgwidth\undefined \global\let\svgscale\undefined \makeatother \begin{picture}(1,1)\lineheight{1}\setlength\tabcolsep{0pt}\put(0,0){\includegraphics[width=\unitlength,page=1]{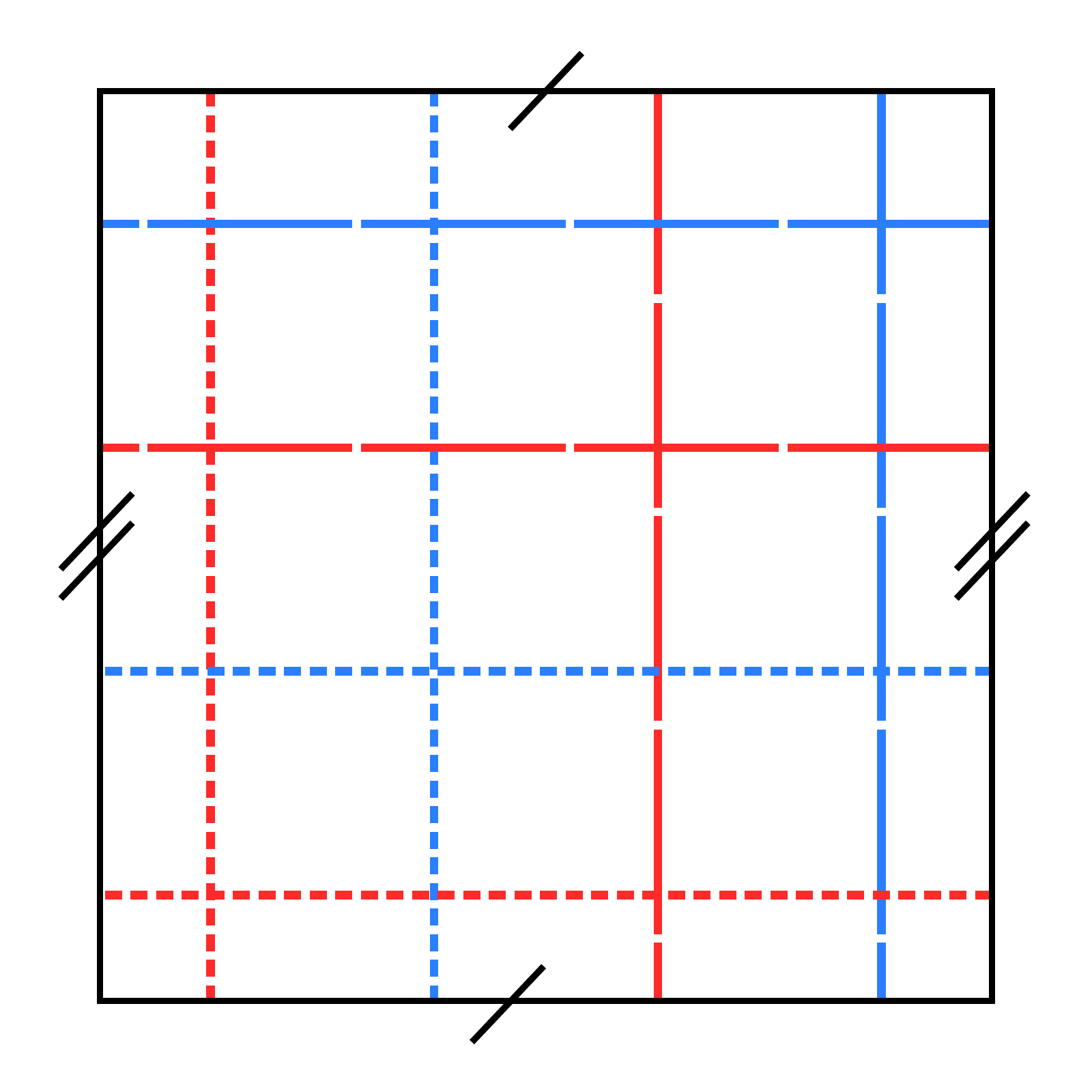}}\put(0.17403576,0.92923714){\color[rgb]{0,0,0}\makebox(0,0)[lt]{\lineheight{1.25}\smash{\begin{tabular}[t]{l}$\bar\X_1$\end{tabular}}}}\put(0.02674509,0.78119078){\color[rgb]{0,0,0}\makebox(0,0)[lt]{\lineheight{1.25}\smash{\begin{tabular}[t]{l}$\bar\Z_1$\end{tabular}}}}\put(0.02674509,0.58239467){\color[rgb]{0,0,0}\makebox(0,0)[lt]{\lineheight{1.25}\smash{\begin{tabular}[t]{l}$\bar\Z_2$\end{tabular}}}}\put(0.02910783,0.37754828){\color[rgb]{0,0,0}\makebox(0,0)[lt]{\lineheight{1.25}\smash{\begin{tabular}[t]{l}$\bar\X_3$\end{tabular}}}}\put(0.02674509,0.17270196){\color[rgb]{0,0,0}\makebox(0,0)[lt]{\lineheight{1.25}\smash{\begin{tabular}[t]{l}$\bar\X_4$\end{tabular}}}}\put(0.37988529,0.92923714){\color[rgb]{0,0,0}\makebox(0,0)[lt]{\lineheight{1.25}\smash{\begin{tabular}[t]{l}$\bar\X_2$\end{tabular}}}}\put(0.58573482,0.92923714){\color[rgb]{0,0,0}\makebox(0,0)[lt]{\lineheight{1.25}\smash{\begin{tabular}[t]{l}$\bar\Z_3$\end{tabular}}}}\put(0.78685886,0.92923714){\color[rgb]{0,0,0}\makebox(0,0)[lt]{\lineheight{1.25}\smash{\begin{tabular}[t]{l}$\bar\Z_4$\end{tabular}}}}\put(0.16817639,0.04349529){\color[rgb]{0,0,0}\makebox(0,0)[lt]{\lineheight{1.25}\smash{\begin{tabular}[t]{l}$\tsf{rx}$\end{tabular}}}}\put(0.37402591,0.04349529){\color[rgb]{0,0,0}\makebox(0,0)[lt]{\lineheight{1.25}\smash{\begin{tabular}[t]{l}$\tsf{bx}$\end{tabular}}}}\put(0.57987541,0.04349529){\color[rgb]{0,0,0}\makebox(0,0)[lt]{\lineheight{1.25}\smash{\begin{tabular}[t]{l}$\tsf{rz}$\end{tabular}}}}\put(0.78099944,0.04349529){\color[rgb]{0,0,0}\makebox(0,0)[lt]{\lineheight{1.25}\smash{\begin{tabular}[t]{l}$\tsf{bz}$\end{tabular}}}}\put(0.92330657,0.7879313){\color[rgb]{0,0,0}\makebox(0,0)[lt]{\lineheight{1.25}\smash{\begin{tabular}[t]{l}$\tsf{bz}$\end{tabular}}}}\put(0.92330657,0.58913519){\color[rgb]{0,0,0}\makebox(0,0)[lt]{\lineheight{1.25}\smash{\begin{tabular}[t]{l}$\tsf{rz}$\end{tabular}}}}\put(0.92566929,0.38428886){\color[rgb]{0,0,0}\makebox(0,0)[lt]{\lineheight{1.25}\smash{\begin{tabular}[t]{l}$\tsf{bx}$\end{tabular}}}}\put(0.92330657,0.17944255){\color[rgb]{0,0,0}\makebox(0,0)[lt]{\lineheight{1.25}\smash{\begin{tabular}[t]{l}$\tsf{rx}$\end{tabular}}}}\end{picture}\endgroup    \end{center}
  \vspace{-12pt}
  \caption{Diagram of the logical operators of the ${CC}$ model on a $2$-torus.
  Shown are the $4$ pairs of anticommuting $\bar\X$ (short-dashed) and $\bar\Z$
  (long-dashed) operators from Eq.~\eqref{eq:logical1}-\eqref{eq:logical2},
  supported on noncontractible cycles of the $2$-torus. The honeycomb lattice
  structure is ignored for simplicity; the string operators follow the red and
  blue links of Fig.~\ref{fig:honeycomb}.}
  \label{fig:logicals}
\end{figure}

Our implementation of the DA color code uses a $2$-torus; future work may find
it fruitful to consider other topologies or open boundary conditions. We employ
a logical algebra with $4$ pairs of anticommuting logical operators constructed
out of $\tsf{rx}$, $\tsf{bz}$, $\tsf{bx}$, and $\tsf{rz}$ effective anyon
strings. Two operators are equivalent (act equivalently on the codestates) if
they are related modulo multiplication with operators in the stabilizer group
\cite{gottesmanStabilizer1997}. This means that multiple representatives of each
logical operator exist; on the $2$-torus these are supported on homologous
noncontractible cycles around the periodic boundaries.\footnote{Two cycles are
  homologous if one can be smoothly deformed into the other without breaking the
chain. A cycle is noncontractible if it is not homologous to a loop with zero
area, i.e. a point.} Let $\bar{\Os}[\tsf a]_h$ and $\bar\Os[\tsf a]_v$ represent
the equivalence class of logical operators, forming the quotient group
$C(\mathcal S)/\mathcal S$. The subscript indicates that the string wraps around
the horizontal or vertical direction respectively (using the orientation of
Fig.~\ref{fig:honeycomb}). An italicized $\bar O[\tsf a]_h$ indicates a
particular representative of the equivalence class. We use the logical algebra 
\begin{eqnarray}
  \bar \X_1 &=& \bar\Os[\tsf{rx}]_v, \quad \bar \Z_1 = \bar\Os[\tsf{bz}]_h,
  \label{eq:logical1} \\
  \bar \X_2 &=& \bar\Os[\tsf{bx}]_v, \quad \bar \Z_2 = \bar\Os[\tsf{rz}]_h, \\
  \bar \X_3 &=& \bar\Os[\tsf{bx}]_h, \quad \bar \Z_3 = \bar\Os[\tsf{rz}]_v, \\
  \bar \X_4 &=& \bar\Os[\tsf{rx}]_h, \quad \bar \Z_4 = \bar\Os[\tsf{bz}]_v.
  \label{eq:logical2}
\end{eqnarray}
Figure~\ref{fig:logicals} shows sketches of these operators. An automorphism
that permutes the anyons now also permutes the logical operators. $(rb)$, for
example, swaps $\bar\X_1$ and $\bar\X_2$, and is indeed equivalent to a
$\texttt{SWAP}_{12}\texttt{SWAP}_{34}$ gate on the four qubits. The $72$
automorphisms furnish a subgroup of the $4$-qubit Clifford group
\cite{davydovaQuantum2024}. This choice of logical algebra is not unique.
Indeed, we are also not restricted to just using bosons; any choice of 4 anyons
that form two pairs of mutual-semions but otherwise have trivial mutual
statistics (hence commuting string operators) can form a valid logical algebra.
Section~\ref{sec:logically}, for example, describes an alternative logical
algebra using fermions.

\section{\label{sec:competing}Competing Automorphisms}

In this section, we introduce systems of disordered Floquet codes where the
temporal domain walls are spatially heterogeneous: acting on disjoint subregions
of the manifold with different automorphisms. As the realization of such
subregions is random, we can interpret this as two (or more) automorphisms
``competing'' for occupancy of the system. We first describe the effect of one
realization of this competition on Abelian-anyon TOs, before extending our
results to measurement-induced FETs where different realizations of the
automorphisms occur at every discrete timestep. We use $CC$ as an example TO, in
preparation for Section~\ref{sec:disordermodel} where we examine a disordered DA
color code from this perspective.

Consider a $2$-torus (although our results are readily generalizable to other
topologies with different genera) and assume at time $t=0$ the system is in the
ground-state of some local Hamiltonian $H=-\sum_j S_j$ with local generators
$S_j$ of a stabilizer group $\mathcal S$. The manifold is randomly partitioned,
and each contiguous region labeled $A$ or $B$. At time $0<t_\varphi<1$,
automorphisms $\varphi_A$ and $\varphi_B$ ($\varphi_A \neq \varphi_B$, and one
of which may be trivial) are simultaneously applied to their respective
subregions. We may equivalently view this as $\varphi_A$ applied everywhere and
$B$-subregions additionally enacting $\tau_{BA}$. 

Global automorphisms implemented via homogeneous temporal domain walls map
ground states of $H$ to ground states. This is not true of heterogeneous domain
walls, however, as the automorphism boundaries may bisect the local stabilizer
terms. In this case, the resulting system may no longer be a ground state (or
eigenstate) of $H$. If the system is then projected back into an eigenstate of
$H$, this may result in a nontrivial (potentially non-unitary) mapping. This
projection is precisely that expected when performing QEC on a stabilizer code,
as error-detection measurements ensure that a system returns to stabilizer
eigenstates. \footnote{In this analysis, for ease of presentation, we imagine
automorphisms and stabilizer measurements occurring at different times. In
dynamical code implementations, such as the DA color code, these two processes
typically occur concurrently.}

\begin{figure}
  \begin{center}
    \def\svgwidth{\linewidth}
    \begingroup \makeatletter \providecommand\color[2][]{\errmessage{(Inkscape) Color is used for the text in Inkscape, but the package 'color.sty' is not loaded}\renewcommand\color[2][]{}}\providecommand\transparent[1]{\errmessage{(Inkscape) Transparency is used (non-zero) for the text in Inkscape, but the package 'transparent.sty' is not loaded}\renewcommand\transparent[1]{}}\providecommand\rotatebox[2]{#2}\newcommand*\fsize{\dimexpr\f@size pt\relax}\newcommand*\lineheight[1]{\fontsize{\fsize}{#1\fsize}\selectfont}\ifx\svgwidth\undefined \setlength{\unitlength}{767.99990845bp}\ifx\svgscale\undefined \relax \else \setlength{\unitlength}{\unitlength * \real{\svgscale}}\fi \else \setlength{\unitlength}{\svgwidth}\fi \global\let\svgwidth\undefined \global\let\svgscale\undefined \makeatother \begin{picture}(1,0.70312508)\lineheight{1}\setlength\tabcolsep{0pt}\put(0,0){\includegraphics[width=\unitlength,page=1]{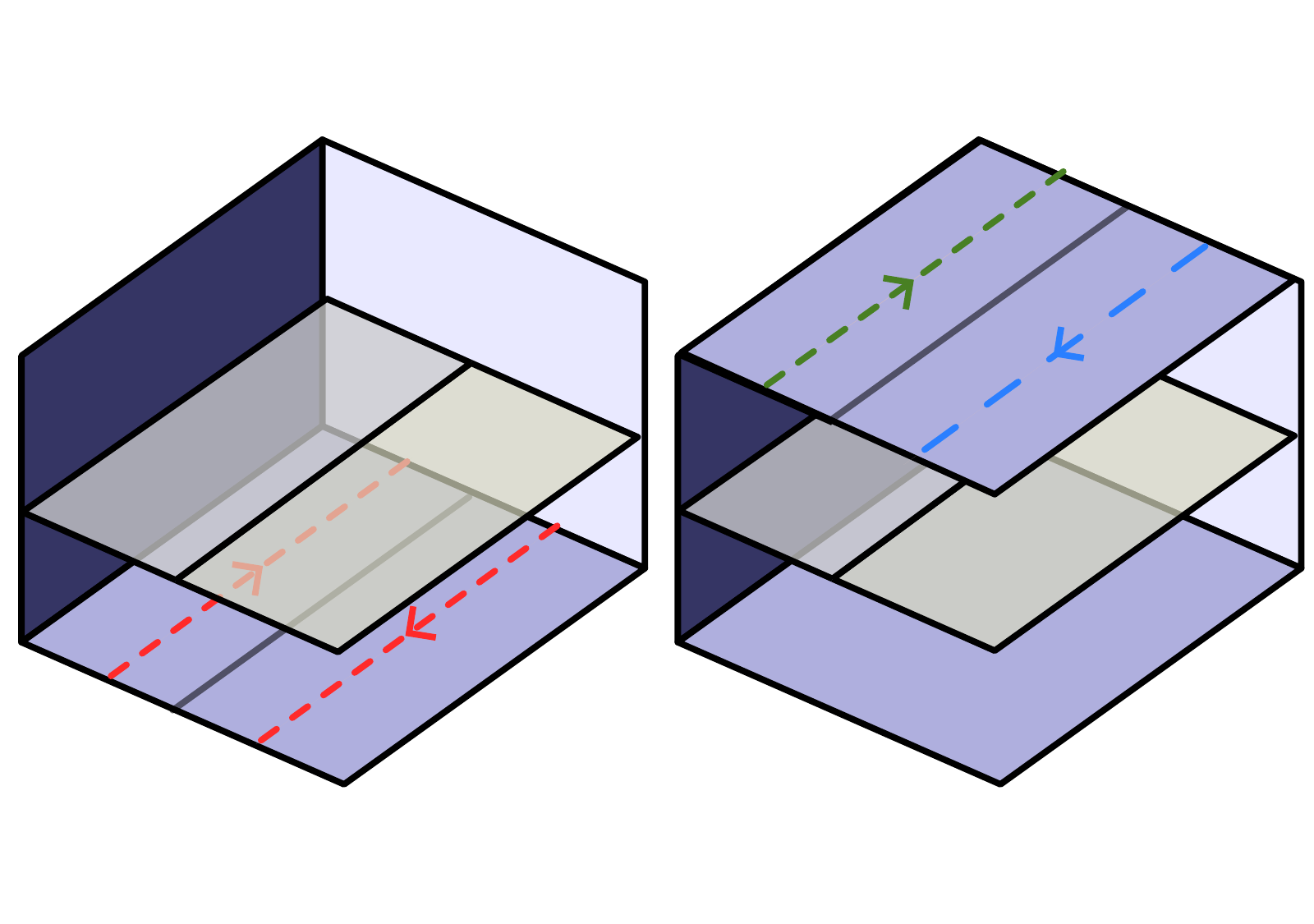}}\put(0.17619525,0.03177889){\color[rgb]{0,0,0}\makebox(0,0)[lt]{\lineheight{1.25}\smash{\begin{tabular}[t]{l}$t=0$\end{tabular}}}}\put(0.36042424,0.35915404){\color[rgb]{0,0,0}\makebox(0,0)[lt]{\lineheight{1.25}\smash{\begin{tabular}[t]{l}$\varphi_A$\end{tabular}}}}\put(0,0){\includegraphics[width=\unitlength,page=2]{domain_walls5.pdf}}\put(0.50919986,0.65552774){\color[rgb]{0,0,0}\makebox(0,0)[lt]{\lineheight{1.25}\smash{\begin{tabular}[t]{l}$t$\end{tabular}}}}\put(0.72430062,0.03177889){\color[rgb]{0,0,0}\makebox(0,0)[lt]{\lineheight{1.25}\smash{\begin{tabular}[t]{l}$t=1$\end{tabular}}}}\put(0.25501937,0.40627622){\color[rgb]{0,0,0}\makebox(0,0)[lt]{\lineheight{1.25}\smash{\begin{tabular}[t]{l}$\varphi_B$\end{tabular}}}}\put(0.04296951,0.08510138){\color[rgb]{1,0.16470588,0.16470588}\makebox(0,0)[lt]{\lineheight{1.25}\smash{\begin{tabular}[t]{l}$W[\varphi_A^{-1}(\bar{\texttt a})]$\end{tabular}}}}\put(-0.00291259,0.13470366){\color[rgb]{1,0.16470588,0.16470588}\makebox(0,0)[lt]{\lineheight{1.25}\smash{\begin{tabular}[t]{l}$W[\varphi_B^{-1}(\texttt a)]$\end{tabular}}}}\put(0.90480942,0.53648225){\color[rgb]{0.16470588,0.49803922,1}\makebox(0,0)[lt]{\lineheight{1.25}\smash{\begin{tabular}[t]{l}$W[\bar{\texttt a}]$\end{tabular}}}}\put(0.81304517,0.58484448){\color[rgb]{0.28235294,0.50196078,0.1372549}\makebox(0,0)[lt]{\lineheight{1.25}\smash{\begin{tabular}[t]{l}$W[\texttt{a}]$\end{tabular}}}}\end{picture}\endgroup    \end{center}
  \vskip -0.5cm
  \caption{Spacetime illustration of a section of an Abelian TO, with time
    flowing upwards, where $\varphi_A$ and $\varphi_B$ domain walls act
    concurrently on disjoint regions of the lattice between times $t=0$ and $t=1$.
    The operator $O_0 = W[\varphi_B^{-1}(\tsf a)]W[\varphi_A^{-1}(\bar{\tsf a})]$
    with support straddling a (closed) segment of the $\varphi_A$-$\varphi_B$
    boundary evolves to the product $O_1 = W[\tsf a]W[\bar{\tsf a}]$ at $t=1$.
    This is a stabilizer.}
  \label{fig:domain_walls5} 
\end{figure}

They key consequence of the competing automorphisms is that logical operators
and stabilizers interchange, which in QEC codes can lead to measurement of
encoded information. To see this, consider a closed segment of the boundary
between $\varphi_A$ and $\varphi_B$ domain walls. Let $W[\tsf a]$ denote a
Wilson loop operator for some anyon $\tsf a$ along this segment. Take the
operator 
\begin{equation}
    O_0 = W[\varphi_B^{-1}(\tsf a)] W[\varphi_A^{-1}(\bar{\tsf a})],
\end{equation}
cf. Fig.~\ref{fig:domain_walls5}. If the $A$-$B$ boundary is noncontractible and
\begin{equation}
    \varphi_B^{-1}(\tsf a) \times \varphi_A^{-1}(\bar{\tsf a}) \neq \tsf 1,
\end{equation}
or equivalently 
\begin{equation}
  \tsf d \equiv \tsf a \times \tau_{BA}(\bar{\tsf a}) \neq \tsf 1,
    \label{eq:d_anyon}
\end{equation}
then $O_0$ is a logical operator. At $t=t_\varphi$, this operator will evolve to
\begin{equation}
    O_1 = W[\tsf a]W[\bar{\tsf a}],
\end{equation}
which is a stabilizer in $\mathcal S$ as $\tsf a \times \bar{\tsf a} = \tsf 1$.
This means that if the system was in a superposition of eigenstates of the
logical operator $O_0$, it is now by $t=1$ in the same superposition of
eigenstates of $O_1$, a stabilizer. Moreover, if this were a QEC code, this
$O_1$ will be measured during syndrome detection, and eventually be corrected
into a $+1$-eigenstate by QEC. Information encoded by the logical operator $O_0$
will thus be measured out.

\begin{figure}
  \begin{center}
    \def\svgwidth{\linewidth}
    \begingroup \makeatletter \providecommand\color[2][]{\errmessage{(Inkscape) Color is used for the text in Inkscape, but the package 'color.sty' is not loaded}\renewcommand\color[2][]{}}\providecommand\transparent[1]{\errmessage{(Inkscape) Transparency is used (non-zero) for the text in Inkscape, but the package 'transparent.sty' is not loaded}\renewcommand\transparent[1]{}}\providecommand\rotatebox[2]{#2}\newcommand*\fsize{\dimexpr\f@size pt\relax}\newcommand*\lineheight[1]{\fontsize{\fsize}{#1\fsize}\selectfont}\ifx\svgwidth\undefined \setlength{\unitlength}{767.99990845bp}\ifx\svgscale\undefined \relax \else \setlength{\unitlength}{\unitlength * \real{\svgscale}}\fi \else \setlength{\unitlength}{\svgwidth}\fi \global\let\svgwidth\undefined \global\let\svgscale\undefined \makeatother \begin{picture}(1,0.70312508)\lineheight{1}\setlength\tabcolsep{0pt}\put(0,0){\includegraphics[width=\unitlength,page=1]{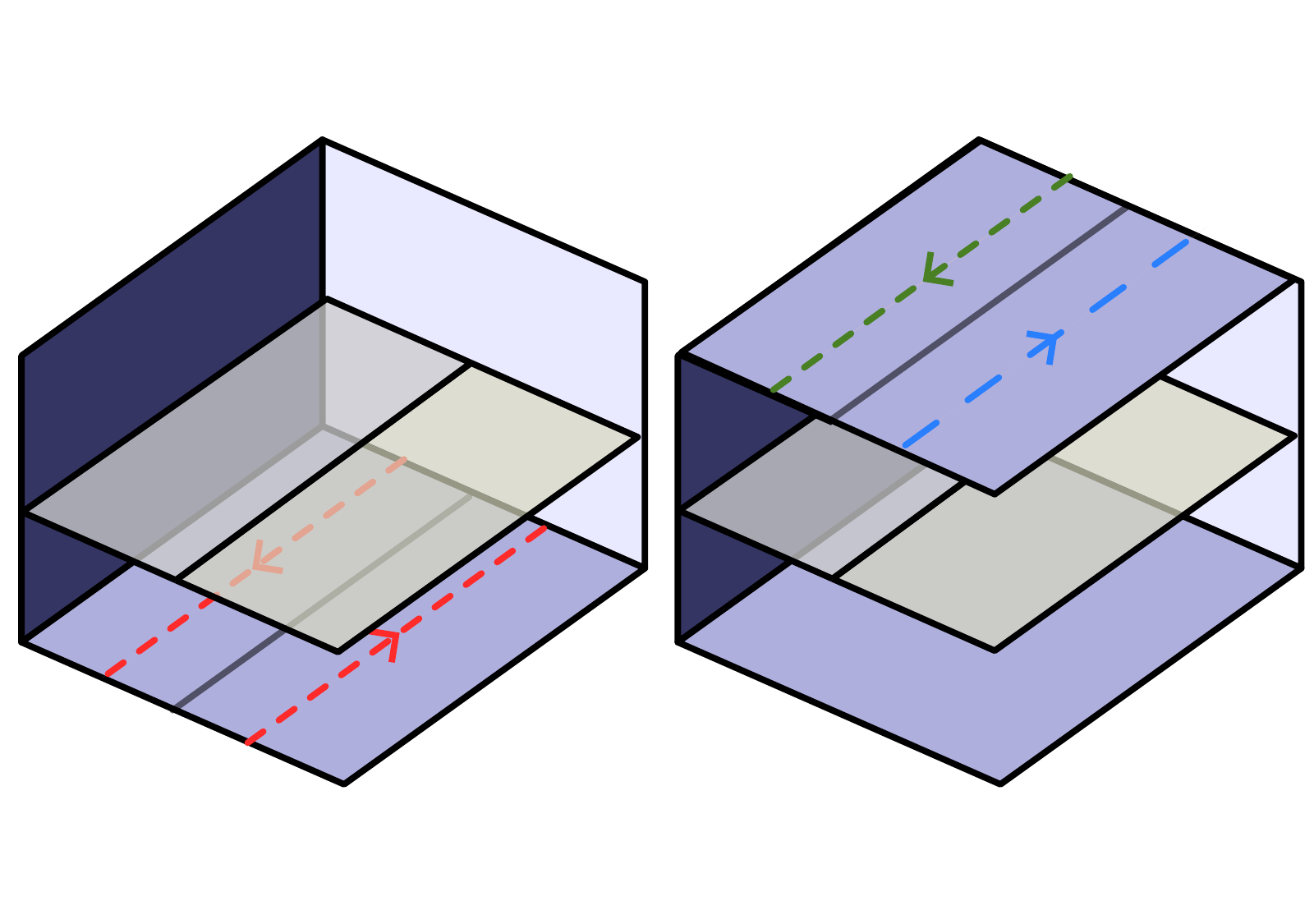}}\put(0.17619525,0.03177889){\color[rgb]{0,0,0}\makebox(0,0)[lt]{\lineheight{1.25}\smash{\begin{tabular}[t]{l}$t=0$\end{tabular}}}}\put(0.36042424,0.35915404){\color[rgb]{0,0,0}\makebox(0,0)[lt]{\lineheight{1.25}\smash{\begin{tabular}[t]{l}$\varphi_A$\end{tabular}}}}\put(0,0){\includegraphics[width=\unitlength,page=2]{domain_walls4.pdf}}\put(0.50919986,0.65552774){\color[rgb]{0,0,0}\makebox(0,0)[lt]{\lineheight{1.25}\smash{\begin{tabular}[t]{l}$t$\end{tabular}}}}\put(0.72430062,0.03177889){\color[rgb]{0,0,0}\makebox(0,0)[lt]{\lineheight{1.25}\smash{\begin{tabular}[t]{l}$t=1$\end{tabular}}}}\put(0.25501937,0.40627622){\color[rgb]{0,0,0}\makebox(0,0)[lt]{\lineheight{1.25}\smash{\begin{tabular}[t]{l}$\varphi_B$\end{tabular}}}}\put(0.09505193,0.10866246){\color[rgb]{1,0.16470588,0.16470588}\makebox(0,0)[lt]{\lineheight{1.25}\smash{\begin{tabular}[t]{l}$W[\tsf{a}]$\end{tabular}}}}\put(0.01320817,0.14834429){\color[rgb]{1,0.16470588,0.16470588}\makebox(0,0)[lt]{\lineheight{1.25}\smash{\begin{tabular}[t]{l}$W[\bar{\tsf{a}}]$\end{tabular}}}}\put(0.89116879,0.53648225){\color[rgb]{0.16470588,0.49803922,1}\makebox(0,0)[lt]{\lineheight{1.25}\smash{\begin{tabular}[t]{l}$W[\varphi_A(\tsf a)]$\end{tabular}}}}\put(0.78824403,0.58608454){\color[rgb]{0.28235294,0.50196078,0.1372549}\makebox(0,0)[lt]{\lineheight{1.25}\smash{\begin{tabular}[t]{l}$W[\varphi_B(\bar{\tsf a})]$\end{tabular}}}}\end{picture}\endgroup    \end{center}
  \vskip -0.5cm
  \caption{At $t=0$, operator $Q_0 = W[\tsf a]W[\bar{\tsf a}]$ with support
    straddling a (closed) segment of the $\varphi_A$-$\varphi_B$ boundary is in
    the stabilizer group. By $t=1$, this operator evolves to the (potentially
  nontrivial) $Q_1 = W[\varphi_A(\tsf a)]W[\varphi_B(\bar{\tsf a})]$.}
  \label{fig:domain_walls4} 
\end{figure}

At the same time, consider the stabilizer 
\begin{equation}
    Q_0 = W[\tsf a]W[\bar{\tsf a}],
\end{equation}
cf. Fig.~\ref{fig:domain_walls4}. This evolves into 
\begin{equation}
    Q_1 = W[\varphi_A(\tsf a)] W[\varphi_B(\bar{\tsf a})]
\end{equation}
at $t=1$. If the boundary is again noncontractible and
\begin{equation}
    \varphi_A(\tsf a)\times \varphi_B(\bar{\tsf a}) \neq \tsf 1,
\end{equation}
or equivalently,
\begin{equation}
    \tsf c \equiv \tsf b \times \tau_{BA}(\bar{\tsf b}) \neq \tsf
    1,\,\text{where}\,\tsf b \equiv \varphi_A(\tsf a),\label{eq:c_anyon}
\end{equation}
then $Q_1$ is a logical operator. Since $Q_0$ was a stabilizer, this
$Q_1$ logical operator will have a $+1$-eigenvalue at $t=1$, regardless of the
encoded state at $t=0$.

Together, these processes show that logical operators and stabilizers
interchange (sketched in Fig.~\ref{fig:logicalsandstab}). Since the sizes of the
stabilizer group and logical operators remain constant between $t=0$ and $t=1$
(as automorphisms are unitary operations) then we must have that for every
logical operator that maps onto a stabilizer, a stabilizer must map onto a
logical operator.\footnote{There is a related behavior whereby local stabilizers
  (that is, their supports are confined to a ball of constant radius relative to
  the linear system size), such as plaquettes, that straddle the $A$-$B$
  boundary map onto operators that create pairs of $\tsf c$ excitations. When
  these stabilizers are measured, they will be re-added to the stabilizer group
along with a random $\pm1$ phase. The product of these phases around a
noncontractible $A$-$B$ boundary segment is precisely the phase of the measured
logical operator.} Moreover, from Eq.~\eqref{eq:anyonlocalization}, $\tsf c$ and
$\tsf d$ are precisely anyons that localize at the $\tau_{BA} =
\varphi_B\varphi_A^{-1}$ twists. Hence, for every independent (under fusion)
nontrivial anyon $\tsf c$ that localizes, if the $\varphi_A$-$\varphi_B$
boundary contains noncontractible segments, a logical operator and stabilizer
interchange, leading to the measurement of encoded information.

\begin{figure}
  \vskip -1cm
  \begin{center}
    \def\svgwidth{\linewidth}
    \begingroup \makeatletter \providecommand\color[2][]{\errmessage{(Inkscape) Color is used for the text in Inkscape, but the package 'color.sty' is not loaded}\renewcommand\color[2][]{}}\providecommand\transparent[1]{\errmessage{(Inkscape) Transparency is used (non-zero) for the text in Inkscape, but the package 'transparent.sty' is not loaded}\renewcommand\transparent[1]{}}\providecommand\rotatebox[2]{#2}\newcommand*\fsize{\dimexpr\f@size pt\relax}\newcommand*\lineheight[1]{\fontsize{\fsize}{#1\fsize}\selectfont}\ifx\svgwidth\undefined \setlength{\unitlength}{283.46456693bp}\ifx\svgscale\undefined \relax \else \setlength{\unitlength}{\unitlength * \real{\svgscale}}\fi \else \setlength{\unitlength}{\svgwidth}\fi \global\let\svgwidth\undefined \global\let\svgscale\undefined \makeatother \begin{picture}(1,0.6)\lineheight{1}\setlength\tabcolsep{0pt}\put(0,0){\includegraphics[width=\unitlength,page=1]{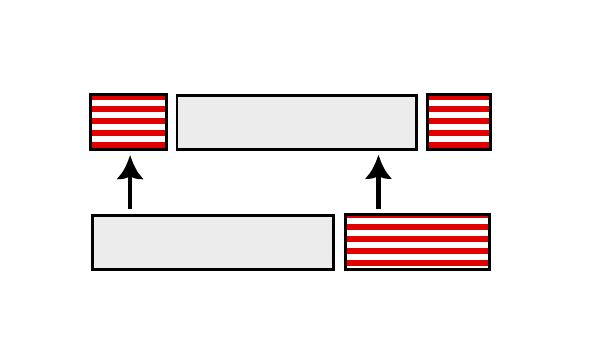}}\put(0.03475745,0.38380754){\color[rgb]{0,0,0}\makebox(0,0)[lt]{\lineheight{1.25}\smash{\begin{tabular}[t]{l}$t=1$\end{tabular}}}}\put(0.03475745,0.18047036){\color[rgb]{0,0,0}\makebox(0,0)[lt]{\lineheight{1.25}\smash{\begin{tabular}[t]{l}$t=0$\end{tabular}}}}\put(0.40595165,0.17419165){\color[rgb]{0,0,0}\makebox(0,0)[lt]{\lineheight{1.25}\smash{\begin{tabular}[t]{l}$\mathcal S$\end{tabular}}}}\put(0,0){\includegraphics[width=\unitlength,page=2]{logicals_and_stabilizers.pdf}}\put(0.76086299,0.17719166){\color[rgb]{0,0,0}\makebox(0,0)[lt]{\lineheight{1.25}\smash{\begin{tabular}[t]{l}$\mathcal L$\end{tabular}}}}\put(0,0){\includegraphics[width=\unitlength,page=3]{logicals_and_stabilizers.pdf}}\end{picture}\endgroup    \end{center}
  \vskip -1.2cm
  \caption{Schematic of the stabilizer group $\mathcal S$ (grey boxes) and
    logical operators $\mathcal L$ (red striped boxes) before ($t=0$) and after ($t=1$)
    a round of competing automorphisms. If the domain wall boundaries contain
    noncontractible segments, then some stabilizers map onto logical operators
    and vice versa. Each independent (up to fusion) nontrivial anyon that
    localizes at $\tau_{BA}$ twists leads to another logical operator that gets
    exchanged with a stabilizer. These logical operators adopt the
    $+1$-eigenvalues of the stabilizers that map onto them. During syndrome
  measurement, these exchanged logical operators will also be measured.}
  \label{fig:logicalsandstab}
\end{figure}

For a $\mathbb Z_2$ TO, such as $TC$ or $CC$, with each anyon its own
antiparticle, and for a $\tau_{BA}$ with quantum dimension $\mathcal D$,
$\log_2\mathcal D^2$ indicates the number of independent nontrivial anyon
strings introduced to the ISG at a $\varphi_A$-$\varphi_B$ boundary segment;
when at least one such segment is noncontractible, $\log_2\mathcal D^2$ is also
the number of logical qubits measured. As noted in
Sec.~\ref{sec:anyonlocalization}, this number is a property of the $\tau_{BA}$
conjugacy class. For example, $\mathcal C\{(c\sigma)(c\sigma)(c\sigma)\}$ is the
only conjugacy class for $\text{Aut}[CC]$ that measures exactly one logical
qubit, $\log_2 \mathcal D^2=1$; the one nontrivial localized anyon for an
automorphism $(c_1\sigma_1)(c_2\sigma_2)(c_3\sigma_3)$ is the fermion
$\tsf{c}_1\sigma_1\times \tsf{c}_2\sigma_2\times\tsf{c}_3\sigma_3$. 

How are the other logical operators affected by the competing automorphisms?
Intuitively, one may imagine that an anyon that transforms identically under
both automorphisms, $\varphi_A(\tsf a) = \varphi_B(\tsf a)$, would be unaffected
by the domain wall boundary. In such a case, the logical would simply transform
as
\begin{equation}
  \bar O[\tsf a] \mapsto \bar O[\varphi_A(\tsf a)],
  \label{eq:logicalmap}
\end{equation}
regardless of the domain wall configurations. The existence of such a logical is
equivalent to there existing an anyon $\tsf b = \varphi_A(\tsf a)$ such that
$\tsf b = \tau_{BA}(\tsf b)$, i.e., $\tsf b$ is invariant under $\tau_{BA}$. We
can extend this idea to show the existence of logical operators that are
unaffected by competing automorphisms, using the lemma (proven in
Appendix~\ref{app:twists}):
\begin{lemma}
  For an automorphism $\tau$ and anyon $\tsf{\emph b}$, $\tsf{\emph b} =
  \tau(\tsf{\emph b})$ if and only if $\tsf{\emph b}$ and $\tsf{\emph c}$ braid
  trivially for all anyons $\tsf{\emph c}$ that localize at $\tau$.
  \label{lemma:retained}
\end{lemma}
Any anyon \tsf{d} that is not invariant under
$\tau_{BA}$ must braid nontrivially with a localized anyon, and therefore
logical operator $\bar\Os[\tsf{d}]$ will be measured out if the
competing automorphisms have a noncontractible boundary segment dual to the
cycle following the support of $\bar\Os[\tsf{d}]$. 
Conversely, if $\tsf b_1$
and $\tsf b_2$ are mutual-semions that are both invariant under $\tau_{BA}$,
then neither $\tsf b_1$ nor $\tsf b_2$ can localize at $\tau_{BA}$. The logical
operators $\bar\Os[\tsf a_1]_v$ and $\bar\Os[\tsf a_2]_h$ [with $\tsf b_1
=\varphi_A(\tsf a_1)$, $\tsf b_2 =\varphi_A(\tsf a_2)$] anticommute, forming a
logical $\bar\X$ and $\bar\Z$ that satisfy Eq.~\eqref{eq:logicalmap}. Moreover,
they do not map onto one of the measured $W[\tsf c]$ $\varphi_A$-$\varphi_B$
boundary-operators, and they also commute with all such measurements.  The
existence of a mutual-semion pair invariant under $\tau_{BA}$ thus guarantees a
logical qubit that is not measured or overwritten in any realization of
competing $\varphi_A$ and $\varphi_B$ domain walls. We call this a ``protected''
logical qubit.\footnote{If invariant anyons do not form a mutual-semion
  pair, then information may not be protected. For example, if $\tau_{BA} =
  (rg)$ in $CC$, then its invariant anyons $\texttt{bx}$, $\texttt{by}$, and
  $\texttt{bz}$ all braid trivially. Moreover, they all localize, e.g.,
  $\texttt{bx} = \texttt{rx}\times (rg)(\texttt{gx})$, and are thus measured out
  along noncontractible boundary segments.} Since both $(\bar\Os[\tsf a_1]_v,
\bar\Os[\tsf a_2]_h)$ and $(\bar\Os[\tsf a_1]_h, \bar\Os[\tsf a_2]_v)$ form
anticommuting $\bar\X$ and $\bar\Z$ pairs, on a $2$-torus we in fact have two
protected qubits per invariant mutual-semion pair. If an anyon $\tsf b$ is
invariant under $\tau_{BA}$, then for any automorphism $\varphi$, $\varphi(\tsf
b)$ is invariant under the conjugate automorphism $\varphi \tau_{BA}
\varphi^{-1}$. Hence, the number of independent invariant mutual-semion pairs
and protected qubits also characterizes the conjugacy class of $\tau_{BA}$. For
$\tau_{BA} \in \text{Aut}[CC]$, the only conjugacy classes with such pairs are
$\mathcal C\{\id\}, \mathcal C\{(c\sigma)(c\sigma)(c\sigma)\}$ and $\mathcal
C\{(ccc)(\sigma \sigma \sigma)\}$, cf. Table~\ref{tbl:conjugacyclasses}. For
example, $\tau_{BA} = (rx)(gy)(bz)$ has $\tsf{rx}$ and $\tsf{bz}$ as invariant
mutual-semions. 

Localized and invariant anyons thus characterize the extent to which logical
information is lost over one period of competing automorphisms. The number of
independent nontrivial localized anyons, $\log_2\mathcal D^2$, indicates the
number of logical qubits measured if the boundaries of the competing
automorphisms contain noncontractible segments. Boundary segments along
homologous noncontractible cycles support representatives of the same measured
logical operator. The number of invariant mutual-semions (IMS) is the number of
logical qubits that are protected and will not be measured out, regardless of
the particular domain wall configurations. For example, if $\log_2\mathcal D^2 =
2$ and $\text{IMS}=0$ for the $CC$ model on the $2$-torus, then for a given
domain wall configuration $2$ of the $4$ qubits will not be measured out between
$t=0$ and $t=1$. However, the same $2$ qubits may be measured in a different
configuration, since $W[\tsf c]$ along each of the $2$ noncontractible cycles
span the logical space.\footnote{This occurs when localized anyons braid
trivially with each other.} On the other hand, if $\text{IMS}=2$ then that means
that the same $2$ logical qubits are protected in any disorder
realization.\footnote{The conclusions in this paragraph use the fact that domain
  wall boundaries cannot contain two disjoint cycles that simultaneously
  encircle orthogonal noncontractible cycles of the $2$-torus. Otherwise, it
would be possible to measure out all $4$ logical qubits when $\log_2\mathcal D^2
=2$ in just one domain wall instance.}

What about multiple consecutive realizations of competing automorphisms?
Consider that between $t=1$ and $t=2$ we again (randomly) partition the manifold
and enact $\varphi_A$ and $\varphi_B$. Assume that the $\varphi_A$-$\varphi_B$
boundaries contain noncontractible segments such that logical operators are
measured. If a protected logical qubit exists from $t=0$ to $t=2$, there must be
mutual-semions that braid trivially with the anyons $\tsf c$ that localize at
$\tau_{BA}$ between $t=0$ and $t=1$. They must also braid trivially with the
anyons that will localize between $t=1$ and $t=2$, that is $\varphi_A^{-1}(\tsf
c)$ or $\varphi_B^{-1}(\tsf c)$. If all localized anyons braid
trivially,\footnote{In $CC$, this is guaranteed whenever there is a single
nontrivial localized anyon, $\log_2 \mathcal D^2=1$.} then
Lemma~\ref{lemma:retained} tells us that $\varphi_A^{-1}(\tsf c) =
\varphi_B^{-1}(\tsf c)$. Extending this argument to $t\rightarrow\infty$, we
therefore have two protected logical qubits for each pair of mutual-semions that
braid trivially with $\varphi_A^{-t}(\tsf c)$ for $t=0,1,2,\ldots$ and for all
$\tsf c$ that localize at $\tau_{BA}$. If not all localized anyons have trivial
mutual statistics,\footnote{In $CC$, nontrivial mutual statistics is required
for $\log_2 \mathcal D^2>1$ with $\text{IMS}>0$, otherwise four commuting
logicals can be measured on conjugate noncontractible $\varphi_A$-$\varphi_B$
boundary cycles.} then requiring that the mutual-semions braid trivially with
all $(\varphi_t\varphi_{t-1}\cdots\varphi_1)(\tsf c)$ where each $\varphi_i \in
\{\varphi_A^{-1},\, \varphi_B^{-1}\}$ and $t=0,1,2,\ldots$ is a sufficient
condition for a protected logical qubit.

\begin{figure}
  \begin{center}
    \includegraphics[width=0.7\linewidth]{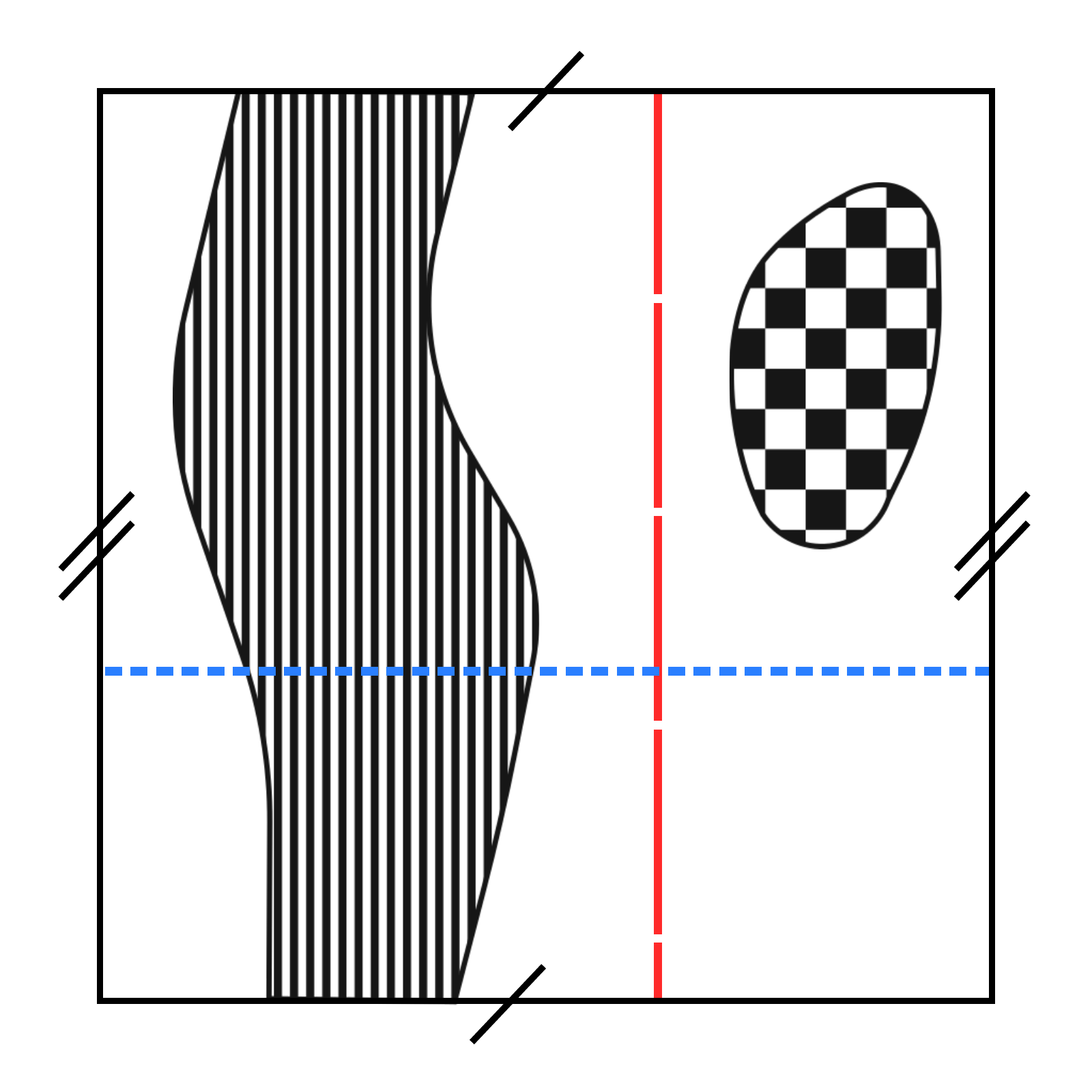}
  \end{center}
  \vskip -0.5cm
  \caption{A $2$-torus showing two temporal domain walls. If the boundary of a
  domain wall has only contractible segments (checked region), there exists at
  least one representative of each logical operator (red and blue lines) that
  avoids intersecting the boundary. On the other hand, a domain wall with a
  noncontractible boundary segment (striped region) always bisects one of each
  $\bar\X$, $\bar\Z$ logical operator pairs.}
  \label{fig:noncontract}
\end{figure}

These behaviors discussed are dependent only on the presence of noncontractible
boundary segments in the domain walls, not on the specific microscopic details.
This allows us to characterize systems of competing automorphism by their
homology. We label a realization of $\varphi_A$ and $\varphi_B$ temporal domain
walls as $\varphi_A$-dominant if the $A$-labelled subregions can completely
contain a noncontractible cycle from every homology class. That is, the support
of a representative of each logical operator can be contained within the
$A$-subregions. $\varphi_B$-dominant configurations are defined analogously. In
these cases, Eq.~\eqref{eq:logicalmap}---with $\varphi = \varphi_A$ or
$\varphi_B$ respectively---holds for a representative of each logical operator.
Moreover, the $\varphi_A$-$\varphi_B$ boundaries must have only contractible
segments, hence logical information is not necessarily lost or
measured.\footnote{Such a domain wall contains a representative of each of the
  torus' noncontractible cycles. Cutting open the torus along one of each cycle
  turns the manifold into a square with open boundaries, which now contains the
  (uncut) boundary segments of the domain wall. These segments must thus all be
contractible.} However, if the $\varphi_A$-$\varphi_B$ boundaries contain
noncontractible segments, such as in Fig.~\ref{fig:noncontract},\footnote{In
  addition to the example striped region in Fig.~\ref{fig:noncontract}, the
  boundary might alternatively contain a noncontractible segment extending
  around both cycles of the torus. In this case, the measured logical operator
  is equivalent to a product of $\bar\Os_v$ and $\bar\Os_h$ logical operators
  using the algebra in Fig.~\ref{fig:logicals}. The resulting behavior is
therefore analogous to the case discussed here.} then at least one logical
operator is measured corresponding to the nontrivial anyons that localize at
$\tau_{BA}$.

Competing automorphisms lead to the measurement of logical information, and can
therefore be naturally applied to explain the purification dynamics of
disordered Floquet codes. Specifically, we envision a Floquet code initialized
into a completely mixed logical state. Subjected to repeated rounds of
heterogeneous temporal domain walls, the average number of logical qubits
measured or protected each timestep can be inferred directly by the conjugacy
classes of the automorphism transition maps and the homology classes of the
cycles contained within the domain wall subregions. This can equivalently be
applied to a fixed initial logical state, and the question of how long 
logical information can remain protected despite the heterogeneous dynamics. In the
following section, we use this information to discuss the behavior of a
disordered dynamical code. In Appendix~\ref{app:critical} we explore in more
detail the purification dynamics of these codes using numerical simulations.

\section{\label{sec:disordermodel}Disordered DA Color Codes}

Section~\ref{sec:competing} showed that heterogeneity in the temporal domain
walls of Floquet codes can lead to the purification of mixed codestates or,
equivalently, the loss of logical information. In this section, we show that
this picture can readily be applied to a class of noise models that affect the
evolution of dynamic automorphism (DA) codes, helping characterize their
resilience to disorder.

In addition to common noise models such as depolarizing channels or
circuit-level noise \cite{knillTheory1997, knillTheory2000,
dennisTopological2002, chubbStatistical2021, nielsenQuantum2010,
loulidiPhysical2024}, to characterize the fault-tolerance of a QEC code and the
stability of its underlying TO one should also consider the effect of
perturbations to the circuit protocol. In measurement-induced FETs such as DA
codes, this can occur via missing measurements. Ref.~\onlinecite{vuStable2024},
for example, considered the effect of missing parity measurements in the
honeycomb Floquet code, finding that the FET phase persisted up to a critical
point despite the disorder. Beyond this, noise models incorporating missing
measurements have been poorly understood in the literature, despite being a
relevant failure mode for quantum devices.

Unlike circuit-level noise, missing measurements can theoretically be tracked
with perfect knowledge by a classical operator of a quantum computer: a
measurement that does not occur may produce no classical output bit. With
perfect knowledge, is it not then trivial to correct for errors that arise from
missing measurements? If the circuit protocol is able to adaptively and
immediately repeat any measurements without output bits this may be easily
overcome. Such on-the-fly adaptivity, however, may introduce classical
bottlenecks and be impractical to implement in systems involving extensive
measurements such as Floquet codes. Repeated measurements are not able to be
parallelized, for example, and allow other errors (such as Pauli errors) to
accumulate while waiting for imperfect measurements to succeed. It would be
beneficial if the effect of missing measurements could be ``swept under the
rug'' by the topological protection of the code, in the same way that Pauli
errors are able to be corrected if their weight is below a threshold
\cite{dennisTopological2002}. This may also allow us to optimize the
circuit protocol, if, for example, we can purposefully skip measurements on
some nonvanishing subset of the physical qubits while still enacting the desired
automorphism gate.

What about correcting the errors after each round of error-detection,
incorporating missing-measurement correction alongside syndrome correction? This
is reminiscent of the problem in $[[n,k,d]]$ stabilizer codes when the locations
of random Pauli errors are known (called ``located'' errors in
Ref.~\onlinecite{nielsenQuantum2010}). In this case, weight $d-1$ errors can be
corrected, instead of just $\lfloor\frac{d-1}2\rfloor$ . Since in topological
codes $d$ is the length of the manifold's shortest nontrivial cycle, this is a
geometric condition, similar to whether competing automorphisms have
noncontractible boundary segments. Missing measurements, as we will argue, can
result in competing automorphisms, and correspond to similar geometric
considerations regarding their capacity to remove logical information. 

In the honeycomb Floquet code, erasure channels \cite{grasslCodes1997}---with
depolarizing noise occurring at known locations---below a certain error rate can
be effectively corrected \cite{guFaulttolerant2023}. Similarly, with missing
measurements in Floquet codes the protection of logical information should
persist up to a critical point (Ref.~\onlinecite{vuStable2024} showed this for
the honeycomb Floquet code). There also exists protected logical subspaces that
are unaffected by missing measurements, even beyond this critical point. As we
will show, this stems directly from Floquet codes with missing-measurement noise
models being systems with competing automorphisms.

For a nontrivial illustration, we focus now on the behavior of the DA color
codes \cite{davydovaQuantum2024} when perturbed by stochastic missing
measurements. This picture can analogously be interpreted as instead randomly
including additional measurements. Recall that in Section~\ref{sec:background}
we introduced DA color codes as a sequence of condensed anyons that correspond
to link measurements on a two-layered honeycomb lattice. To add disorder to a
given stage of measurements, for each associated link on the lattice we
independently include that link in the measurement sequence with probability $p
\in [0,1]$. Otherwise, it is omitted. If we assume that the initial sequence
consists of only reversible pairs of condensations---that is, it is a valid
implementation of the DA color code---then it is not guaranteed that removing
one measurement sequence retains this reversibility. Consider
Eq.~\eqref{eq:rgb}, for example. Focusing on the $\tsf{bz}_1$ condensations in the
third stage, we denote the disordered sequence in the shorthand
\begin{eqnarray}
  \widetilde{CC} \; && \rightarrow \; 
  \magic{1}{}{}{}{}{}{2}{}{} \; \rightarrow \;
  \magic{}{}{}{}{}{2}{}{}{} \; \rightarrow \; 
  \underbrace{\magic{}{}{}{}{}{}{}{}{1}}_{\displaystyle p} \nonumber \\
  && \rightarrow \; 
  \magic{}{2}{}{}{1}{}{}{}{} \; \rightarrow \; 
  \widetilde{CC}.
  \label{eq:example}
\end{eqnarray}
When $p=0$, no $\tsf{bz}_1$ links are measured and we follow the sequence 
\begin{equation*}
  \widetilde{CC} \rightarrow \magic{1}{}{}{}{}{}{2}{}{} \rightarrow
  \magic{}{}{}{}{}{2}{}{}{} \rightarrow \magic{}{2}{}{}{1}{}{}{}{} \rightarrow
  \widetilde{CC},
\end{equation*}
which enacts a valid $(rxgybz)$ automorphism (see
Appendix~\ref{app:automorphisms} for an explanation of deriving this
automorphism). When $p=1$ we measure all the $\tsf{bz}_1$ links and follow
\begin{equation*}
  \widetilde{CC} \rightarrow \magic{1}{}{}{}{}{}{2}{}{} \rightarrow
  \magic{}{}{}{}{}{2}{}{}{1} \rightarrow \magic{}{2}{}{}{1}{}{}{}{} \rightarrow
  \widetilde{CC},
\end{equation*}
with automorphism $(rgb)$. In this way, the presence of measured
$\tsf{bz}_1$-links across the lattice defines disjoint regions of $(rxgybz)$ and
$(rgb)$ temporal domain walls, completely analogous to the discussion of
competing automorphisms in Section~\ref{sec:competing}.
Figure~\ref{fig:domain_walls2} depicts this relationship using a spacetime
illustration. 

\begin{figure}
  \begin{center}
    \def\svgwidth{0.95\columnwidth}
    \begingroup \makeatletter \providecommand\color[2][]{\errmessage{(Inkscape) Color is used for the text in Inkscape, but the package 'color.sty' is not loaded}\renewcommand\color[2][]{}}\providecommand\transparent[1]{\errmessage{(Inkscape) Transparency is used (non-zero) for the text in Inkscape, but the package 'transparent.sty' is not loaded}\renewcommand\transparent[1]{}}\providecommand\rotatebox[2]{#2}\newcommand*\fsize{\dimexpr\f@size pt\relax}\newcommand*\lineheight[1]{\fontsize{\fsize}{#1\fsize}\selectfont}\ifx\svgwidth\undefined \setlength{\unitlength}{767.99990845bp}\ifx\svgscale\undefined \relax \else \setlength{\unitlength}{\unitlength * \real{\svgscale}}\fi \else \setlength{\unitlength}{\svgwidth}\fi \global\let\svgwidth\undefined \global\let\svgscale\undefined \makeatother \begin{picture}(1,0.70312508)\lineheight{1}\setlength\tabcolsep{0pt}\put(0,0){\includegraphics[width=\unitlength,page=1]{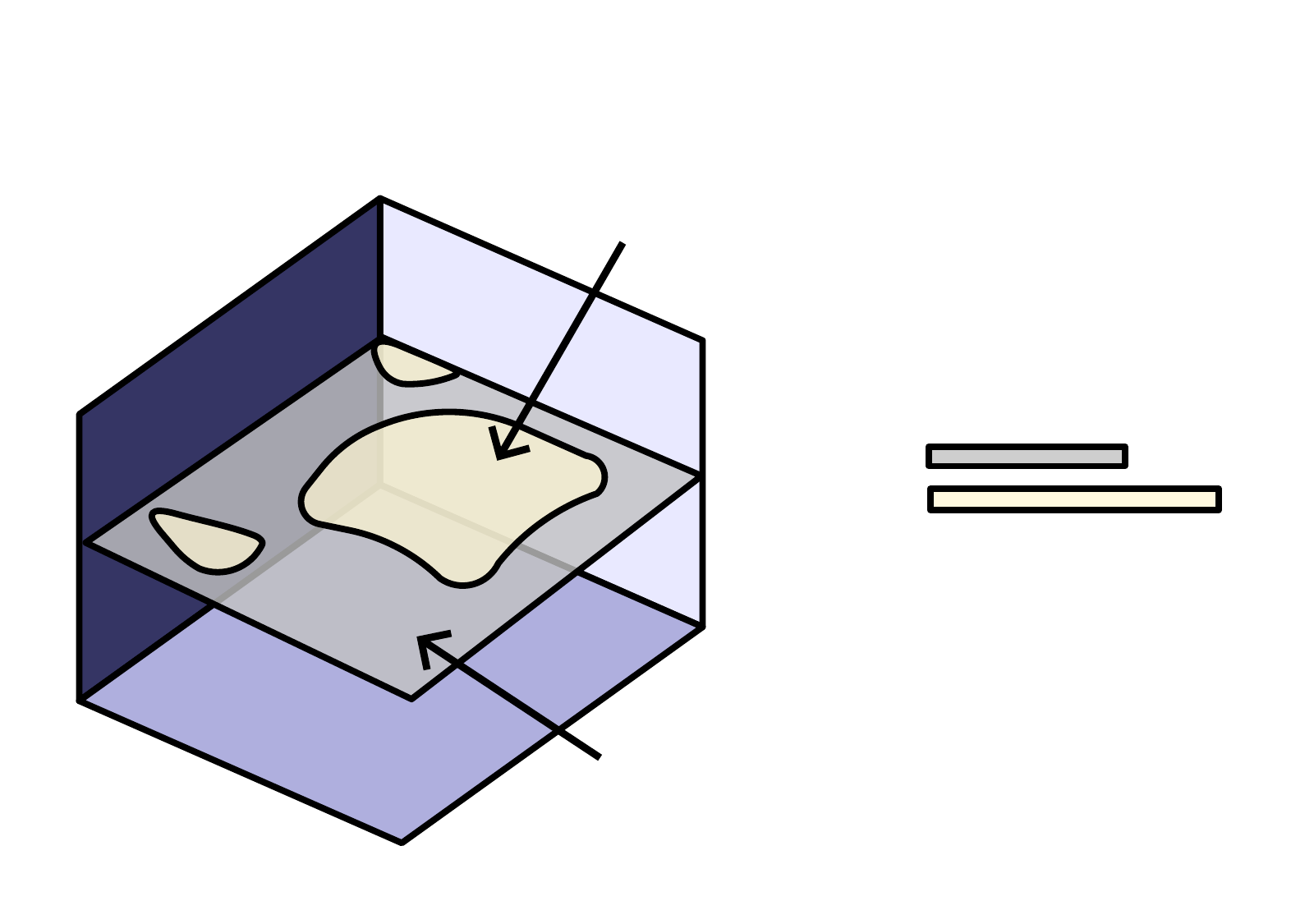}}\put(0.47007627,0.60837429){\color[rgb]{0,0,0}\makebox(0,0)[lt]{\lineheight{1.25}\smash{\begin{tabular}[t]{l}$\magic{1}{}{}{}{}{}{2}{}{} \rightarrow \magic{}{}{}{}{}{2}{}{}{} \rightarrow \magic{}{2}{}{}{1}{}{}{}{}$\end{tabular}}}}\put(0.2136018,0.23886845){\color[rgb]{0,0,0}\makebox(0,0)[lt]{\lineheight{1.25}\smash{\begin{tabular}[t]{l}$(rgb)$\end{tabular}}}}\put(0.41864441,0.2029068){\color[rgb]{0,0,0}\makebox(0,0)[lt]{\lineheight{1.25}\smash{\begin{tabular}[t]{l}$\widetilde{CC}$\end{tabular}}}}\put(0.264258,0.315752){\color[rgb]{0,0,0}\makebox(0,0)[lt]{\lineheight{1.25}\smash{\begin{tabular}[t]{l}$(rxgybz)$\end{tabular}}}}\put(0.01854059,0.37775489){\color[rgb]{0,0,0}\makebox(0,0)[lt]{\lineheight{1.25}\smash{\begin{tabular}[t]{l}$t$\end{tabular}}}}\put(0,0){\includegraphics[width=\unitlength,page=2]{domain_walls2.pdf}}\put(0.4798419,0.09536265){\color[rgb]{0,0,0}\makebox(0,0)[lt]{\lineheight{1.25}\smash{\begin{tabular}[t]{l}$\magic{1}{}{}{}{}{}{2}{}{} \rightarrow \magic{}{}{}{}{}{2}{}{}{1} \rightarrow \magic{}{2}{}{}{1}{}{}{}{}$\end{tabular}}}}\put(0.94033641,0.31274492){\color[rgb]{0,0,0}\makebox(0,0)[lt]{\lineheight{1.25}\smash{\begin{tabular}[t]{l}$\varphi_A$\end{tabular}}}}\put(0.86649121,0.34622647){\color[rgb]{0,0,0}\makebox(0,0)[lt]{\lineheight{1.25}\smash{\begin{tabular}[t]{l}$\tau_{BA}$\end{tabular}}}}\put(0,0){\includegraphics[width=\unitlength,page=3]{domain_walls2.pdf}}\put(0.60961325,0.32886566){\color[rgb]{0,0,0}\makebox(0,0)[lt]{\lineheight{1.25}\smash{\begin{tabular}[t]{l}$\varphi_B \, \biggr\{$\end{tabular}}}}\end{picture}\endgroup    \end{center}
  \caption{Spacetime illustration of one period of the $1$-component disorder
  model from Eq.~\eqref{eq:example}, with two competing automorphisms:
  $\varphi_A = (rxgybz)$ and $\varphi_B = (rgb)$. Time runs upwards and the
  spatial planes represent timeslices of the $\widetilde{CC}$ TO. The $(rgb)$
  domain wall is the dominant region, containing the support of noncontractible
  cycles, characteristic of the supercritical disordered phase. The inset shows
  the perspective where we enact a $\varphi_A$ domain wall everywhere, and only
  the transition map $\tau_{BA}$ forms clusters.}
  \label{fig:domain_walls2} 
\end{figure}

On the other hand, if we modify the sequence to 
\begin{eqnarray}
  \widetilde{CC} \; && \rightarrow \; 
  \magic{1}{}{}{}{}{}{2}{}{} \; \rightarrow \;
  \magic{}{}{}{}{}{2}{}{}{} \; \rightarrow \; 
  \underbrace{\magic{}{}{}{}{}{}{}{}{1}}_{\displaystyle p} \nonumber \\
  && \rightarrow \; 
  \magic{}{2}{}{{\color{RedOrange} 1}}{}{}{}{}{} \; \rightarrow \; 
  \widetilde{CC},
\end{eqnarray}
then when $p=0$ there is an irreversible condensation $\tsf{rx}_1 \rightarrow
\tsf{gx}_1$. This sequence therefore does not produce a transparent domain wall
or enact an automorphism. In general, there are two possible irreversible
scenarios: ``intralayer'' $TC_i \rightarrow TC_i$ irreversible condensations
($i=1,2$); and ``interlayer'' $TC\boxtimes TC \leftrightarrow \widetilde{CC}$
irreversible condensations. Such irreversibility means that anyons of the
effective child theory are condensed, creating
punctures in the TO~\cite{kesselringAnyon2024}. If the affected regions cannot
enclose noncontractible cycles, then, with knowledge of where missing
measurements occurred,\footnote{If we are alerted to a missing measurement, we
can determine if an irreversible condensation occurred by referencing the
measurement schedule and checking if the preceding and following condensed
bosons (that are not missed) braid trivially.} in theory these patches can be
reinitialized later thus restoring the original code without affecting the
logical information. Otherwise, logical information will be lost. The effects of
such irreversible condensations on the code are discussed further in
Appendix~\ref{app:irrev}. For the remainder of this section, we focus on
measurement sequences that are inherently immune to this behavior: missing any
measurement stage will---by construction---not introduce any irreversible
condensations. This avoids the potential for irreparable damage to the
codespace, focusing on the errors that are correctable. One question that we
address by the end of this section is precisely what this limitation leaves us
in terms of potential automorphisms that can still be constructed in this way.
If that set is large enough, this limitation may be beneficial in informing the
design of DA color code quantum circuits that are robust against missing
measurements.

Limiting our systems to just those that produce differing automorphisms at
different realizations of the disorder, we can interpret them as models of
competing automorphisms from Section~\ref{sec:competing}. Understanding what
automorphisms can arise in the different outcomes of disorder can therefore
inform us as to which logical operators are measured or protected, if any. Let
$\varphi_A$ be the automorphism enacted when $p=0$, and $\varphi_B$ when $p=1$.
When the $\varphi_A$-$\varphi_B$ boundaries contain a noncontractible path, we
expect that logical information is measured. Microscopically, the measured
logical operators here are formed from the links at the boundary of the
disordered stage and the immediately-preceding condensate in the same $CC$
layer. In Eq.~\eqref{eq:example}, for example, the measured anyon after the
disordered $\tsf{bz}_1$ condensations is the fermion
$\tsf{rx}_1\times\tsf{bz}_1$. Following the condensations forward to
$\widetilde{CC}$ at $t=1$, this fermion updates to become
$\tsf{rx}_1\tsf{ry}_2\times\tsf{bz}_1$, which is equivalent to the
$\tsf{ry}\times\tsf{bz}$ fermion in $\widetilde{CC}$. This is precisely the
anyon that localizes at $\tau_{BA} = (rgb)\cdot(rxgybz)^{-1} = (rz)(gx)(by)$.
The corresponding logical operators are the $\bar \Os[\tsf{ry}\times\tsf{bz}]$
vertical or horizontal operators in the $\widetilde{CC}$ phase at $t=1$ (or
$\bar\Os[\tsf{by}\times\tsf{gz}]$ at $t=0$). Because the string is formed from
the product of links set by the disorder, there is exactly one anyon species
measured. 

We also saw in Section~\ref{sec:competing} that protected logical
operators correspond to the invariant mutual-semion pairs of $\tau_{BA}$. In our
example, Table~\ref{tbl:conjugacyclasses} shows one such pair for $\tau_{BA}$,
and therefore there exists a $2$-qubit logical Hilbert subspace that remains
unaffected by the missing measurements. We note that in an analysis of the
simpler honeycomb Floquet code, Ref.~\onlinecite{vuStable2024} also found that
fermionic string operators were measured around the boundaries of
missing-measurement regions. Their discussion relied on specific microscopic
details of the measurement sequence, however, whereas our perspective of
competing automorphisms is detail-agnostic and generalizable to other
Abelian-anyon FETs, while also allowing us to intuit further details such as the
presence of protected logical qubits.

Since the size of each domain wall subregion is directly determined by the value
of the parameter $p$, we can associate certain regions in phase-space with the
different behaviors of the $\varphi_A$-$\varphi_B$ boundaries. For some critical
value $p_c \in [0,1]$, we have:
\begin{enumerate}[label=(\arabic*)]
  \item \textbf{Subcritical phase:} $p < p_c$ such that with high probability
    the configuration is $\varphi_A$-dominant.
  \item \textbf{Critical point:} $p \sim p_c$ such that with high probability a
    $\varphi_A$-$\varphi_B$ boundary contains noncontractible segments.
  \item \textbf{Supercritical phase:} $p > p_c$ such that with high probability
    the configuration is $\varphi_B$-dominant.
\end{enumerate}
In the thermodynamic limit of $L\rightarrow \infty$, one---and only
one---logical qubit is measured out only at the critical point $p=p_c$ of this
disorder model. Whether only one logical degree of freedom is measured out over
multiple periods depends on the structure of subsequent disorder partitions as
well as the properties of the automorphisms; these cases are discussed in
Section~\ref{sec:logically}. The value of $p_c$ and the critical behavior of the
model are discussed further in Appendix~\ref{app:critical}. We find that the
systems show transitions in the universality class of bond percolation on kagome
or triangular lattices; this is consistent with earlier works examining missing
measurements in the simpler honeycomb Floquet code \cite{vuStable2024}.

\subsection{\label{sec:connected}Connected FETs}

To understand the effect of missing measurements on the DA color code, we first
simplify our noise model further by assuming that only one stage of the
measurement sequence is disordered. We call this a ``$1$-component disorder
model'' with a parameter $p$; all the examples given thus far fall in this
category. In Section~\ref{sec:mcomp} we generalize our results to the case where
multiple---or all---stages may contain missing measurements. We are concerned
with two important behaviors: (1) what automorphisms can be realized
while maintaining only reversible condensation sequences; and (2) what the
requirements are for there to exist protected logical subspaces, despite the
disorder. We focus on the first question in this section, and address the second
in Section~\ref{sec:logically}.

To answer this, we introduce the following definitions:
\begin{definition}[Adjacent FETs]
  Two FETs $A$ and $B$ with automorphisms $\varphi_A$ and $\varphi_B$ are
  adjacent if there exists a $1$-component disorder model that realizes
  competing $\varphi_A$ and $\varphi_B$ temporal domain walls.
\end{definition}
\begin{definition}[Connected FETs]
  Two FETs, $A_0$ and $A_m$, are connected if there exists an adjacency sequence
  of FETs $\{A_0,\, A_1,\,\ldots,\, A_m\}$, such that $A_i$ and $A_{i+1}$ are
  adjacent for all $i=0,\ldots,{m-1}$. The length of this sequence is defined as
  $m$.
  \label{def:connected}
\end{definition}

It is worth clarifying here that we are working in a space of FETs in which they
are labeled solely by their automorphisms. Hence, distinct measurement sequences
realizing the same automorphism are identified. If two FETs are adjacent, it
does not guarantee that there is a 1-component disorder model involving any
given two measurement sequences realizing FETs $A$ and $B$. Rather, at least one
exists. Similarly, with connected FETs, for each $j=1,\,\ldots,\, m-1$, we allow
for FET $A_j$ to be realized by distinct measurement sequences in the
$1$-component disorder models connecting it to $A_{j-1}$ and to
$A_{j+1}$.\footnote{In a space that distinguishes distinct measurement
sequences, different measurement sequences realizing the same FET may not be
connected in the sense of Definition~\ref{def:connected}, see
App.~\ref{app:connected}.} 

We argued previously that at criticality a $1$-component disorder model measures
exactly one logical qubit (with high probability). By
Section~\ref{sec:competing}, this necessitates that $\tau_{BA}$ has
$\log_2\mathcal D^2 = 1$. From Table~\ref{tbl:conjugacyclasses}, for
$\text{Aut}[CC]$ there is only one conjugacy class for which this is true. Two
FETs in the DA color code are thus adjacent only if their automorphisms satisfy
the \emph{separation condition}
\begin{equation}
  \tau_{BA} \in \mathcal C\{(c\sigma)(c\sigma)(c\sigma)\}.
  \label{eq:separation}
\end{equation}
There are several immediate consequences of this:
\begin{enumerate}[label=(\arabic*)]
  \item The trivial FET, $\mathbbm 1$, with $\varphi = {\id}$ is adjacent only
    to FETs with automorphisms in the $\mathcal
    C\{(c\sigma)(c\sigma)(c\sigma)\}$ conjugacy class. This follows from
    $\varphi_B = \tau_{BA}\cdot\id$.
  \item Adjacent FETs always have different parities on the $S_2$ subgroup of
    $\text{Aut}[CC]$ (color-flavor exchange, c.f.
    Appendix~\ref{app:autotheory}), but the same parities on the $S_3\times S_3$
    subgroup (color or flavor permutations). This comes from $\tau_{BA}$ having
    odd-parity on the $S_2$ subgroup, but even-parity on the $S_3\times S_3$
    subgroup.
  \item Two FETs in the same conjugacy class are never adjacent, since elements
    of a conjugacy class have the same parities on all subgroups.
  \item The logical operator that is measured when the system is tuned near the
    critical point must be a fermion string. This is because there is precisely
    one nontrivial anyon---a fermion---that localizes at twists corresponding to
    the automorphisms in $\mathcal C\{(c\sigma)(c\sigma)(c\sigma)\}$
    \cite{kesselringBoundaries2018}. For a
    $(c_1\sigma_1)(c_2\sigma_2)(c_3\sigma_3)$ automorphism, this fermion is
    $\tsf{c}_1\sigma_1\times \tsf{c}_2\sigma_2 \times \tsf{c}_3\sigma_3$.
\end{enumerate}

We can promote this separation condition to a sufficient and necessary condition
by showing that for all FET pairs with automorphisms satisfying
Eq.~\eqref{eq:separation} there exists a measurement sequence and $1$-component
disorder model that connects them. We first introduce the idea of concatenating
two measurement sequences: let $\mathcal A_i, \mathcal C_j$ denote some
$TC\boxtimes TC$ child theories. For a sequence 
\begin{equation} 
  \widetilde{CC} \rightarrow \mathcal A_1 \rightarrow \cdots \rightarrow
  \mathcal A_m \rightarrow \widetilde{CC} 
\end{equation} 
that realizes automorphism $\varphi_A$, and a sequence 
\begin{equation} 
  \widetilde{CC} \rightarrow \mathcal A_m \rightarrow \mathcal C_1 \rightarrow
  \cdots \rightarrow \mathcal C_n \rightarrow \widetilde{CC} 
\end{equation}
that realizes automorphism $\varphi_C$, we can construct the concatenated
sequence 
\begin{equation} 
  \widetilde{CC} \rightarrow \mathcal  A_1 \rightarrow \cdots \rightarrow
  \mathcal A_m \rightarrow C_1 \rightarrow \cdots \rightarrow \mathcal C_n
  \rightarrow \widetilde{CC} \label{eq:concatenated} 
\end{equation}
that realizes automorphism $\varphi_C\varphi_A$. Now, the trivial FET, $\mathbbm
1$, is adjacent to all FETs in $\mathcal C\{(c\sigma)(c\sigma)(c\sigma)\}$; we
explicitly provide example $1$-component disorder models in
Appendix~\ref{app:examples} to prove this. Then, let $A$ and $B$ be any two FETs
with automorphisms $\varphi_A$ and $\varphi_B$ that satisfy the separation
condition, $\tau_{BA} \in \mathcal C\{(c\sigma)(c\sigma)(c\sigma)\}$. It is
possible (see Appendix~\ref{app:automorphisms}) to construct a measurement
sequence for $A$ such that its final $TC\boxtimes TC$ child theory is the same
as the first $TC\boxtimes TC$ child theory of the measurement sequence that
connects $\id$ and $\tau_{BA}$ via a $1$-component disorder model. Using the
result above, we concatenate the measurement sequence for $A$ with the
measurement sequence that realizes $\varphi_C = \id$ or $\varphi_C = \tau_{BA}$.
We now have a $1$-component disorder model that creates automorphisms $\varphi_A
= \id\cdot\varphi_A$ and $\varphi_B = \tau_{BA}\varphi_A$. Thus, two FETs are
adjacent if and only if their automorphisms satisfy Eq.~\eqref{eq:separation}.

We now build towards an understanding of what automorphisms can be realized
while maintaining just $1$-component disorder models: for example, are any two
arbitrary FETs connected? Equivalently, constructing the graph $G = (N, E)$ with
each node in $N$ a distinct FET and an edge in $E$ joining adjacent FETs, is
this graph connected? If two nodes $n, m \in N$ are joined by an edge, then
their automorphisms $\varphi_n, \varphi_m$ satisfy $\tau_{mn} \in \mathcal
C\{(c\sigma)(c\sigma)(c\sigma)\}$. By Section~\ref{sec:connected}, $\varphi_m =
\tau_{mn} \varphi_n$ has the same parity on $S_3\times S_3$ as $\varphi_n$.
There is thus no path in $G$ connecting two FETs with automorphisms of different
parity on $S_3\times S_3$, and $G$ has (at least) two non-empty connected
components. Each component contains $36$ FETs, grouped by the conjugacy class of
their automorphisms (cf. Table~\ref{tbl:conjugacyclasses}).\footnote{There is no
direct interpretation of these two components in terms of logical gates. In
particular, the mapping from automorphisms to gates is intrinsically dependent
on the geometry and boundary conditions of the manifold. For a fixed algebra on
the $2$-torus, each automorphism can be written as a logical gate in the
$4$-qubit Clifford group. No apparent structure emerges, however, when viewing
these gates grouped by their parity.}

\begin{figure}
  \begin{center}
    \def\svgwidth{\linewidth}
    \begingroup \makeatletter \providecommand\color[2][]{\errmessage{(Inkscape) Color is used for the text in Inkscape, but the package 'color.sty' is not loaded}\renewcommand\color[2][]{}}\providecommand\transparent[1]{\errmessage{(Inkscape) Transparency is used (non-zero) for the text in Inkscape, but the package 'transparent.sty' is not loaded}\renewcommand\transparent[1]{}}\providecommand\rotatebox[2]{#2}\newcommand*\fsize{\dimexpr\f@size pt\relax}\newcommand*\lineheight[1]{\fontsize{\fsize}{#1\fsize}\selectfont}\ifx\svgwidth\undefined \setlength{\unitlength}{1518.75bp}\ifx\svgscale\undefined \relax \else \setlength{\unitlength}{\unitlength * \real{\svgscale}}\fi \else \setlength{\unitlength}{\svgwidth}\fi \global\let\svgwidth\undefined \global\let\svgscale\undefined \makeatother \begin{picture}(1,0.89975309)\lineheight{1}\setlength\tabcolsep{0pt}\put(0,0){\includegraphics[width=\unitlength,page=1]{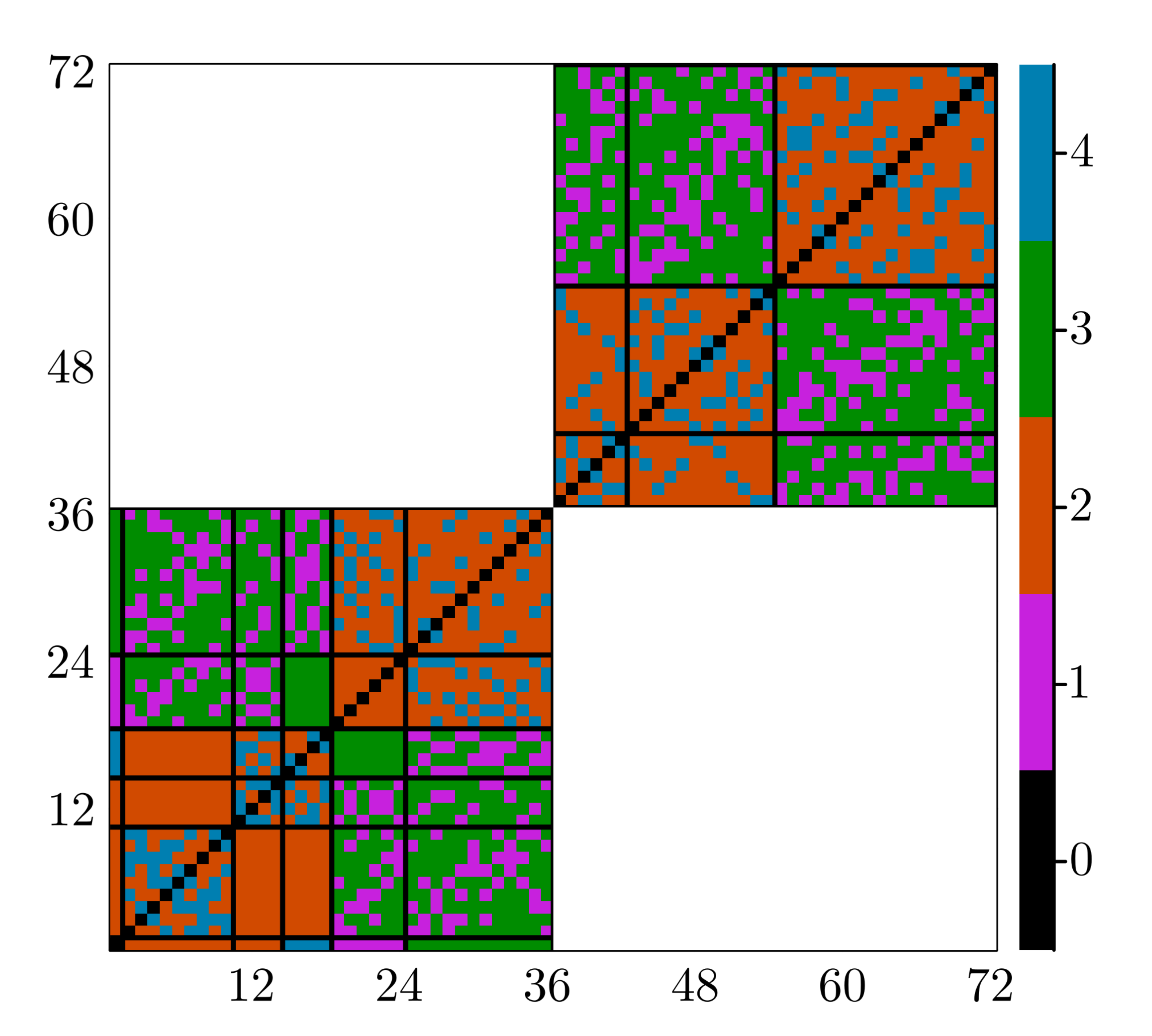}}\put(0.50299167,0.07264602){\color[rgb]{0,0,0}\makebox(0,0)[lt]{\lineheight{1.25}\smash{\begin{tabular}[t]{l}$\leftarrow \quad \mathcal C\{\id\}$\end{tabular}}}}\put(0.50299167,0.12589586){\color[rgb]{0,0,0}\makebox(0,0)[lt]{\lineheight{1.25}\smash{\begin{tabular}[t]{l}$\leftarrow \quad \mathcal C\{(cc)(\sigma\sigma)\}$\end{tabular}}}}\put(0.50299167,0.1956278){\color[rgb]{0,0,0}\makebox(0,0)[lt]{\lineheight{1.25}\smash{\begin{tabular}[t]{l}$\leftarrow \quad \mathcal C\{(ccc)(\sigma\sigma\sigma)\}$\end{tabular}}}}\put(0.50299167,0.23875616){\color[rgb]{0,0,0}\makebox(0,0)[lt]{\lineheight{1.25}\smash{\begin{tabular}[t]{l}$\leftarrow \quad \mathcal C\{(ccc)\}$\end{tabular}}}}\put(0.50299167,0.29174239){\color[rgb]{0,0,0}\makebox(0,0)[lt]{\lineheight{1.25}\smash{\begin{tabular}[t]{l}$\leftarrow \quad \mathcal C\{(c\sigma)(c\sigma)(c\sigma)\}$\end{tabular}}}}\put(0.50299167,0.38416028){\color[rgb]{0,0,0}\makebox(0,0)[lt]{\lineheight{1.25}\smash{\begin{tabular}[t]{l}$\leftarrow \quad \mathcal C\{(c\sigma c\sigma c\sigma)\}$\end{tabular}}}}\put(0.25215028,0.48397164){\color[rgb]{0,0,0}\makebox(0,0)[lt]{\lineheight{1.25}\smash{\begin{tabular}[t]{l}$\mathcal C\{(cc)\} \quad \rightarrow$\end{tabular}}}}\put(0.16414676,0.57638952){\color[rgb]{0,0,0}\makebox(0,0)[lt]{\lineheight{1.25}\smash{\begin{tabular}[t]{l}$\mathcal C\{(ccc)(\sigma\sigma)\} \quad \rightarrow$\end{tabular}}}}\put(0.14812865,0.73411606){\color[rgb]{0,0,0}\makebox(0,0)[lt]{\lineheight{1.25}\smash{\begin{tabular}[t]{l}$\mathcal C\{(c\sigma c\sigma)(c\sigma)\} \quad \rightarrow$\end{tabular}}}}\end{picture}\endgroup    \end{center}
  \vskip -0.5cm
  \caption{Each row and column corresponds to one of the $72$ automorphisms in
    $\text{Aut}[CC]$, grouped by conjugacy class. The color indicates the minimum
    number of edges connecting the nodes in $G$ associated with the row and column
    automorphisms; equivalently, it is the length of the minimum adjacency
    sequence connecting the two FETs of those automorphisms. There are two
    distinct clusters, grouped by their parity on the $S_3\times S_3$ subgroup;
    even-parity automorphisms are in the bottom left quadrant and odd-parity in
    the top right quadrant. White squares indicate that there is no possible
    adjacency sequence to connect those two FETs.}
  \label{fig:distances}
\end{figure}

We numerically compute all inequivalent $1$-component disorder models by
enumerating the possible isomorphism contributions (see
Appendix~\ref{app:automorphisms}), thus specifying the graph $G$.
Figure~\ref{fig:distances} shows the minimum graph distances on $G$ between all
FETs, as well as displaying the separation into exactly two connected
components.

We can explain these graph distances, starting from the simple case of
transitions from the trivial FET: We begin with the sequence of two adjacent
FETs $\{\mathbbm 1,\, A_1\}$. Using a $1$-component disorder model, if $A_1$ is
adjacent to the trivial FET then its automorphism $\varphi_1$ satisfies
$\varphi_1 = \tau_{10} \cdot \id = \tau_{10}$, where $\tau_{10} \in \mathcal
C\{(c\sigma)(c\sigma)(c\sigma)\}$. Any automorphism in $\mathcal
C\{(c\sigma)(c\sigma)(c\sigma)\}$ can be realized from the identity using a
$1$-component disorder model. Since these automorphisms exchange color and
flavor, we can interpret $\tau_{10}$ geometrically as a reflection of the magic
square (placed on a $2$-torus) along a mirror line parallel to the diagonal or
antidiagonal. Specifically, the mirror line intersects the three bosons listed
by the three $2$-cycles of the automorphism. $(ry)(gz)(bx)$, for example, has a
mirror line through anyons $\tsf{ry}$, $\tsf{gz}$, and $\tsf{bx}$. 

For a sequence $\{\mathbbm 1,\, A_1,\,A_2\}$, the associated automorphism for
$A_2$ must satisfy $\varphi_2 = \tau_{21}\varphi_1 = \tau_{21}\tau_{10}$, where
both $\tau_{21},\,\tau_{10} \in \mathcal C\{(c\sigma)(c\sigma)(c\sigma)\}$ are
such reflections. If $\tau_{21} = \tau_{10}$ we arrive back at $\id$. Otherwise,
there are two scenarios to consider: 
\begin{enumerate}[label=(\arabic*)]
  \item If the two mirror lines are perpendicular, then by the
    \emph{Compositions of Reflections over Intersecting Lines Theorem}, this
    enacts a rotation by $\pi$ about their intersection. This populates the
    $\mathcal C\{(cc)(\sigma\sigma)\}$ conjugacy class. There are $9$ such
    intersections (each entry of the magic square), agreeing with the class's
    number of elements in Table~\ref{tbl:conjugacyclasses}.
  \item If the two mirror lines are parallel, then by the \emph{Reflection in
    Parallel Lines Theorem}, this enacts a translation normal to the two lines.
    That is, $\varphi_2$ translates along either the diagonal or antidiagonal
    directions of the magic square, and thus belongs to $\mathcal
    C\{(ccc)(\sigma\sigma\sigma)\}$. There are two directions and two nontrivial
    and nonequivalent magnitudes of translation, forming the $4$ elements in
    this class. 
\end{enumerate}

There are two more conjugacy classes remaining in the even-parity component:
$\mathcal C\{(c\sigma c\sigma c\sigma)\}$ and $\mathcal C\{(ccc)\}$. We can
realize any automorphism in the former class, e.g. $(c_1\sigma_1 c_2\sigma_2 c_3
\sigma_3)$, with a sequence $\{\mathbbm 1,\,A_1,\,A_2,\,A_3\}$ by choosing
$\tau_{32} = (c_1\sigma_3)(c_2\sigma_2)(c_3\sigma_3)$ and $\tau_{21}\tau_{10} =
(c_1c_3)(\sigma_1\sigma_2)$ such that $\varphi_3 = \tau_{32}\tau_{21}\tau_{10}$.
It is not possible to realize $\mathcal C\{(ccc)\}$ with such a sequence because
this class has trivial $S_2$ components but an odd number of $\tau$ reflections
results in a net nontrivial reflection. Rather, using a sequence $\{\mathbbm
1,\,A_1,\,A_2,\,A_3,\,A_4\}$ we can compose a diagonal translation
$\tau_{43}\tau_{32}$ with an antidiagonal translation $\tau_{21}\tau_{10}$ such
that $\varphi_4 = \tau_{43}\tau_{32}\tau_{21}\tau_{10}$ translates along the
vertical or horizontal directions of the magic square, realizing any $\mathcal
C\{(ccc)\}$.\footnote{We can also show that these are the minimum graph
distances between each FET by considering the $\mathcal D^2$ of the conjugacy
classes. The reflections $\mathcal C\{(c\sigma)(c\sigma)(c\sigma)\}$ have
$\mathcal D^2=2$. $\mathcal C\{(ccc)(\sigma\sigma\sigma)\}$ and $\mathcal
C\{(cc)(\sigma\sigma)\}$ have $\mathcal D^2=4$, requiring two reflections to
populate a fusion group with the required number of localized anyons. Similarly,
$\mathcal C\{(c\sigma c\sigma c\sigma)\}$ has $\mathcal D^2=8$ and $\mathcal
C\{(ccc)\}$ has $\mathcal D^2=16$, requiring three and four reflections
respectively.}

For any two arbitrary FETs in the same component, $A$ and $B$, with associated
automorphisms $\varphi_A$ and $\varphi_B$, the minimum graph distance between
them can be found by identifying the minimum graph distance between the trivial
FET and the FET with automorphism $\tau_{BA}$. That is, let the minimum
adjacency sequence between $\mathbbm 1$ and the FET with automorphism
$\tau_{BA}$ be $\{\mathbbm 1,\,A_1,\ldots,\,A_{\tau_{BA}}\}$ containing $m+1$
FETs. Then by concatenating each of the $m$ $1$-component disorder models in
this adjacency sequence with the measurement sequence for $A$, we get $\{A, \,
A\cdot A_1,\,\ldots,\,A\cdot A_{\tau_{BA}} = B\}$ (with $\cdot$ used informally
here to denote the result of concatenating the two measurement sequences of
those FETs). There does not exist a shorter adjacency sequence between $A$ and
$B$, because if there did then we could perform the reverse process and
concatenate each of its FETs with a measurement sequence for $\varphi_A^{-1}$,
thereby realizing an adjacency sequence between $\mathbbm 1$ and $\tau_{BA}$
with less than $m+1$ FETs. For example, take $(rb)$ and $(rgb)(xy)$. Their
transition map is $\tau = (rgb)(xy)\cdot (rb)^{-1} = (gb)(xy)$, which is
connected to the trivial FET via a graph distance of $2$, and thus the FETs with
automorphisms $(rb)$ and $(rgb)(xy)$ are also connected via a graph distance of
$2$.

$1$-component disorder models therefore prompt a notion of connectivity between
FETs of the DA color code. We have seen that two FETs with different parity on
the $S_3\times S_3$ subgroup cannot be connected in this way, and the minimum
length adjacency sequence for two FETs with automorphisms $\varphi_A$ and
$\varphi_B$ is found by taking the minimum length sequence between the trivial
FET and the FET with automorphism $\tau_{BA}$.
These results will be used in Section~\ref{sec:mcomp} to justify conditions
relating the potential automorphisms that can be realized while preserving
reversible condensations despite missing measurements.

\subsection{\label{sec:logically}Logically-Connected FETs}
We have so far stated sufficient and necessary conditions on the theoretical
ability for disorder to generate different automorphisms and FETs. We now
consider the behavior of the systems evolving over multiple periods of these
measurement sequences, and whether there exists protected logical qubits that
remain unaffected, even near the critical point of $p\sim p_c$ where logical
measurements are expected. In a DA circuit, restricting gates to only those that
correspond to the automorphisms of logically-connected FETs helps ensure that
logical qubits remain protected even with increasing circuit depth.

In the long-time limit, the number of protected logical qubits supported on a
$2$-torus will always be even: if an operator $\bar\Os[\tsf c]_v$ is measured,
the commuting operator $\bar\Os[\tsf c]_h$ can also be measured under a
different disorder realization in subsequent periods. A code evolving under a
particular disorder model will therefore always have $0$, $2$, or $4$ qubits
measured out (or conversely, protected) in the limit of $t\rightarrow \infty$
periods. To quantify which competing automorphisms permit such protected qubits,
we introduce a stricter definition of connectedness:
\begin{definition}[Logically-Connected]
  Two FETs $A_0$ and $A_m$ are logically-connected if there exists an adjacency
  sequence of FETs $\{A_0,\,\ldots,\,A_m\}$ with a consistent
  nonzero-dimensional logical Hilbert subspace that remains protected in the
  limit of $t\rightarrow \infty$ periods, for any $p \in [0,1]$ in any of the
  $m$ sets of $1$-component disorder models between FETs $A_i$ and $A_{i+1}$.
\end{definition}
\noindent Here, ``consistent'' means that the same logical Hilbert subspace is
protected in all of the $1$-component disorder models; we interpret this as
there being a logical qubit that is unaffected by the disorder modifying $A_0$
into $A_m$. By definition, any $A_i$, $A_j$ in the sequence are also
logically-connected, and logically-connected FETs are necessarily also connected
FETs. 

We first consider adjacent FETs. Let $\varphi_A$ and $\varphi_B$ be the two
competing automorphisms. The nontrivial anyon that localizes at $\tau_{BA} \in
\mathcal C\{(c\sigma)(c\sigma)(c\sigma)\}$ is a fermion, which braids trivially
with itself and the vacuum. From Section~\ref{sec:competing}, if all anyons
$\tsf c$ that localize at $\tau_{BA}$ braid trivially, then each pair of
mutual-semions that are invariant under $\tau_{BA}$ and braid trivially with all
$\varphi_A^{-t}(\tsf c)$ for $t = 0,1,2,\ldots$ guarantees the existence of $2$
protected logical qubits as $t\rightarrow \infty$. 

This allows us to restrict which automorphisms can be logically-connected; to do
so, we first explain how automorphisms map the fermions of the color code. As
detailed in Appendix~\ref{app:fermions}, there are $6$ fermions, forming two
fermion groups $F$ and $F'$ with $-1$ mutual statistics between different
fermions within the same group and trivial mutual statistics otherwise. These
fermions are mapped by automorphisms according to the lemma (proof in
Appendix~\ref{app:fermions}):
\begin{lemma}
  For any fermion $\emph{\tsf{f}}$, if the automorphism $\varphi$ has even
  parity on the subgroup $S_3\times S_3$, then the fermion
  $\varphi(\emph{\tsf{f}})$ is in the same fermion group as $\emph{\tsf{f}}$. If
  the parity is odd, then $\varphi(\emph{\tsf{f}})$ is in the other fermion
  group.
  \label{lemma:fermions}
\end{lemma}

Now, if two fermions are in different fermion groups, there is no fermion that
braids trivially with both. Moreover, using the ``fermion magic square'' from
Appendix~\ref{app:fermions}, a given fermion $\tsf f \in F$ only braids
trivially with the bosons in its row (and the vacuum). For example,
$\tsf{ry}\times\tsf{bx}\times\tsf{gz}$ only braids trivially with the
$\tsf{ry}$, $\tsf{bx}$, and $\tsf{gz}$ bosons. Similarly, $\tsf f' \in F'$ only
braids trivially with the bosons in its column. Therefore, the sole boson that
braids trivially with both $\tsf f$ and $\tsf f'$ is the boson at the
intersection of the row and column. By Lemma~\ref{lemma:fermions}, for any
fermion $\tsf f$ there thus exists one nontrivial anyon, not a pair of
mutual-semions, that braids trivially with both $\tsf f$ and $\varphi^{-1}(\tsf
f)$ if $\varphi$ has odd-parity on $S_3\times S_3$ (noting that the parity of
$\varphi$ and $\varphi^{-1}$ are the same). Hence, two FETs with automorphisms
of odd-parity on $S_3\times S_3$ cannot be logically-connected. Moreover, since
FETs in the odd-parity component are connected only to other FETs in the
odd-parity component, this means that they are logically-connected to no FET.
Any $1$-component disorder model involving an odd-parity FET and tuned near the
critical point will necessarily measure out all $4$ logical qubits given enough
time.

We now consider FETs connected via a sequence of adjacent FETs $\{A_0, \, A_1,
\, \ldots, \, A_m\}$ with automorphisms $\varphi_0, \varphi_1, \ldots,
\varphi_m$. If there exists a common protected logical subspace in the
$1$-component disorder models between each pair of FETs $A_i$ and $A_{i+1}$,
then the localized fermion in each case must be from the same fermion group. We
note the lemma (proof in Appendix~\ref{app:fermions}):

\begin{lemma}
  If two reflections $\tau_1, \tau_2 \in \mathcal
  C\{(c\sigma)(c\sigma)(c\sigma)\}$ are about parallel mirror lines of the magic
  square, then their localized anyons are fermions in the same fermion group.
  Otherwise, they are in different fermion groups. 
  \label{lemma:reflections}
\end{lemma}

Therefore, for this condition to hold, each $\tau_{(i+1)i}$ must be a reflection
about a mirror line parallel to all other $\tau_{(j+1)j}$ in the adjacency
sequence, with $i,j = 0,\ldots,m-1$. By the \emph{Reflection in Parallel Lines
Theorem}, the only possible conjugacy classes created from an even number of
these reflections are translations along a diagonal of the magic square, that is
$\mathcal C\{\id\}$ or $\mathcal C\{(ccc)(\sigma\sigma\sigma)\}$. An odd number
of reflections results in another reflection about a parallel mirror line.
Therefore, assuming that $A_0 \neq A_m$, a necessary condition for $A_0$ and
$A_m$ to be logically-connected is that
\begin{equation}
  \tau_{m0} \in \mathcal C\{(ccc)(\sigma\sigma\sigma)\}
  \label{eq:m_even}
\end{equation}
for even $m$, or 
\begin{equation}
  \tau_{m0} \in \mathcal C\{(c\sigma)(c\sigma)(c\sigma)\}
  \label{eq:m_odd}
\end{equation}
for odd $m$. Furthermore, we know from Section~\ref{sec:connected} that if
$\tau_{BA} \in \mathcal C\{(c\sigma)(c\sigma)(c\sigma)\}$ then FETs $A$ and $B$
are adjacent, and if $\tau_{BA} \in \mathcal C\{(ccc)(\sigma\sigma\sigma)\}$
then FETs $A$ and $B$ are connected using two $1$-component disorder models.
Moreover, in a valid FET with reversible condensations we conjecture that there
is no other mechanism for affecting the logical subspace beyond the measured
fermions discussed here, and so these conditions should guarantee the existence
of a consistent protected non-zero dimensional logical subspace. Assuming this
conjecture to be true, any two arbitrary FETs $A$ and $B$ are
logically-connected iff their automorphisms have even parity on $S_3\times S_3$,
and $\tau_{BA}$ satisfies either Eq.~\eqref{eq:m_even} or Eq.~\eqref{eq:m_odd}. 

For models where there is a protected logical subspace, these results also
enable a prescription to identify representatives of its protected logical
operators: for $S_3\times S_3$ even-parity automorphisms,
$\varphi^{-t}(\tsf{f})$ are guaranteed to be in the same fermion group for all
integers $t\geq 0$ by Lemma~\ref{lemma:fermions}. Therefore, the protected
logical operators are constructed out of fermions from the other group, which
are guaranteed to braid trivially with all measured fermions. There are thus two
possible candidates for logical subspaces that are protected in a $1$-component
disorder model: if the localized fermion belongs to the $F$ fermion group, then
a representative logical algebra is given by Table~\ref{tbl:F_logicals}; if it
belongs to the $F'$ fermion group, then we instead use
Table~\ref{tbl:Fbar_logicals}. Because the anyons that construct these protected
operators are all invariant under the transition maps between the
logically-connected FETs, the time-evolution of an observable from this
protected algebra at point $p\in[0,1]$ in any logically-connected $1$-component
disorder model is indistinguishable from any other point $\tilde p \in [0,1]$.
That is, if two FETs are logically-connected, their automorphisms have the same
action on a protected logical qubit.

\begin{table}
  \vskip 0.4mm
  \caption{\label{tbl:F_logicals}Protected logical algebra for a localized
  $F$-fermion.}
  \begin{ruledtabular}
    \begin{tabular}{ccc}
      Operator & Anyon Representation & Equivalent Logical \\ 
      \hline \\[-1mm] & & \\[-5mm]
      $\widetilde{\X}_1$ & $\bar\Os[\tsf{rx}\times\tsf{bz}]_v$ &
      $\bar\X_1\bar\Z_4$ \\ 
      $\widetilde{\Z}_1$ & $\bar\Os[\tsf{rz}\times\tsf{by}]_h$ &
      $\bar\Z_1\bar\Z_2\bar\X_3$ \\ 
      $\widetilde{\X}_2$ & $\bar\Os[\tsf{rx}\times\tsf{bz}]_h$ &
      $\bar\Z_1\bar\X_4$ \\ 
      $\widetilde{\Z}_2$ & $\bar\Os[\tsf{rz}\times\tsf{by}]_v$ &
      $\bar\X_2\bar\Z_3\bar\Z_4$
    \end{tabular}
  \end{ruledtabular}
\end{table}

\begin{table}
  \caption{\label{tbl:Fbar_logicals}Protected logical algebra for a localized
  $F'$-fermion.}
  \begin{ruledtabular}
    \begin{tabular}{ccc}
      Operator & Anyon Representation & Equivalent Logical \\ 
      \hline \\[-1mm] & & \\[-5mm]
      $\widetilde{\X}_1$ & $\bar\Os[\tsf{rz}\times\tsf{bx}]_v$ &
      $\bar\X_2\bar\Z_3$ \\ 
      $\widetilde{\Z}_1$ & $\bar\Os[\tsf{rx}\times\tsf{by}]_h$ &
      $\bar\Z_1\bar\X_3\bar\X_4$ \\ 
      $\widetilde{\X}_2$ & $\bar\Os[\tsf{rz}\times\tsf{bx}]_h$ &
      $\bar\Z_2\bar\X_3$ \\ 
      $\widetilde{\Z}_2$ & $\bar\Os[\tsf{rx}\times\tsf{by}]_v$ &
      $\bar\X_1\bar\X_2\bar\Z_4$
    \end{tabular}
  \end{ruledtabular}
\end{table}

\subsection{\label{sec:mcomp}\texorpdfstring{$m$}{m}-Component Disorder Models}
We have so far considered only pairs of FETs that arise from $1$-component
disorder models. That is, only one stage of the measurement sequence was
permitted to have missing measurements. Although this revealed several important
insights into the structure of automorphisms in the DA color code, it is
ultimately a largely unphysical assumption. We now generalize our results to
so-called ``$m$-component disorder models'' where we introduce $m$ parameters
$p_1,\ldots,p_m \in [0,1]$ that determine the probability of independently
measuring links in $m$ components of the measurement sequence. This allows for
more realistic models where more---or all---of the measurement may be randomly
included or excluded. For example, a $2$-component disorder model with two
independent disordered stages and probabilities $p_1$ and $p_2$ is
\begin{eqnarray}
  \widetilde{CC} && \; \rightarrow \;
  \magic{1}{}{}{}{}{}{2}{}{} \; \rightarrow \;
  \magic{}{}{}{}{2}{1}{}{}{} \; \rightarrow \;
  \underbrace{\magic{}{}{}{}{}{}{}{1}{}}_{\displaystyle p_1} \nonumber \\ 
  && \rightarrow \;
  \magic{1}{}{}{}{}{}{2}{}{} \; \rightarrow  \;
  \underbrace{\magic{}{}{}{}{2}{}{}{}{}}_{\displaystyle p_2} \; \rightarrow \;
  \widetilde{CC}.
  \label{eq:ex1}
\end{eqnarray}
We can associate such a model with a $2$-dimensional parameter space $[0,1]^{2}$
indexed by vectors $\mathbf p = (p_1, p_2)$. The four corners of this parameter
space have measurement sequences that enact four different automorphisms:
\begin{equation}
  \begin{tabular}{l||c|c}
    \diagbox{$p_2$}{$p_1$} & $0$ & $1$ \\ 
    \hline \hline
    $0$ & $\id$ & $(rx)(gy)(bz)$ \\ 
    \hline
    $1$ & $(rz)(gx)(by)$ & $(rgb)(xzy)$
  \end{tabular}
  \label{eq:ex1_automorphisms}
\end{equation}
Additional examples of $2$-component disorder models are given in
Appendix~\ref{app:examples}. 

What $m$-component disorder models do not contain any measurement sequences with
irreversible condensations? Or, conversely, what models must contain
irreversible condensations? We saw in Section~\ref{sec:connected} that two FETs
with automorphisms of different parity on $S_3\times S_3$ cannot be connected.
If an $m$-component disorder model realizes two such automorphisms, then corners
of the parameter space must host a measurement sequence with irreversible
condensations (otherwise, we could could construct several $1$-component
disorder models using pairs of FETs in adjacent corners that connect the two
different-parity FETs). For example, consider the following:
\begin{eqnarray}
  \widetilde{CC} \; && \; \rightarrow 
  \magic{1}{}{}{}{}{}{2}{}{} \; \rightarrow \;
  \magic{}{}{}{}{2}{}{}{}{} \;\rightarrow \;
  \underbrace{\magic{}{}{}{}{1}{}{}{}{}}_{\displaystyle p_1} \nonumber \\ 
  && \; \rightarrow \; \underbrace{\magic{2}{}{}{}{}{}{}{}{}}_{\displaystyle
  p_2} \; \rightarrow \; \widetilde{CC}.
  \label{eq:diffparity}
\end{eqnarray}
The corners of the $2$-dimensional parameter space are furnished by
\begin{equation}
  \begin{tabular}{l||c|c}
    \diagbox{$p_2$}{$p_1$} & $0$ & $1$ \\ 
    \hline \hline
    $0$ & $(rz)(gx)(by)$ & IrrP \\ 
    \hline
    $1$ & IrrP & $(rybz)(gx)$
  \end{tabular}
\end{equation}
containing two FETs with automorphisms of different parity on $S_3\times S_3$,
and two phases with interlayer irreversible condensations indicated by ``IrrP''.
Such models, when tuned near the critical lines of $p_i \sim p_c$, will feature
semi-punctures that irreversibly remove logical information from the system. An
$m$-component disorder model that realizes automorphisms with differing
$S_3\times S_3$ parities will therefore not having protected logical subspaces.
If we take every stage in a measurement sequence to be disordered, then this
tells us that if one of the automorphisms has odd-parity on $S_3\times S_3$,
there must be irreversible condensations somewhere in the parameter space. This
is because if every stage is disordered, then the trivial FET with
$\widetilde{CC}\rightarrow\widetilde{CC}$ exists in the parameter space, and
$\id$ has even parity.

We can also consider a stronger criterion: what disorder models are possible
that retain a nonzero-dimensional protected logical subspace? From
Section~\ref{sec:logically}, two FETs that are logically-connected must satisfy
Eq.~\eqref{eq:m_even} or Eq.~\eqref{eq:m_odd}. If these FETs populate corners of
the $m$-dimensional parameter space for an $m$-component disorder model, this
now enables not only $\tau_{BA} \in \mathcal C\{(c\sigma)(c\sigma)(c\sigma)\}$
domain walls, but by Eq.~\eqref{eq:m_even} also $\tau_{CA} \in \mathcal
C\{(ccc)(\sigma\sigma\sigma)\}$ when $p_i\notin\{0,1\}$ for at least two
coordinates. These $\tau_{CA}$ boundaries have $\log_2 \mathcal D^2=2$; the
nontrivial anyons that localize are the three fermions from one of the $F$ or
$F'$ groups, and the protected logical algebra is again given by either
Table~\ref{tbl:F_logicals} or \ref{tbl:Fbar_logicals}. If there does not exist a
consistent logically-protected subspace, then there must exist $\tau$ in other
conjugacy classes with $\log_2 \mathcal D^2 > 2$; anyons beyond the three
fermions in $F$ or $F'$ are measured and $\text{IMS}=0$. This is consistent with
$\mathcal C\{(c\sigma)(c\sigma)(c\sigma)\}$ and $\mathcal
C\{(ccc)(\sigma\sigma\sigma)\}$ being the sole (nontrivial) conjugacy classes
with automorphisms that have invariant mutual-semion pairs. In order for an
$m$-component disorder model to have a consistent logically-protected subspace
at all values of its $m$-dimensional parameter space, we require that:
\begin{enumerate}[label=(\arabic*)]
  \item all corners of the hypercube in parameter space are FETs consisting of
    reversible sequences of condensations;
  \item the enacted automorphisms have even parity on the subgroup $S_3\times
    S_3$; and
  \item all pairs of FETs satisfy Eq.~\eqref{eq:m_even} if their locations
    differ by an even Manhattan distance,\footnote{The Manhattan distance is the
    sum of the component-wise (absolute) differences between two $\mathbf p$
  vectors.} or Eq.~\eqref{eq:m_odd} if an odd Manhattan distance.
\end{enumerate}

An immediate---and intuitive---consequence of these criteria is that if every
anyon condensation in a measurement sequence is disordered, then if there exists
a protected logical subspace, those logical operators do not evolve under any of
the involved automorphisms. This is because the trivial automorphism $\id$ is in
the parameter space, and all automorphisms act equivalently on the protected
subspace. It is thus not possible to construct a DA color code that implements
an automorphism yielding a nontrivial gate while being completely protected
against missing measurements. As we show in Appendix~\ref{app:critical},
however, in the thermodynamic limit logical information is  measured only once
$p$ approaches a critical point, dependent on the particular disorder model. For
$1$-component models, we get $p_c\sim 0.346\ldots$ consistent with the
universality class of bond percolation on a triangular lattice. This allows for
some flexibility in the construction of disorder-resistant DA color codes.

Analyzing missing measurements via competing automorphisms reveals key structure
about the effects of disorder in the DA color code. Given a sequence of anyon
condensations (or link measurements) for an FET, by identifying the
automorphisms created when stages are omitted, one can immediately read off
whether logical qubits are measured or protected due to the disorder. We have
shown necessary conditions for disorder models to support a $2$-qubit logical
subspace that is entirely immune to the disorder. In these cases, at any point
in the $m$-dimensional parameter space, the time-evolution of protected logical
operators will be independent of the disorder realizations. These results
therefore may help determine---or rule-out---feasible implementations and
measurement sequences for disorder-resistant DA color codes.

\section{\label{sec:conclusion}Conclusions and Outlook}

Spatiotemporally heterogeneous domain walls naturally introduce disorder to
measurement-induced Floquet-enriched topological orders. We analyzed the
evolution and purification dynamics of the degenerate codespace amid this
heterogeneity---or ``competing automorphisms''. We showed that this behavior is
agnostic to microscopic details and is directly determined by TQFT features:
anyon braiding and fusion properties (the modular data of the Abelian TO) and
the properties of the involved automorphisms. 

Interpreting these systems as topological stabilizer codes, the number of
independent nontrivial anyons that localize at the domain wall boundaries
determines the number of nontrivial logical operators that may undergo
measurement during one period of competition. For a $\mathbb{Z}_2$-based TO,
this equals $\log_2\mathcal D^2$, where $\mathcal D$ is the quantum dimension of
the twist associated to the transition map between the neighboring
automorphisms. In these systems, the number of mutual-semions that are invariant
under this transition map, IMS, indicates the number of logical qubits that are
protected from any such logical measurement over one period. For a code with $k$
logical qubits, these properties satisfy $\log_2\mathcal D^2 + \text{IMS} \leq
k$. For a $CC$ topological order, Table~\ref{tbl:conjugacyclasses} shows these
values for all conjugacy classes of its automorphism group $\text{Aut}[CC]$.

This new understanding of disordered FETs enables us to readily discern the
effects that missing-measurement noise models or perturbations to measurement
sequences have on dynamical codes such as the DA color code
\cite{davydovaQuantum2024}. The regions subjected to missing measurements in a
given random realization are precisely the heterogeneous temporal domain walls
in an FET with competing automorphisms. $\log_2\mathcal D^2$ and IMS are thus
key metrics that characterize the ability for a given noise model to result in
logical measurements and the number of  logical qubits that remain protected.

We first established results for models with only one measurement stage
disordered, before generalizing to multiple or $m$-component disorder models. We
argued that it is not possible to construct a measurement sequence where every
stage is disordered while simultaneously maintaining a consistent protected
logical subspace and nontrivial effect of automorphisms on logical operators in
this subspace. Implementations of the DA color code must therefore take into
account the effects of missing measurements. We show, nevertheless, in
Appendix~\ref{app:critical}, that even where a logical subspace is not immune
across the entire parameter space, if the noise level is below a critical value
then information can still be protected. In this scenario, it would be
beneficial for such a code to not result in irreversible condensations in its
measurement sequence. We showed that this is only possible throughout the entire
parameter space if the enacted automorphisms have even parity on the subgroup
$S_3\times S_3$. In practice, the occurrence of irreversible condensations may
be detected, and the realized measurement sequences can be post-selected by
discarding those that result in a loss of encoded information.

Our perspective of competing automorphisms has allowed us to chart the topology
of the parameter space of DA color code FETs. For example, two FETs, $A$ and $B$
with automorphisms $\varphi_A$ and $\varphi_B$, can compete using a disorder
model with one random measurement stage (a ``$1$-component disorder model'') if
and only if the transition map $\varphi_B\varphi_A^{-1}$ is in the conjugacy
class $\mathcal C\{(c\sigma)(c\sigma)(c\sigma)\}$ of automorphisms that reflect
the anyons of the color code magic square about a diagonal mirror line. We
presented additional conditions that restrict the ability for two competing
automorphisms to support a consistent nonzero-dimensional logical subspace that
remains unmeasured over multiple periods. This allows one to better understand
the behavior of a given measurement sequence, and how it may be modified or
corrupted by possible missing measurements.

Important open questions remain about disorder in the DA color code, such as the
interplay between competing automorphisms and open boundaries or lattice defects
\cite{ellisonFloquet2023, kesselringBoundaries2018}. It remains to be seen how
the effectiveness of the code's error-correction capabilities (such as the
existence of a threshold, fault-tolerance, or decoders) is affected by competing
automorphisms. Future work should also consider the effect of more general
disorder models, such as with weak measurements \cite{zhuQubit2023,
zhuNishimoris2023}, interspersed random unitaries
\cite{lavasaniMeasurementinduced2021}, coherent errors
\cite{vennCoherentError2023}, or single-qubit measurements
\cite{gullansDynamical2020, botzungRobustness2023, vuStable2024}. We expect our
perspective of competing automorphisms to be potentially useful in any disorder
model that realizes spatiotemporally heterogeneous domain walls.

Moving beyond the DA color code, this work is a first step towards understanding
general dynamical TOs that can support multiple automorphisms. Although we
focused on FETs on a $2$-torus, our results are readily generalizable to other
manifolds by considering their noncontractible cycles. Future works could
investigate other TOs, including microscopic models for those with Abelian
anyons beyond mutual-semions, and general features of competing automorphisms in
non-Abelian anyon theories \cite{ellisonPauli2023, zhaoNonabelian2024} or
fracton Floquet phases \cite{sullivanFloquet2023, zhangXcube2023}. It would be
also interesting to study what competing automorphisms can reveal about FETs
evolving under unitary dynamics instead of
measurements~\cite{potterDynamically2017, poRadical2017, sullivanFloquet2023}. 

\acknowledgments{This research was supported by the Gates Cambridge Trust and by
EPSRC grant EP/V062654/1.}

\appendix

\counterwithin{table}{section}
\counterwithin{figure}{section}

\section{\label{app:additional}Additional Background Material}
\subsection{\label{app:fermions}Color Code Fermions}

In the color code, there are $6$ fermions. These can be written as
a unique fusion product of three mutual-semions:
\begin{equation}
  \begin{aligned}
    \tsf{ry}\times\tsf{bx}\times\tsf{gz} \qquad & \qquad
    \tsf{ry}\times\tsf{bz}\times\tsf{gx} \\
    \tsf{bz}\times\tsf{gy}\times\tsf{rx} \qquad & \qquad
    \tsf{bx}\times\tsf{gy}\times\tsf{rz} \\ 
    \underbrace{\tsf{gx}\times\tsf{rz}\times\tsf{by}}_{\displaystyle F} \qquad &
    \qquad \underbrace{\tsf{gz}\times\tsf{rx}\times\tsf{by}}_{\displaystyle F'}
  \end{aligned}
\end{equation}
\noindent The fermions form two groups, $F$ and $F'$. Fermions within $F$ are
mutual-semions with fermions from $F$ and braid trivially with those from $F'$,
and vice versa. These products can also be summarized in the fermion magic
square \cite{davydovaQuantum2024}:
\begin{equation}
  F\Bigg\{
    \underbrace{\magicsf{ry}{bx}{gz}{bz}{gy}{rx}{gx}{rz}{by}}_{\displaystyle F'}
    \label{eq:fermionsquare}
\end{equation}
such that the product of the three anyons in a row or column give fermions in
$F$ and $F'$ respectively. In contract to the magic square [cf.
Eq.~\eqref{eq:magicsq}] where anyons in the same row or column braid trivially,
anyons in the same row or column of the fermion magic square are mutual-semions.
Anyons not sharing a row or column braid trivially. These fermions can also be
(non-uniquely) formed from the fusion of just two mutual-semions:
\begin{equation}
  \begin{aligned}
    \tsf{gx}\times\tsf{bz} \equiv \tsf{gy}\times\tsf{rz} \equiv
    \tsf{rx}\times\tsf{by} \\
    \tsf{bx}\times\tsf{rz} \equiv \tsf{by}\times\tsf{gz} \equiv 
    \tsf{gx}\times\tsf{ry} \\
    \underbrace{\tsf{rx}\times\tsf{gz} \equiv \tsf{ry}\times\tsf{bz} \equiv 
    \tsf{bx}\times\tsf{gy}}_{\displaystyle F} \\[12pt]
\tsf{bx}\times\tsf{gz} \equiv \tsf{by}\times\tsf{rz} \equiv 
    \tsf{rx}\times\tsf{gy} \\ 
    \tsf{rx}\times\tsf{bz} \equiv \tsf{ry}\times\tsf{gz} \equiv 
    \tsf{gx}\times\tsf{by} \\ 
    \underbrace{\tsf{gx}\times\tsf{rz} \equiv \tsf{gy}\times\tsf{bz} \equiv 
    \tsf{bx}\times\tsf{ry}}_{\displaystyle F'}
  \end{aligned}
\end{equation}

In the main text, we stated a lemma about whether the groups $F$ and $F'$ are
closed under automorphisms. We provide a proof of this lemma now:
\begin{proof}[Proof of Lemma~\ref{lemma:fermions}]
  The automorphism that is trivial on $S_3\times S_3$ and nontrivial on $S_2$ is
  the color-flavor reflection $(rx)(gy)(bz)$. On the fermion magic square
  [Eq.~\eqref{eq:fermionsquare}], this automorphism acts as a reflection about a
  horizontal mirror line through $\tsf{rx}$, $\tsf{gy}$, and $\tsf{bz}$; this
  reflection does not change the fermion group. It is therefore sufficient to
  examine the action only on the $S_3\times S_3$ component of an automorphism.
  We first consider the simplest such nontrivial automorphisms, the $2$-cycles
  $(cc)$ or $(\sigma\sigma)$. These act as a reflection on the fermion magic
  square about diagonal or antidiagonal mirror lines (considering the square on
  a $2$-torus). A fermion in $F$ thus maps to one in $F'$ and vice versa as the
  rows and columns are interchanged. Any general even-parity (odd-parity)
  permutation is an even (odd) product of $2$-cycles, swapping $F$ and $F'$ an
  even (odd) number of times. Therefore, the even-parity $S_3\times S_3$
  automorphisms do not swap $F$ and $F'$, while odd-parity automorphisms do.
\end{proof}

We also stated a lemma on the localized anyons for parallel reflections in
$\mathcal C\{(c\sigma)(c\sigma)(c\sigma)\}$. We provide the proof here: 
\begin{proof}[Proof of Lemma~\ref{lemma:reflections}]
  If two automorphisms $\tau_1, \tau_2 \in \mathcal
  C\{(c\sigma)(c\sigma)(c\sigma)\}$ are equal, then their localized anyons are
  the same. If they are unequal, consider the case where we can interpret them
  geometrically as reflections about inequivalent but parallel lines of the
  magic square (on a $2$-torus). Their mirror lines do not intersect, and their
  constituent $2$-cycles $(c\sigma)$ are all different (since the mirror line
  for $(c_1\sigma_1)(c_2\sigma_2)(c_3\sigma_3)$ is precisely the line through
  anyons $\tsf{c}_1\sigma_1$, $\tsf{c}_2\sigma_2$ and $\tsf{c}_3\sigma_3$).
  Moreover, the product of these three
  anyons---$\tsf{c}_1\sigma_1\times\tsf{c}_2\sigma_2\times\tsf{c}_3\sigma_3$---is
  precisely the fermion that localizes at the twist for the automorphism.
  Therefore, the two fermions that localize at $\tau_1$ and $\tau_2$
  respectively occupy distinct rows (or columns) of the fermion magic square,
  and are thus in the same fermion group. On the other hand, if the
  automorphisms have perpendicular mirror lines, then their $c\sigma$ labels
  intersect at some point on the magic square such that the localized fermions
  have one shared label, and thus are in different fermion groups. 
\end{proof}

\subsection{\label{app:theory}Group Theory}
\subsubsection{\label{app:grouptheory}Group Theory Essentials}
In this section we briefly outline some ideas from group theory that are
relevant for our discussions of the automorphism group in
Appendix~\ref{app:autotheory}. 

We first outline general group properties, starting with group products. For a
group $G$, a subgroup $N \triangleleft G$ is normal if and only if $gng^{-1} \in
N$ for all $g\in G$ and $n\in N$. That is, elements of $N$ are invariant under
conjugation by all elements of $G$, or equivalently the left and right cosets
$gN$ and $Ng$ are equal for all $g \in G$. When $G$ is a semidirect product,
written as $G = N \rtimes H$, and where $N$ is normal in $G$ but $H$ may not be,
then for every $g \in G$, there are unique $n \in N$ and $h \in H$ such that $g
= nh$. 

For a group $G$ and two elements $a,b \in G$, if $b = gag^{-1}$ for some $g \in
G$, then $a$ and $b$ are conjugate. Conjugacy is an equivalence relation that
partitions $G$ into conjugacy classes, denoted as
\begin{equation} \mathcal C_a = \{ gag^{-1} \, | g \in G \} \end{equation}
for some representative $a \in G$. All elements belonging to the same conjugacy
class have the same order, the minimal $k$ such that $a^k = \id$.

We now consider specifically the permutation group, $S_n$, which is the group of
re-orderings of a set of $n$ elements. We write these permutations using cycle
notation. For example, labelling the $n$ elements as $a,b,\ldots,n$, the
permutation $(ade)(fg)$ indicates the map $a\mapsto d \mapsto e \mapsto a$,
$f\mapsto g\mapsto f$, with the other elements unchanged. A cycle is a closed
mapping, such as $(ade)$ or $(fg)$. The order (or length) of a cycle can be read
off as the number of elements listed; a cycle of length $k$ is called a
$k$-cycle. $(ade)$ is a $3$-cycle, for example. A cycle is defined as having
even (odd) parity if it can be written as an even (odd) number of $2$-cycles.
For example, $(ade) = (ad)(de)$ and therefore is even parity. Equivalently, a
$k$-cycle is even (odd) if $k-1$ is even (odd). A permutation can be written in
multiple ways: $(ade)(fg)$ is the same as $(gf)(dea)$. However, the number and
lengths of disjoint cycles forming a permutation are a fixed property of that
permutation \cite{humphreysCourse1996}. The ``cycle type'' of a permutation is
written in bracket notation as 
\begin{equation}
  [1^{\alpha_1}2^{\alpha_2}\cdots n^{\alpha_n}]
\end{equation}
where $\alpha_k$ is its number of disjoint $k$-cycles. For example, $(ade)(fg)$
has cycle type $[2^13^1]$ (with $\alpha_k = 0$ omitted for brevity). Notably,
the conjugacy classes of a permutation group are characterized by its elements
all having the same cycle type \cite{humphreysCourse1996}.

\subsubsection{\label{app:autotheory}Automorphism Group Essentials}
In this section we apply the previous group-theory ideas to the automorphism
group of the color code, $\text{Aut}[CC]$. Concretely, $\text{Aut}[CC]$ is the
group of permutations of the anyons of the color code that preserve the
relationships and structure of the anyons. That is, mutual-statistics,
self-statistics, and fusion rules must remain equivalent after applying $\varphi
\in \text{Aut}[CC]$ to all anyons. We can represent any automorphism by a
relabelling of the $6$ color and flavor labels, $r,g,b,x,y,z$, since all
$c$-colored or $\sigma$-flavored anyons must transform equally in order to
maintain their mutual statistics and fusion rules. $\text{Aut}[CC]$ is therefore
a subgroup of $S_6$. We represent its elements using cycle notation, such as
$(rbg)(xy)$. When color and flavor are interchanged, we have cycles such as
$(rxgy)(bz)$; although the standalone map $r \mapsto x$ is ill-defined, the
construction of each anyon in terms of both a color and a flavor ensures that as
long as all the cycles either alternate colors and flavors, or have disjoint
color-only and flavor-only cycles, the anyon mapping is valid. This cycle, for
example, maps $\tsf{rx} \mapsto \tsf{gx}$ as $r\mapsto x$, $x \mapsto g$. A
cycle such as $(rxg)$ is not in $\text{Aut}[CC]$ for this reason.

As with the permutation groups, $\text{Aut}[CC]$ can be partitioned into
conjugacy classes; Table~\ref{tbl:conjugacyclasses} lists the $9$ conjugacy
classes and their cycle types. These are subsets of the $S_6$ conjugacy classes
and are generated by and closed under conjugation with elements from
$\text{Aut}[CC]$.

Automorphisms in $\text{Aut}[CC]$ can be identified by the decomposition into 
subgroups $(S_3\times S_3) \rtimes S_2$, representing the $S_3$ group of $3!$
color (magic square row) permutations, the $S_3$ group of $3!$ flavor (magic
square column) permutations, and the $S_2$ group of $2$ color-flavor exchanges
\cite{davydovaQuantum2024, kesselringAnyon2024}. Specifically, the $S_2$
corresponds to a reflection about the mirror line through the
$\tsf{rx}-\tsf{gy}-\tsf{bz}$ diagonal. In cycle notation, the nontrivial element
of $S_2$ [with trivial $(S_3\times S_3)$ contribution] is $(rx)(gy)(bz)$.
$(S_3\times S_3)$ is closed under conjugation and hence is a normal subgroup.
$S_2$, on the other hand, is not. For example, \begin{equation} (rb)
\cdot (rx)(gy)(bz) \cdot (rb)^{-1} = (rz)(gy)(bx). \end{equation} Hence, we use
the semidirect product $\rtimes$. To identify whether an element has nontrivial
$S_2$ contribution, we note that for this there must be alternating color-flavor
labels in the cycle notation: $(rbg)(xy)$ is trivial on $S_2$ while $(rxbygz)$
is not. 

An important concept in Section~\ref{sec:connected} is the parity of an
automorphism on the subgroup $(S_3\times S_3)$. To identify this, we first
trivialize any $S_2$ contribution by composing the automorphism with
$(rx)(gy)(bz)$ if it has alternating color-flavor labels. We then multiply the
parity of the resulting automorphism's disjoint cycles. For example,
$(ry)(gz)(bx)$ involves a color-flavor reflection and therefore we modify it by 
\begin{equation} (rx)(gy)(bz)\cdot (ry)(gz)(bx) = (rgb)(xzy),  \end{equation} 
which results in a $[3^2]$ automorphism with net $\text{even}\times\text{even} =
\text{even}$ parity.\footnote{Equivalently, it can be written as an even number
of $2$-cycles, $(rg)(gb)(xz)(zy)$} Table~\ref{tbl:conjugacyclasses} lists these
parities for all conjugacy classes of $\text{Aut}[CC]$.

\subsection{\label{app:twists}Localized Anyons}
In Section~\ref{sec:competing}, we introduced the notion of localized anyons
around the boundaries or twists of domain walls. Lemma~\ref{lemma:retained}
was used to determine the presence of operators that are protected during an
evolution under multiple temporal domain walls. We prove that lemma here: 
\begin{proof}[Proof of Lemma~\ref{lemma:retained}]
  We first show that if $\tsf b = \tau(\tsf b)$ for some automorphism $\tau$ and
  anyon $\tsf b$, then $\tsf b$ braids trivially with all $\tsf c$ that localize
  at $\tau$. For any such anyon $\tsf c$ that localizes, there exists an $\tsf
  a$ such that $\tsf c = \tsf a \times \tau(\bar{\tsf a})$. Now, the composite
  anyon $\tsf a \times \bar{\tsf a}$ is equivalent to the vacuum $\tsf 1$, and
  therefore it braids trivially with $\tsf b$. We can write this as 
  \begin{equation} 
    \exp(2i\theta_{\tsf b, \tsf a}) \exp(2i\theta_{\tsf b, \bar{\tsf a}}) = 1 
  \end{equation}
  where $\exp(2i\theta_{\tsf a, \tsf b})$ encodes the phase factor accumulated
  when clockwise encircling an $\tsf a$ with $\tsf b$ (or vice versa):
    \begin{center}
      \vskip -0.5cm
      \def\svgwidth{.9\columnwidth}
      \begingroup \makeatletter \providecommand\color[2][]{\errmessage{(Inkscape) Color is used for the text in Inkscape, but the package 'color.sty' is not loaded}\renewcommand\color[2][]{}}\providecommand\transparent[1]{\errmessage{(Inkscape) Transparency is used (non-zero) for the text in Inkscape, but the package 'transparent.sty' is not loaded}\renewcommand\transparent[1]{}}\providecommand\rotatebox[2]{#2}\newcommand*\fsize{\dimexpr\f@size pt\relax}\newcommand*\lineheight[1]{\fontsize{\fsize}{#1\fsize}\selectfont}\ifx\svgwidth\undefined \setlength{\unitlength}{600bp}\ifx\svgscale\undefined \relax \else \setlength{\unitlength}{\unitlength * \real{\svgscale}}\fi \else \setlength{\unitlength}{\svgwidth}\fi \global\let\svgwidth\undefined \global\let\svgscale\undefined \makeatother \begin{equation}
  \begin{picture}(1,0.5)\lineheight{1}\setlength\tabcolsep{0pt}\put(0,0){\includegraphics[width=\unitlength,page=1]{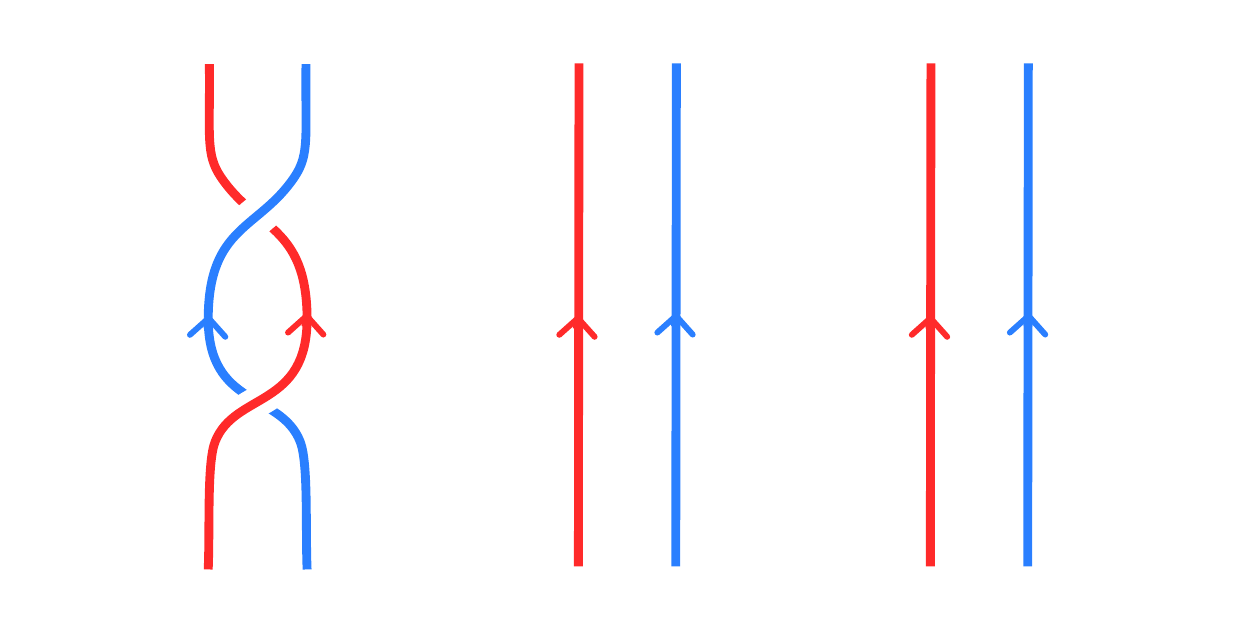}}\put(0.03954377,0.28252446){\color[rgb]{0,0,0}\makebox(0,0)[lt]{\lineheight{1.25}\smash{\begin{tabular}[t]{l}$t$\end{tabular}}}}\put(0,0){\includegraphics[width=\unitlength,page=2]{braiding.pdf}}\put(0.15634053,0.45641733){\color[rgb]{0,0,0}\makebox(0,0)[lt]{\lineheight{1.25}\smash{\begin{tabular}[t]{l}$\tsf a$\end{tabular}}}}\put(0.15634053,0.01056632){\color[rgb]{0,0,0}\makebox(0,0)[lt]{\lineheight{1.25}\smash{\begin{tabular}[t]{l}$\tsf a$\end{tabular}}}}\put(0.23628906,0.01056632){\color[rgb]{0,0,0}\makebox(0,0)[lt]{\lineheight{1.25}\smash{\begin{tabular}[t]{l}$\tsf b$\end{tabular}}}}\put(0.23628906,0.45610707){\color[rgb]{0,0,0}\makebox(0,0)[lt]{\lineheight{1.25}\smash{\begin{tabular}[t]{l}$\tsf b$\end{tabular}}}}\put(0.29057524,0.23544857){\color[rgb]{0,0,0}\makebox(0,0)[lt]{\lineheight{1.25}\smash{\begin{tabular}[t]{l}$=e^{2i\theta_{\tsf b, \tsf a}}$\end{tabular}}}}\put(0.56963778,0.23544857){\color[rgb]{0,0,0}\makebox(0,0)[lt]{\lineheight{1.25}\smash{\begin{tabular}[t]{l}$=e^{2i\theta_{\tsf a, \tsf b}}$\end{tabular}}}}\put(0.45355132,0.45641733){\color[rgb]{0,0,0}\makebox(0,0)[lt]{\lineheight{1.25}\smash{\begin{tabular}[t]{l}$\tsf a$\end{tabular}}}}\put(0.45355132,0.01056632){\color[rgb]{0,0,0}\makebox(0,0)[lt]{\lineheight{1.25}\smash{\begin{tabular}[t]{l}$\tsf a$\end{tabular}}}}\put(0.53099984,0.01056632){\color[rgb]{0,0,0}\makebox(0,0)[lt]{\lineheight{1.25}\smash{\begin{tabular}[t]{l}$\tsf b$\end{tabular}}}}\put(0.53099984,0.45610707){\color[rgb]{0,0,0}\makebox(0,0)[lt]{\lineheight{1.25}\smash{\begin{tabular}[t]{l}$\tsf b$\end{tabular}}}}\put(0.73461702,0.45641733){\color[rgb]{0,0,0}\makebox(0,0)[lt]{\lineheight{1.25}\smash{\begin{tabular}[t]{l}$\tsf a$\end{tabular}}}}\put(0.73461702,0.01056632){\color[rgb]{0,0,0}\makebox(0,0)[lt]{\lineheight{1.25}\smash{\begin{tabular}[t]{l}$\tsf a$\end{tabular}}}}\put(0.81456555,0.01056632){\color[rgb]{0,0,0}\makebox(0,0)[lt]{\lineheight{1.25}\smash{\begin{tabular}[t]{l}$\tsf b$\end{tabular}}}}\put(0.81456555,0.45610707){\color[rgb]{0,0,0}\makebox(0,0)[lt]{\lineheight{1.25}\smash{\begin{tabular}[t]{l}$\tsf b$\end{tabular}}}}\end{picture}\end{equation}
\endgroup      \end{center}
  Since $\tau$ is an automorphism, it preserves these mutual statistics and so
  \begin{align}
    1 &= \exp(2i\theta_{\tsf b, \tsf a}) \exp(2i\theta_{\tau(\tsf
    b),\tau(\bar{\tsf a})}) \notag \\ 
    &= \exp(2i\theta_{\tsf b, \tsf a}) \exp(2i\theta_{\tsf b, \tau(\bar{\tsf
    a})}) \notag \\
    &= \exp(2i\theta_{\tsf b, \tsf a \times \tau(\bar{\tsf a})}) \notag \\
    &= \exp(2i\theta_{\tsf b, \tsf c})
  \end{align}
  where we used $\tsf b = \tau(\tsf b)$ and that $2(\theta_{\tsf b, \tsf a} +
  \theta_{\tsf b, \tau(\bar{\tsf a})})$ is the phase accumulated by $\tsf b$
  encircling $\tsf a \times \tau(\bar{\tsf a}) = \tsf c$. Hence, invariant
  anyons of $\tau$ braid trivially with all localized anyons of $\tau$.

  We next prove the reverse direction: any anyon $\tsf b$ that braids trivially
  with all localized anyons must be invariant. For any anyon $\tsf a$, the anyon
  $\tsf c = \tsf a \times \tau(\bar{\tsf a})$ localizes at automorphism $\tau$,
  and therefore $\tsf b$ braids trivially with $\tsf c$: \begin{equation}
  \exp(2i\theta_{\tsf b \tsf c}) = 1. \end{equation} We again use the
  statistics-preserving property of automorphisms to write 
  \begin{align}
    1 &= \exp(2i\theta_{\tsf b, \tsf a}) \exp(2i\theta_{\tsf b, \tau(\bar{\tsf
    a})}) \notag \\ 
    &= \exp(2i\theta_{\tsf b, \tsf a}) \exp(2i\theta_{\tau^{-1}(\tsf
    b),\bar{\tsf a}}). 
  \end{align}
  By anyons and anti-anyons having the same mutual-statistics, $\theta_{\tsf p,
  \tsf q} = \theta_{\bar{\tsf p},\bar{\tsf q}}$, we have
  \begin{align}
    1 &= \exp(2i\theta_{\tsf b, \tsf a}) \exp(2i\theta_{\tau^{-1}(\bar{\tsf
    b}),\tsf a}) \notag \\ 
    &= \exp(2i\theta_{\tsf d,\tsf a})
  \end{align}
  where $\tsf d = \tsf b \times \tau^{-1}(\bar{\tsf b})$. This relation holds for
  all anyons $\tsf a$, meaning that there is no anyon encircling action that can
  distinguish $\tsf d$ from the vacuum $\tsf 1$ ($\tsf d$ is ``transparent''
  \cite{simonTopological2023}). For modular anyon theories, such as those formed
  by topological stabilizer codes, braiding is nondegenerate
  \cite{barkeshliSymmetry2019, ellisonPauli2023}, which implies that $\tsf d$ is
  the vacuum. That is, $\tsf b = \tau^{-1}(\tsf b)$ or equivalently, $\tau(\tsf
  b) = \tsf b$ as required.
\end{proof}

\subsection{\label{app:automorphisms}Computing and Creating Automorphisms}

Given a dynamical scheme of the form $\widetilde{CC} \rightarrow TC\boxtimes TC
\rightarrow \cdots \rightarrow TC\boxtimes TC \rightarrow \widetilde{CC}$, we
compute the enacted automorphism by the formula \cite{davydovaQuantum2024} 
\begin{equation}
  \varphi_f[(rx)(gy)(bz)]^{\alpha}[(rz)(gy)(bx)]^\beta 
  \varphi_i^{-1}
  \label{eq:compute}
\end{equation}
where $\varphi_i$ and $\varphi_f$ are the contributions from the $\widetilde{CC}
\rightarrow TC\boxtimes TC$ and $TC\boxtimes TC \rightarrow \widetilde{CC}$
transitions respectively. Table~\ref{tbl:automorphisms} lists these for all
possible reversible transitions. $\alpha$ is the number of reversible
transitions that the first $CC$ layer undergoes. $\beta$ is the number of
reversible transitions that the second $CC$ layer undergoes.

\begin{table*}
  \caption{\label{tbl:automorphisms}All possible isomorphism contributions from
  the $\widetilde{CC} \leftrightarrow TC\boxtimes TC$ reversible transitions of
  a dynamical scheme. Adapted from Ref.~\onlinecite{davydovaQuantum2024}. There
  are two possible $TC\boxtimes TC$ theories for each isomorphism.}
  \begin{ruledtabular}
    \begin{tabular}{cccc}
      Isomorphism & Theories & Isomorphism & Theories \\
      \hline & & & \\
      $\id$ & $\magic{1}{}{}{}{}{}{2}{}{}$, $\magic{}{1}{}{}{}{}{}{2}{}$ &
      $(xy)$ & $\magic{1}{}{}{}{}{}{}{2}{}$, $\magic{}{1}{}{}{}{}{2}{}{}$
      \\[16pt] 
      $(rg)$ & $\magic{}{}{}{1}{}{}{2}{}{}$, $\magic{}{}{}{}{1}{}{}{2}{}$ & 
      $(rg)(xy)$ & $\magic{}{}{}{1}{}{}{}{2}{}$, $\magic{}{}{}{}{1}{}{2}{}{}$
      \\[16pt] 
      $(gb)$ & $\magic{1}{}{}{2}{}{}{}{}{}$, $\magic{}{1}{}{}{2}{}{}{}{}$ & 
      $(gb)(xy)$ & $\magic{1}{}{}{}{2}{}{}{}{}$, $\magic{}{1}{}{2}{}{}{}{}{}$
      \\[16pt] 
      $(rb)$ & $\magic{2}{}{}{}{}{}{1}{}{}$, $\magic{}{2}{}{}{}{}{}{1}{}$ & 
      $(rb)(xy)$ & $\magic{}{2}{}{}{}{}{1}{}{}$, $\magic{2}{}{}{}{}{}{}{1}{}$
      \\[16pt]
      $(rgb)$ & $\magic{2}{}{}{1}{}{}{}{}{}$, $\magic{}{2}{}{}{1}{}{}{}{}$ & 
      $(rgb)(xy)$ & $\magic{}{2}{}{1}{}{}{}{}{}$, $\magic{2}{}{}{}{1}{}{}{}{}$
      \\[16pt] 
      $(rbg)$ & $\magic{}{}{}{2}{}{}{1}{}{}$, $\magic{}{}{}{}{2}{}{}{1}{}$ & 
      $(rbg)(xy)$ & $\magic{}{}{}{}{2}{}{1}{}{}$, $\magic{}{}{}{2}{}{}{}{1}{}$
      \\[16pt]
    \end{tabular}
  \end{ruledtabular}
\end{table*}

For example, consider 
\begin{equation}
  \widetilde{CC} 
  \rightarrow \magic{1}{}{}{2}{}{}{}{}{} \rightarrow
  \magic{}{}{}{}{1}{}{}{}{} \rightarrow \magic{1}{}{}{}{}{}{}{2}{} \rightarrow
  \widetilde{CC}.
\end{equation}
We have that $\varphi_i = (gb)$, $\varphi_f = (xy)$, $\alpha = 2$, and $\beta =
1$. $\alpha$ and $\beta$ are important only mod $2$, since they exponent
$2$-cycles. This gives overall automorphism 
\begin{equation}
  (xy)\cdot\id\cdot(rz)(gy)(bx)\cdot (gb)^{-1} = (rz)(gy)(bx).
\end{equation}

On the other hand, given any automorphism $\varphi \in \text{Aut}[CC]$, it is
possible to create a measurement sequence that realizes it. All $72$
automorphisms and example measurement sequences are given in
Ref.~\onlinecite{davydovaQuantum2024}. Moreover, it is possible to ensure that
this measurement sequence ends with a particular $TC\boxtimes TC$ condensation
prior to returning to $\widetilde{CC}$.

We summarize here the procedure described in
Ref.~\onlinecite{davydovaQuantum2024}. Let $\mathcal B_1,\,\ldots,\,\mathcal
B_n$ be a sequence of $n$ sets of bosons of $CC\boxtimes CC$ that form the
condensations for $n$ child theories $C_1,\,\ldots,\,C_n$. We set $\mathcal B_n
= \{\tsf 1,\, \tsf{rz}_1\tsf{rz}_2,\,
\tsf{gz}_1\tsf{gz}_2,\,\tsf{bz}_1\tsf{bz}_2\}$ such that $C_n = \widetilde{CC}$.
We require that the sequence of condensations is reversible. This allows us to
associate an isomorphism $\lambda$ from $C_1$ to $C_n$. Let $\phi \in
\text{Aut}[CC\boxtimes CC]$ be an automorphism of the parent theory such that
$\phi(\mathcal B_1) = \mathcal B_1$ and $\phi(\mathcal B_n) = \mathcal B_n$.
This also defines an automorphism on $C_1$ and $C_n$; assume we have chosen
$C_1$ such that $\phi$ acts as the trivial automorphism $\id_{1}$ on $C_1$. On
$C_n$ it acts as $\varphi \in \text{Aut}[CC]$. The sequence of reversible
condensations $\phi(\mathcal B_1),\,\ldots,\,\phi(\mathcal B_n)$ now enacts the
isomorphism $\varphi\lambda$ from $C_1$ to $C_n$. The sequence \begin{equation}
  \mathcal B_n ,\, \mathcal B_{n-1},\,\ldots,\, \mathcal B_1, \, \phi(\mathcal
  B_2), \, \ldots, \, \phi(\mathcal B_{n-1}),\, \mathcal B_n \end{equation}
  therefore enacts the automorphism $\varphi\lambda\lambda^{-1} = \varphi$ on
  $C_n = \widetilde{CC}$. This allows us to create any automorphism $\varphi \in
  \text{Aut}[CC]$. 

Because $\lambda$ does not affect the final automorphism, we can arbitrarily 
specify that $\mathcal B_{n-1}$ is any set of bosons as long as it creates a
reversible pair of condensations with $\widetilde{CC}$ (i.e. any of the theories
in Table~\ref{tbl:automorphisms}). If we wish to end with a specific (valid)
$TC\boxtimes TC$ theory with condensed bosons $\mathcal B$, we therefore specify
that $\mathcal B_{n-1} = \phi^{-1}(\mathcal B)$. Note that this reasoning allows
us to specify only the ending or the first $TC\boxtimes TC$ theory, not both. 

\section{\label{app:disordered}Disordered DA Color Code}
\subsection{\label{app:examples}Example Disorder Models}

In this section, we provide more examples of disorder models. 

\paragraph{Trivial-adjacent FETs}
We first consider examples of $1$-component disorder models that realize
connections between the trivial FET (when $p=0$) and each of the $6$
automorphisms in $\mathcal C\{(c\sigma)(c\sigma)(c\sigma)\}$ (when $p=1$):

\begin{equation*}
  (ry)(gx)(bz):\qquad  
  \widetilde{CC} \rightarrow 
  \magic{1}{}{}{}{}{}{2}{}{}\rightarrow 
  \underbrace{\magic{}{}{}{}{1}{}{}{}{}}_{\displaystyle p} \rightarrow
  \widetilde{CC} 
\end{equation*}
\begin{equation*}
  (rz)(gx)(by):\qquad 
  \widetilde{CC} \rightarrow 
  \magic{1}{}{}{}{}{}{2}{}{}\rightarrow 
  \underbrace{\magic{}{}{}{}{2}{}{}{}{}}_{\displaystyle p} \rightarrow
  \widetilde{CC}
\end{equation*}
\begin{equation*}
  (ry)(gz)(bx):\qquad  
  \widetilde{CC} \rightarrow 
  \magic{1}{}{}{2}{}{}{}{}{}\rightarrow 
  \underbrace{\magic{}{}{}{}{}{}{}{1}{}}_{\displaystyle p} \rightarrow
  \widetilde{CC} 
\end{equation*}
\begin{equation*}
  (rx)(gz)(by):\qquad 
  \widetilde{CC} \rightarrow 
  \magic{}{}{}{1}{}{}{2}{}{}\rightarrow 
  \underbrace{\magic{}{2}{}{}{}{}{}{}{}}_{\displaystyle p} \rightarrow
  \widetilde{CC}
\end{equation*}
\begin{eqnarray*}
  (rx)(gy)(bz): \qquad
  \widetilde{CC} \rightarrow &&
  \magic{1}{}{}{}{}{}{2}{}{}\rightarrow 
  \underbrace{\magic{}{}{}{}{1}{}{}{}{}}_{\displaystyle p} \rightarrow \\ 
  && \magic{}{}{}{}{}{}{}{}{1} \rightarrow \magic{1}{}{}{}{}{}{}{}{} \rightarrow
  \widetilde{CC}
\end{eqnarray*}
\begin{eqnarray*}
  (rz)(gy)(bx): \qquad
  \widetilde{CC} \rightarrow &&
  \magic{1}{}{}{}{}{}{2}{}{}\rightarrow 
  \underbrace{\magic{}{}{}{}{2}{}{}{}{}}_{\displaystyle p} \rightarrow \\ 
  && \magic{}{}{2}{}{}{}{}{}{} \rightarrow \magic{}{}{}{}{}{}{2}{}{} \rightarrow
  \widetilde{CC}
\end{eqnarray*}

\paragraph{Example 2.}
Secondly, consider 
\begin{eqnarray*}
  \widetilde{CC} && \; \rightarrow \;
  \magic{1}{}{}{}{}{}{2}{}{} \; \rightarrow \;
  \magic{}{}{}{}{1}{2}{}{}{} \; \rightarrow \;
  \underbrace{\magic{}{}{1}{}{}{}{}{}{}}_{\displaystyle p_1} \nonumber \\ 
  && \rightarrow \;
  \underbrace{\magic{}{}{}{}{}{}{2}{}{}}_{\displaystyle p_2}
  \; \rightarrow \; \magic{}{2}{}{}{}{}{1}{}{} \; \rightarrow \;
  \widetilde{CC}
\end{eqnarray*}
with the corners of the parameter space supporting FETs with automorphisms
\begin{equation*}
  \begin{tabular}{l||c|c}
    \diagbox{$p_2$}{$p_1$} & $0$ & $1$ \\ 
    \hline \hline
    $0$ & $(rb)(xy)$ & $(rygxbz)$ \\ 
    \hline
    $1$ & $(rzbygx)$ & $(xzy)$
  \end{tabular}.
\end{equation*}
\noindent Notably, there does not exist a logically-protected trajectory between
$\mathbf p = (0,0)$ and $\mathbf p = (1,1)$ because
\begin{equation}
  (xzy)\cdot[(rb)(xy)]^{-1} = (rb)(yz) \notin \mathcal
  C\{(ccc)(\sigma\sigma\sigma)\}
\end{equation}
in violation of Eq.~\eqref{eq:m_even}.

\paragraph{Example 3.}
This is an example of a $2$-component disorder model where we have an
irreversible phase in one corner, from an interlayer irreversible condensation.
This differs from Eq.~\eqref{eq:diffparity} in that all other FETs are in the same parity
cluster. 
\begin{eqnarray*}
  \widetilde{CC} && \rightarrow 
  \magic{1}{}{}{}{}{}{2}{}{}\rightarrow \magic{}{}{}{}{2}{}{}{1}{}\rightarrow
  \underbrace{\magic{}{}{}{1}{}{}{}{}{}}_{\displaystyle p_1} \nonumber \\ 
  && \rightarrow 
  \underbrace{\magic{2}{}{}{}{}{}{}{}{}}_{\displaystyle p_2} \rightarrow 
  \widetilde{CC}
\end{eqnarray*}
The corners host FETs with automorphisms
\begin{equation*}
  \begin{tabular}{l||c|c}
    \diagbox{$p_2$}{$p_1$} & $0$ & $1$ \\ 
    \hline \hline
    $0$ & $(rg)(xz)$ & IrrP \\ 
    \hline
    $1$ & $(rygxbz)$ & $(rgb)$
  \end{tabular}.
\end{equation*}

\paragraph{Example 4.}
This measurement sequence demonstrates multiple automorphisms in the even-parity
component; these are connected but cannot be logically-connected.
\begin{eqnarray*}
  \widetilde{CC} && \rightarrow 
  \magic{1}{}{}{}{}{}{2}{}{} \rightarrow 
  \magic{}{2}{}{}{}{}{}{1}{} \rightarrow 
  \underbrace{\magic{}{}{1}{}{}{}{}{}{}}_{\displaystyle p_1} \nonumber \\ 
  && \rightarrow \underbrace{\magic{}{}{}{}{}{2}{}{}{}}_{\displaystyle p_2}
  \rightarrow \magic{}{}{}{1}{}{}{2}{}{} \rightarrow \widetilde{CC}
\end{eqnarray*}
The automorphisms are
\begin{equation*}
  \begin{tabular}{l||c|c}
    \diagbox{$p_2$}{$p_1$} & $0$ & $1$ \\ 
    \hline \hline
    $0$ & $(rg)$ & $(rxgy)(bz)$ \\
    \hline
    $1$ & $(rzgy)(bx)$ & $(rbg)(xz)$
  \end{tabular}.
\end{equation*}

\subsection{\label{app:connected}Connected FETs}
We first provide an example of two measurement sequences that both realize the
same FET and automorphism, $\id$, but are not ``connected'' to each other: one
sequence cannot be made into the other by adding or removing condensation steps
and while maintaining solely reversible transitions. The two sequences are
\begin{align}
    \widetilde{CC} &\rightarrow \magic{1}{}{}{}{}{}{2}{}{} \rightarrow 
    \magic{}{}{}{}{1}{2}{}{}{} \rightarrow \magic{2}{}{}{}{}{}{1}{}{} \notag \\ 
    & \rightarrow 
    \magic{}{}{}{}{2}{1}{}{}{} \rightarrow 
    \magic{1}{}{}{}{}{}{2}{}{} \rightarrow \widetilde{CC}
    \label{eq:id1}
\end{align}
and 
\begin{equation}
    \widetilde{CC} \rightarrow \magic{1}{}{}{}{}{}{2}{}{} \rightarrow \widetilde{CC}.
    \label{eq:id2}
\end{equation}

Let us begin with Eq.~\eqref{eq:id1} and consider for now just the first layer.
We wish to add or remove condensed bosons to manipulate it into the form of
Eq.~\eqref{eq:id2}. The sequence of bosons $\tsf{rx}, \tsf{gy}, \tsf{bx},
\tsf{gz}, \tsf{rx}$ forms a path on the fermion magic square (the path can wrap
around periodic boundary conditions) \cite{davydovaQuantum2024}: 

\begin{equation}
  {\color{gray}
  \magicsf{ry}{b{\tikzmark{2}}x}{g{\tikzmark{3}}z}{bz}{g{\tikzmark{1}}y}{r{\tikzmark{0}}x}{gx}{rz}{by}}
\end{equation}\begin{tikzpicture}[remember picture, overlay, black, line width=2pt]\draw[->] ([yshift=.5ex]pic cs:0) to ([yshift=.5ex]pic cs:1);
  \draw[->] ([yshift=.5ex]pic cs:1) to ([yshift=.5ex]pic cs:2);
  \draw[->] ([yshift=.5ex]pic cs:2) to ([yshift=.5ex]pic cs:3);
  \draw[->] ([yshift=.5ex]pic cs:3) to ([yshift=.5ex]pic cs:0);
\end{tikzpicture}

The requirement for reversible transitions (that adjacent bosons in the sequence
are mutual-semions) necessitates that these arrows only point horizontally or
vertically; we must follow this rule when adding or removing condensed bosons.
As such, we are never able to remove bosons that sit at the corners of the
path---this would result in a diagonal arrow. We may add or remove bosons from
the start or end of the sequence only if the resulting pair of layer-$1$ and
layer-$2$ bosons creates a reversible transition to $\widetilde{CC}$; the
$\tsf{gz}$ boson can never begin or end the sequence. We therefore cannot remove
$\tsf{gz}$ from our sequence. The same argument applies to the sequence of
bosons in the second layer. It is thus impossible to connect the measurement
sequence of Eq.~\eqref{eq:id1} with Eq.~\eqref{eq:id2} while maintaining
reversible transitions.

Furthermore, we can also show that not every length-$m$ adjacency sequence can
be made into an $m$-component disorder model. We discuss here one such example. 

Consider the adjacency sequence $\{A_0,\, \mathbbm 1,\, A_2\}$, where $\mathbbm
1$ is the trivial FET with automorphism $\id$, and $A_0$ and $A_2$ have
automorphisms $\varphi_0$ and $\varphi_2 \in \mathcal
C\{(c\sigma)(c\sigma)(c\sigma)\}$ respectively [as required by the separation
condition, Eq.~\eqref{eq:separation}]. We wish to show that for some choice of
$A_0$ and $A_2$ there is no $2$-component disorder model that realizes $A_0$,
$\mathbbm 1$, and $A_2$ in three of its parameter-space corners. Equivalently,
there is no measurement sequence $\widetilde{CC} \rightarrow (TC\boxtimes TC)_1
\rightarrow \cdots \rightarrow (TC\boxtimes TC)_k \rightarrow \widetilde{CC}$
that realizes the $\id$ automorphism and that can form both a $1$-component
disorder model with $A_0$ and a $1$-component disorder model with $A_2$. 

Assume that we have chosen some measurement sequence for $\mathbbm 1$.
Eq.~\eqref{eq:compute} tells us that 
\begin{equation}
    \id = \varphi_f[(rx)(gy)(bz)]^{\alpha}[(rz)(gy)(bx)]^\beta 
  \varphi_i^{-1}.
\end{equation}
We first simplify this by noting that for all $\varphi_i$, $\varphi_f$
isomorphism contributions listed in Table~\ref{tbl:automorphisms}, none contain
permutations affecting the $z$ flavor label. This enforces that $\alpha = \beta
= 0$. Then, $\id = \varphi_f \varphi_i^{-1}$ gives $\varphi_f = \varphi_i$. This
means that the first condensation $(TC\boxtimes TC)_1$ is the same as the final
condensation $(TC\boxtimes TC)_k$ in the measurement sequence. 

We now consider the possible $1$-component disorder models that can be made that
involve our chosen measurement sequence. By Eq.~\eqref{eq:compute}, there are
only $6$ possible effects that changing the disorder parameter from $p=0$ to
$p=1$ (or equivalently, $p=1$ to $p=0$) can have on the enacted automorphism: 
\begin{enumerate}[label=(\arabic*)]
    \item $\alpha \mapsto \alpha + 1 \mod 2$; 
    \item $\beta \mapsto \beta + 1 \mod 2$; 
    \item $\alpha \mapsto \alpha + 1 \mod 2$ and $\varphi_i \mapsto \varphi_i'$
      by adding an additional condensation step in the first $CC$ layer prior to
      $(TC\boxtimes TC)_1$; 
    \item $\alpha \mapsto \alpha + 1 \mod 2$ and $\varphi_f \mapsto \varphi_f'$
      by adding an additional condensation step in the first $CC$ layer after
      $(TC\boxtimes TC)_k$; 
    \item $\beta \mapsto \beta + 1 \mod 2$ and $\varphi_i \mapsto \varphi_i''$
      by adding an additional condensation step in the second $CC$ layer prior
      to $(TC\boxtimes TC)_1$; and 
    \item $\beta \mapsto \beta + 1 \mod 2$ and $\varphi_f \mapsto \varphi_f''$
      by adding an additional condensation step in the second $CC$ layer after
      $(TC\boxtimes TC)_k$. 
\end{enumerate}
We can show that these $6$ options realize only $4$ different automorphisms. For
each of (3)-(6) there is only one choice of condensation boson that we can
introduce that is a mutual-semion with the  condensates before and after it.
Thus, since $(TC\boxtimes TC)_1 = (TC\boxtimes TC)_k$ and $\varphi_i =
\varphi_f$, we must have $\varphi_i' = \varphi_f'$ and $\varphi_i'' =
\varphi_f''$. Let $\varphi_{(3)}$ be the enacted automorphism when we follow
option (3),
\begin{equation}
    \varphi_{(3)} = \varphi_i[(rx)(gy)(bz)]\varphi_i'^{-1},
\end{equation}
where we use $\varphi_i = \varphi_f$. Similarly, let $\varphi_{(4)}$ be the
enacted automorphism for option (4), 
\begin{equation}
    \varphi_{(4)} = \varphi_i'[(rx)(gy)(bz)]\varphi_i^{-1},
\end{equation}
where we use $\varphi_i' = \varphi_f'$. Finally, by noting that all $\varphi \in
\mathcal C\{(c\sigma)(c\sigma)(c\sigma)\}$ satisfy $\varphi = \varphi^{-1}$, we
have 
\begin{align}
    \varphi_{(3)} &= \varphi_{(3)}^{-1} \notag \\ 
    &= \left(\varphi_i[(rx)(gy)(bz)]\varphi_i'^{-1}\right)^{-1} \notag \\ 
    &= \varphi_i'[(rx)(gy)(bz)]\varphi_i^{-1} \notag \\ 
    &= \varphi_{(4)}.
\end{align}
An equivalent argument shows that $\varphi_{(5)} = \varphi_{(6)}$. This means
that a given measurement sequence for $\mathbbm 1$ can only be in $1$-component
disorder models with $4$ other FETs. However, there are $6$ FETs that are
connected to $\mathbbm 1$ (corresponding to the $6$ elements of $\mathcal
C\{(c\sigma)(c\sigma)(c\sigma)\}$). We therefore cannot construct a
$2$-component disorder model with $\mathbbm 1$ and all possible $A_0$ and $A_2$. 

If we were to extend the disorder model definition to allow for a measurement
sequence involving intermediary $\widetilde{CC}$ stages, such as $\widetilde{CC}
\rightarrow \mathcal A_1 \rightarrow \cdots \rightarrow \widetilde{CC}
\rightarrow \mathcal B_1 \rightarrow \cdots \rightarrow \widetilde{CC}$, then it
becomes possible for any length-$m$ adjacency sequence to form an $m$-component
disorder model.\footnote{Specifically, if $\{A_0,\,A_1,\ldots,A_m\}$ is an
adjacency sequence (cf. Definition~\ref{def:connected}), we can construct the
$m$-component disorder model by beginning with the measurement sequence for
$A_0$. Once returned to $\widetilde{CC}$, we append the $1$-component disorder
model with automorphisms $\id$ for $p_1=0$ and $\tau_{10}$ for $p_1=1$.
Repeating this for all $\tau_{i(i-1)}$ up to $i=m$ gives the required
$m$-component disorder model. This is similar to the procedure in
Eq.~\eqref{eq:concatenated}, but now the freedom to return to $\widetilde{CC}$
allows us to concatenate multiple $1$-component disorder models.} This freedom,
however, strays the model further from its interpretation as disordered DA color
codes that are designed to implement logical gates on an encoded system. 

This could be generalized further to allow condensation of any bosons of
$CC\boxtimes CC$. For example, starting from $\widetilde{CC}$ and condensing
$\tsf{rx}_1$ (with no boson in the second layer), then $\tsf{bx}_1$, then
returning to $\widetilde{CC}$ includes only reversible transitions (all logical
operators update reversibly); this allows us to make changes to the measurement
sequences in ways previously forbidden [i.e., $TC(\tsf{rx})\rightarrow
TC(\tsf{bx})$ is not a reversible transition]. We may also consider models with
anticorrelated disorder between two link types: measuring $\tsf{rx}_1$ with
probability $p$ and $\tsf{bx}_1$ with probability $1-p$, for example. Even with
condensed bosons in the second layer, this would not result in an irreversible
phase because the regions of $\tsf{rx}$- and $\tsf{bx}$-measurements do not
overlap. In the case of Eq.~\eqref{eq:id1}, this would enable removing the
$\tsf{gz}_1$ condensation while retaining reversible transitions, and thus
Eq.~\eqref{eq:id1} and Eq.~\eqref{eq:id2} can be connected. 

\subsection{\label{app:irrev}Irreversible Phases}
As discussed in Section~\ref{sec:disordermodel}, irreversible phases arise 
when the measurement sequences consecutively condense bosons that braid
trivially. These generally result in logical information measured out due to the
condensation of commuting anyons \cite{kesselringAnyon2024}. Moreover, there are
also disruptions to the ISG of the code, such that we do not return to the
$\widetilde{CC}$ ISG after the period. For intralayer scenarios the plaquette
operators are not reintroduced and links remain from the intermediary $TC$
phases. For interlayer scenarios the two $CC$ layers are not recoupled,
resulting in logicals that may reside on only one layer, for example. We
describe these behaviors in more detail here.

We first consider irreversible $TC\boxtimes TC$ intralayer transitions, starting
from the illustrative case of child theories of $CC$. Measuring a link
corresponding to the hopping operator for anyon $\tsf{c}\sigma$ has two effects
on the ISG: (1) remove plaquettes of color $c$ and flavor $\sigma'$ for $\sigma'
\neq \sigma$; and (2) add plaquettes of color $c'$ and flavor $\sigma$ for
$c'\neq c$.\footnote{Technically, $c\sigma$-links are added, not plaquette
  terms. However, the product of $c\sigma$ links around a $c' \neq c$ hexagon is
equivalent to the $c'\sigma$-plaquette operator.} By tracking the presence of
these plaquettes at each measurement stage, we can determine if the ISG is
reproduced. For example, contrast the case of 
\begin{equation}
  \magicp{$\bullet$}{}{}{}{}{}{}{}{} \rightarrow
  \magicp{}{}{}{}{}{$\bullet$}{}{}{} \rightarrow
  \magicp{}{}{}{}{}{}{$\bullet$}{}{}
\end{equation}
with 
\begin{equation}
    \magicp{$\bullet$}{}{}{}{}{}{}{}{} \rightarrow
    \magicp{}{}{}{}{}{}{$\bullet$}{}{}. 
\end{equation}
In the following table we denote by $\bullet$ the presence of each type of
plaquette operator in the ISG at each stage of the measurement sequences. For
the first example, we have:
\begin{equation}
  \begin{tabular}{c|c|c|c||c|c|c} 
    & $rx$ & $gx$ & $bx$ & $rz$ & $gz$ & $bz$ \\ 
    \hline \hline
    $TC(\tsf{rx})$ & $\bullet$ & $\bullet$ & $\bullet$ & & $\bullet$ & $\bullet$
    \\ 
    $TC(\tsf{gz})$ & $\bullet$ & & $\bullet$ & $\bullet$ & $\bullet$ & $\bullet$
    \\ 
    $TC(\tsf{bx})$ & $\bullet$ & $\bullet$ & $\bullet$ & $\bullet$ & $\bullet$ &
  \end{tabular}
\end{equation}
All plaquettes are recovered as expected and we realize the $TC(\tsf{bx})$
phase. In the second example, however, we have 
\begin{equation}
  \begin{tabular}{c|c|c|c||c|c|c} 
    & $rx$ & $gx$ & $bx$ & $rz$ & $gz$ & $bz$ \\ 
    \hline \hline
    $TC(\tsf{rx})$ & $\bullet$ & $\bullet$ & $\bullet$ & & $\bullet$ & $\bullet$
    \\ 
    $TC(\tsf{bx})$ & $\bullet$ & $\bullet$ & $\bullet$ & & $\bullet$ &
  \end{tabular}
\end{equation}
In this final state, there are individual $rx$ and $bx$-type links, but no
$rz$-type plaquettes. The effect of this is that a logical $\bar O(\tsf{bz})$,
for example, no longer has representatives that extend across all homologous
cycles of the torus, since there are no $rz$-plaquettes in the ISG with which to
multiply to deform one string into another. 

Additional disruptions occur when we consider $TC\boxtimes TC \rightarrow
\widetilde{CC}$ interlayer irreversible transitions. In these, we have a
situation where the measured $\Z_1\Z_2$ interlayer links will not always
anticommute with some element of the ISG (with the precise scenario determined
by the random disorder realization). These links are therefore not added to
the ISG, and we produce a phase with interlinked $CC$ layers that behaves
differently to $\widetilde{CC}$. In particular, there exist $\X$-logical
operators that reside on just one layer, e.g. $\bar O(\tsf{bx}_1)$, or
$\Z$-logical operators that cannot freely switch between the two layers as they
may in $\widetilde{CC}$ (wherein $\tsf{rz}_1 \sim \tsf{rz}_2$, for example).

\section{\label{app:critical}Trajectories, Phase Transitions and Critical Behaviour}

In this section, we introduce the idea of ``trajectories'' in parameter space as
another interpretation of $m$-component disorder models. We then formalize the
idea of competing automorphism phases (Section~\ref{sec:competing}) by linking
the critical behavior of trajectories to the universality class of bond
percolation, aided by numerical simulations.

In the main text, we treated an $m$-component disorder model as having $m$
independent parameters. This is not required, however. Within an $m$-dimensional
parameter space we can, for example, define a $1$-parameter ``trajectory''
$\mathbf p(p) \in [0,1]^{m}$ that interpolates between the FETs at two different
corners, with $p\in[0,1]$ such that $\mathbf p(0)$ realizes FET $A$ and $\mathbf
p(1)$ realizes FET $B$. For example, in the case of
Eq.~\eqref{eq:ex1_automorphisms}, $\mathbf p(p) = (p,p)$ maps between the FETs
with automorphism $\id$ at $p=0$ and $(rgb)(xzy)$ at $p=1$. The results for
logically-connected FETs can also be readily modified to trajectories: we call a
trajectory logically-protected if there exists a consistent nonzero-dimensional
logical Hilbert subspace that remains unmeasured in the limit of
$t\rightarrow\infty$ periods at any point $p\in[0,1]$ of the trajectory. If two
FETs are joined by a logically-protected trajectory, they must also be
logically-connected. Conversely, an $m$-component disorder model containing an
irreversible phase must support a trajectory that is not logically-protected. 

Moreover, these trajectories allow us to smoothly interpolate the measurement
sequence of one FET into that of another. As noted in
Section~\ref{sec:logically}, we cannot distinguish two points $p, \tilde p \in
[0,1]$ in a logically-protected trajectory solely using the time evolution of an
observable from the protected algebra. This means that an ``adiabatic
transition'' where a system evolves between two FETs using a time-dependent
trajectory $p(t) = t/T \in [0,1]$ for $t =0,1,\ldots,T$ will not affect the
periodic behavior of any logical operator that is protected. Observables that
belong to not protected qubits, however, can distinguish these two phases. This
allows us to detect phase transitions between different FETs and determine their
critical behavior.

\begin{figure}
  \begin{center}
    \includegraphics[width=0.48\columnwidth]{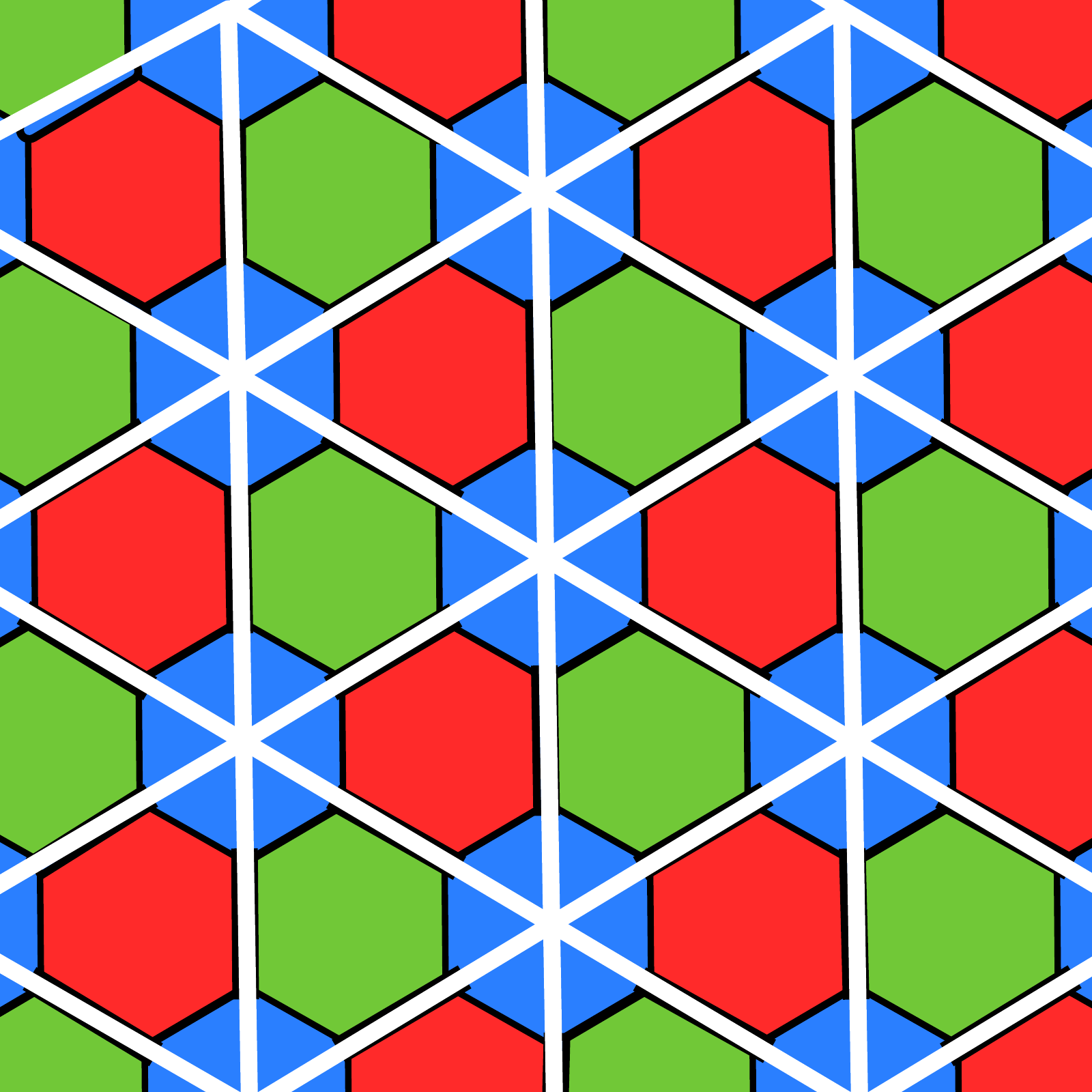}
    \includegraphics[width=0.48\columnwidth]{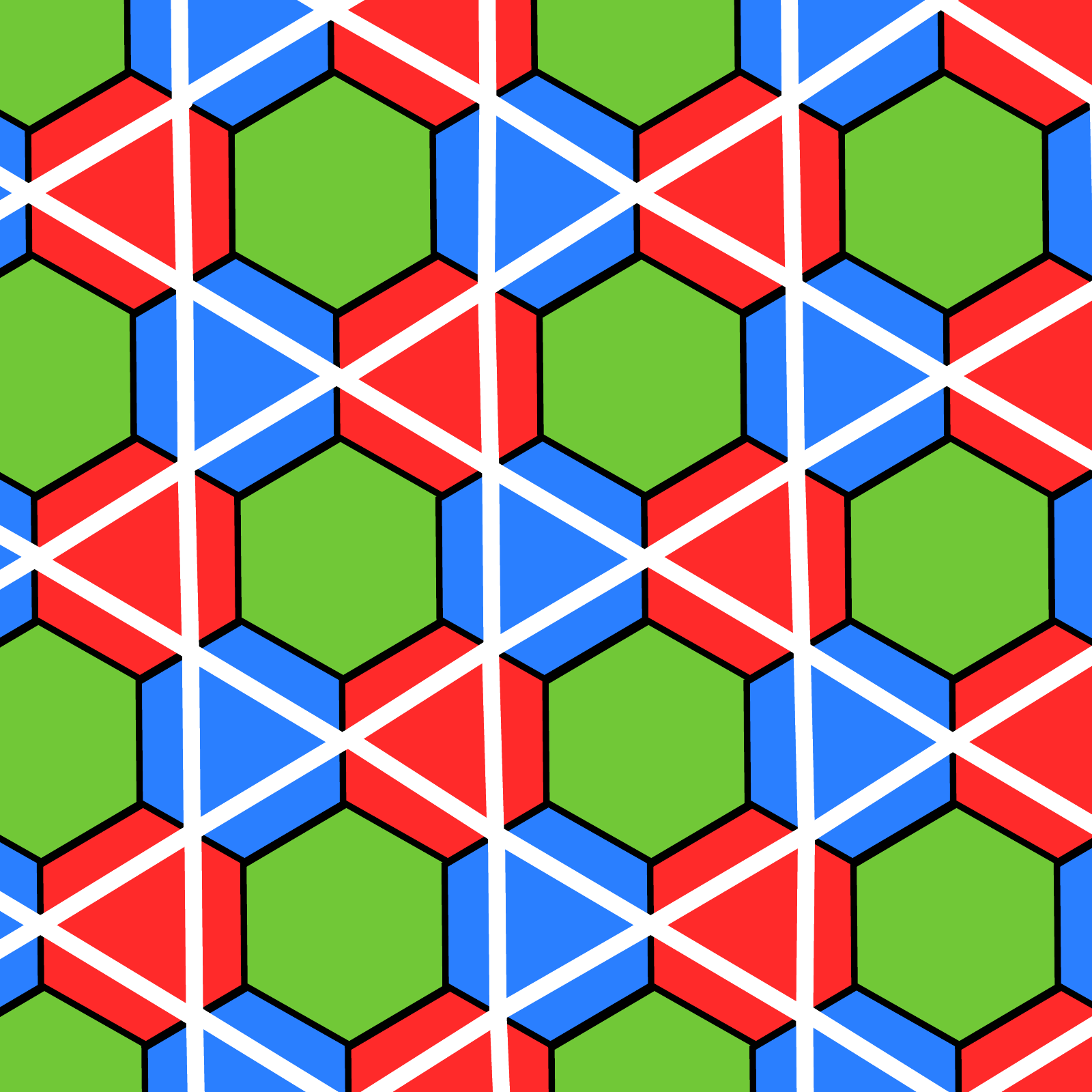}
  \end{center}
  \caption{{Left.} In a model with blue links disordered, we contract the
  remaining green and red links to a point to form a triangular superlattice
  (shown in white). Spanning domain walls are equivalent to percolating bonds on
  this superlattice. {Right.} In a model with blue and red links disordered,
  we contract the remaining green links to a point to form a kagome superlattice
  (shown in white).}
  \label{fig:triangularkagome}
\end{figure}

To determine the critical behavior, we first provide some intuition as to the
universality class. In a $1$-component disorder model, links of one color are
selected randomly. We thus contract each of the other two colored links to a
point to consider only the behavior of the disordered links. In the honeycomb
lattice, this contraction leaves behind the bonds of a triangular superlattice,
cf. Fig.~\ref{fig:triangularkagome}. Each bond is chosen independently with the
same probability, and a temporal domain wall containing a contiguous region
extending around a noncontractible cycle is in direct correspondence to the
existence of spanning clusters of these chosen bonds. Therefore, the critical
behavior of a $1$-component disorder model is expected to be in the same
universality class as bond percolation on a triangular lattice with a critical
parameter $p_{c,\text{triangular}} = 0.347\ldots$ and exponent $\nu=1.\dot3$
\cite{staufferIntroduction2003, sykesExact1964}.

\begin{figure*}
  \begin{center}
    \includegraphics[width=\textwidth]{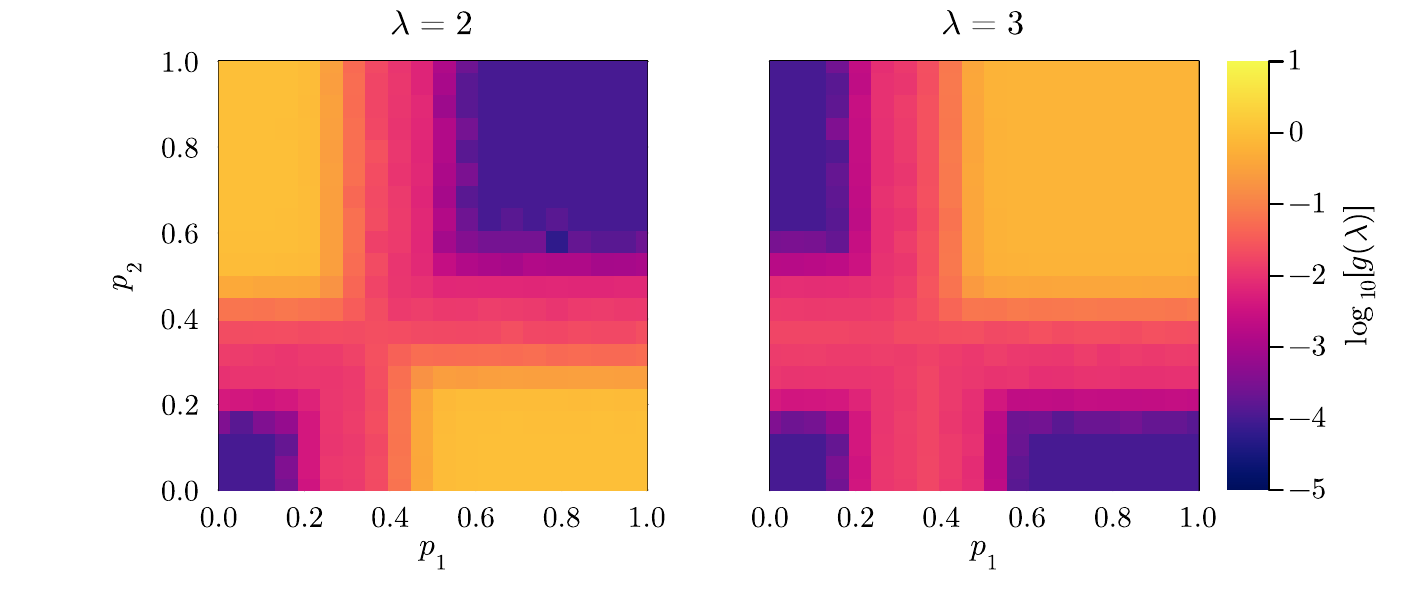}
  \end{center}
  \vskip -1cm
  \caption{$\pi$ and $2\pi/3$ ($\lambda=2$ and $\lambda=3$) Fourier components
  of the average-squared-expectation of the $\bar\X_3$ logical operator computed
  from an initial logical state of $\ket{\text{++++}}$ and evolving under
  Eq.~\eqref{eq:ex1} for $t \leq 96$ time steps. This initial state and
  observable were chosen to demonstrate the different FETs at the
  corners of the parameter space. We see nonzero period-$2$ oscillations in the
  $(p_1,p_2) = (0,1)$ and $(1,0)$ corners, and nonzero period-$3$ oscillations
  in the $(1,1)$ corner, consistent with the predicted automorphisms, cf.
  Eqs.~\eqref{eq:ex1_automorphisms}, \eqref{eq:evolve1}, \eqref{eq:evolve2}, and
  \eqref{eq:evolve3}.}
  \label{fig:ex1_g}
\end{figure*}

In a general $m$-component disorder model, more complicated behaviors emerge.
Disordered links in different $CC$ layers operate independently of each other,
and therefore it is possible for multiple triangular-bond percolation problems
to occur concurrently, potentially with different probability parameters
dependent on the point in parameter space. If two different-colored links on the
same layer are in consecutive disordered stages, however, then the universality
class changes. Now, we contract only the one set of non-disordered links, and
consider the spanning set of both disordered colors. These form the bonds of a
kagome superlattice, cf. Fig.~\ref{fig:triangularkagome}. If both colors are
chosen independently with the same probability, then the critical behavior now
is expected to follow the universality class of bond percolation on a kagome
lattice with critical parameter $p_{c,\text{kagome}} = 0.524\ldots$ and exponent
$\nu=1.\dot3$ \cite{staufferIntroduction2003, sykesExact1964}. We note that
these same two percolation behaviors can be applied to, and were indeed observed
in, the disordered honeycomb Floquet code with missing measurements
\cite{vuStable2024}. In our model, the increased parameter space enables
combinations of both percolation problems to arise concurrently.

We now provide evidence for these claims by performing numerical simulations of
disordered DA color codes. This analysis was done in Julia, using the
$\textsf{QuantumClifford.jl}$ package \cite{QuantumClifford} for efficient
computation with the stabilizer formalism \cite{aaronsonImproved2004,
gottesmanStabilizer1997, gottesmanHeisenberg1998, poulinStabilizer2005}. The
system is on a honeycomb lattice characterized by linear system size $L$, such
that there are $L$ plaquettes in each of the horizontal and vertical
directions,\footnote{Equivalently, $2L$ lattice sites in the horizontal
direction and $L$ lattice sites vertically.} and joined by periodic boundary
conditions on all sides. Unless specified otherwise, simulations were repeated
$N=508$ times using $L=18$.

We use two metrics: firstly, the Fourier components of the
average-squared-expectation of an observable $O$ for a given initial state
$\ket\psi$: 
\begin{equation}
  g_{O,\ket\psi}(\lambda) = \lim_{T\rightarrow \infty} \frac{2}{T}
  \sum_{t=0}^{T-1} e^{2i\pi t/\lambda} G_{O,\ket\psi}(t)  
\end{equation}
where 
\begin{equation}
  G_{O, \ket\psi}(t) = \overline{\braket{\psi(t)|O|\psi(t)}^2}
\end{equation}
and $\psi(t)$ is the evolution of the initial state after an integer $t\geq0$
periods of a given measurement sequence. We can use $g(\lambda)$ to distinguish
between the subcritical and supercritical phases of a trajectory by choosing a
state and observable that evolve differently under the automorphisms at the
two endpoints. For example, a $3$-cycle automorphism at $\mathbf p(0)$ and a
$2$-cycle automorphism at $\mathbf p(1)$ can result in a nonvanishing $g(3)$
when $p=0$ and a nonvanishing $g(2)$ when $p=1$. In practice, we need to
truncate the limit at some finite $T'$ that must be a multiple of the periods
under study; this ensures that the Fourier decomposition equation is valid. We
chose $T'=96$ for our simulations: a multiple of $2,3,4$, and $6$.

Secondly, we consider the purification dynamics of the system. We track the
evolution starting from a maximally-mixed logical state $\rho$ by measuring its
(average) von Neumann entropy $S = \rho \log\rho$ over time
\cite{fattalEntanglement2004}. At $t=0$, we start from the maximum $4$. Over
multiple periods of the measurement sequence, entropy reduces if logical qubits
are measured. A logically-protected trajectory must retain $S(t) > 0$ in the
limit $t\rightarrow \infty$ at all points $p$. To model this, we assume the form 
\begin{equation}
  S(t) = S_\infty + (S_0 - S_\infty)e^{-\Gamma t}
  \label{eq:entropy}
\end{equation}
and consider the decay rate $\Gamma$ or associated timescale $\tau = 1/\Gamma$.
We again truncate using $t \leq 96$ to approximate $S_\infty$ and
$\tau$ in numerical simulations.

We consider now an example of a $2$-component disorder model. Specifically, take
the measurement sequence given by Eq.~\eqref{eq:ex1} with automorphisms in
Eq.~\eqref{eq:ex1_automorphisms}. All corners of the $2$-dimensional parameter
space are FETs, and the associated automorphisms are even-parity on the
$S_3\times S_3$ subgroup. We choose $\ket\psi$ such that it is an
eigenstate of logical operators that evolve differently in the corners of the
phase diagram. This is easily observed using the stabilizer picture
\cite{nielsenQuantum2010}; take $\bar \X_3 = \bar\Os[\tsf{bx}]_h$, for example.
It does not change when evolving under the identity. Under $(rx)(gy)(bz)$ we
see
\begin{equation} 
  \bar\Os[\tsf{bx}]_h \mapsto \bar\Os[\tsf{rz}]_h \mapsto \bar\Os[\tsf{bx}]_h
  \mapsto \cdots 
  \label{eq:evolve1}
\end{equation}
while under $(rz)(gx)(by)$ it evolves as 
\begin{equation}
  \bar\Os[\tsf{bx}]_h \mapsto \bar \Os[\tsf{gy}]_h \mapsto \bar\Os[\tsf{bx}]_h
  \mapsto \cdots 
  \label{eq:evolve2}
\end{equation}
Finally, under $(rgb)(xzy)$ we get 
\begin{equation} 
  \bar\Os[\tsf{bx}]_h \mapsto \bar \Os[\tsf{rz}]_h \mapsto \bar\Os[\tsf{gy}]_h
  \mapsto \bar\Os[\tsf{bx}]_h \mapsto \cdots 
  \label{eq:evolve3}
\end{equation}
Notably, it returns to an eigenstate of $\bar\X_3$ only after $2$ Floquet
periods in $(rx)(gy)(bz)$ and $(rz)(gx)(by)$, and after $3$ periods in
$(rgb)(xzy)$. At other times the operator maps to $\bar{\textsf{Y}}$ or
$\bar\Z$, and the expectation of $\bar\X_3$ is $0$. Starting in a
$\plus1$-eigenstate of $\bar\X_1$, $\bar\X_2$, $\bar\X_3$, and $\bar\X_4$ (denoted
$\ket{\text{++++}}$), the squared expectation of $\bar\X_3$ should show
period-doubling oscillations between $1$ and $0$ at $\mathbf p = (0,1)$ and
$(1,0)$, and period-tripling behavior at $\mathbf p = (1,1)$. The Fourier
components of $G_{\bar X_3,\, \ket{\text{++++}}}$ are plotted in
Fig.~\ref{fig:ex1_g}, with nonvanishing values of $g(2)$ and $g(3)$ appearing
only in these predicted corners. 

\begin{figure}
  \begin{center}
    \includegraphics[width=\columnwidth]{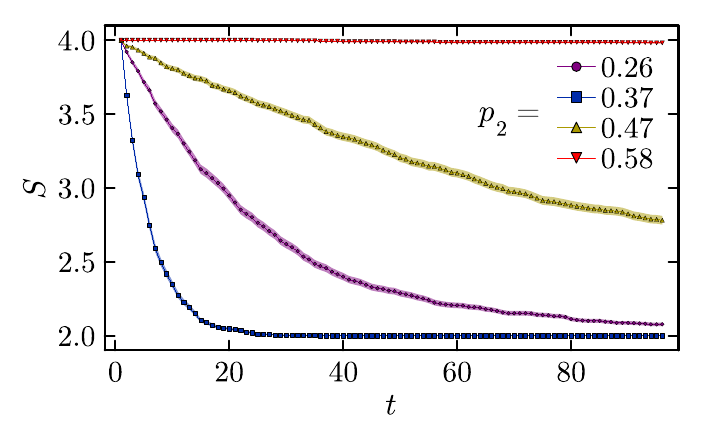}
  \end{center}
  \vskip -0.7cm
  \caption{Average von Neumann entropy $S = \rho\log\rho$ of a maximally-mixed
  logical state with $S_0=4$ and evolving under Eq.~\eqref{eq:ex1} at various
  values of $p_2$ with $p_1= 0$ (thus reducing to a $1$-component disorder
  model). Ribbon shows the standard error of the mean based on $N=508$
  repetitions. Near the critical value of $p_c = 0.347\ldots$, the entropy
  approaches the long-term value of $S_\infty = 2$. Away from there, the entropy
  remains at $S=4$.}
  \label{fig:ex1_entropy}
\end{figure}

\begin{figure}
  \begin{center}
    \includegraphics[width=0.99\columnwidth]{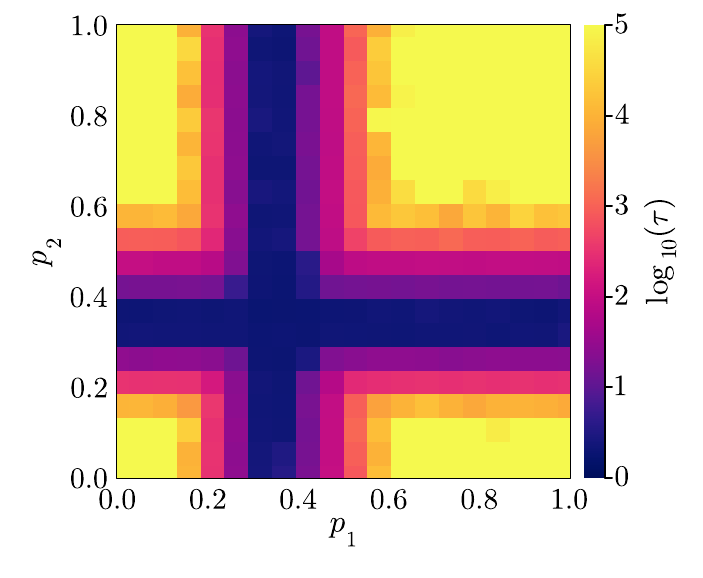}
  \end{center}
  \vskip -0.7cm
  \caption{Average purification timescale $\tau=1/\Gamma$ from
  Eq.~\eqref{eq:entropy} for a maximally-mixed logical state with $S_0=4,
  S_\infty = 2$ evolving under Eq.~\eqref{eq:ex1} at various values of $p_1,
  p_2$, up to $t \leq 96$. Exactly $2$ logical qubits are measured out when
  tuned near the critical lines at $p_c \sim 0.35$ (dark blue regions).}
  \label{fig:ex1_tau}
\end{figure}

\begin{figure}[t]
  \begin{center}
    \includegraphics[width=0.99\linewidth]{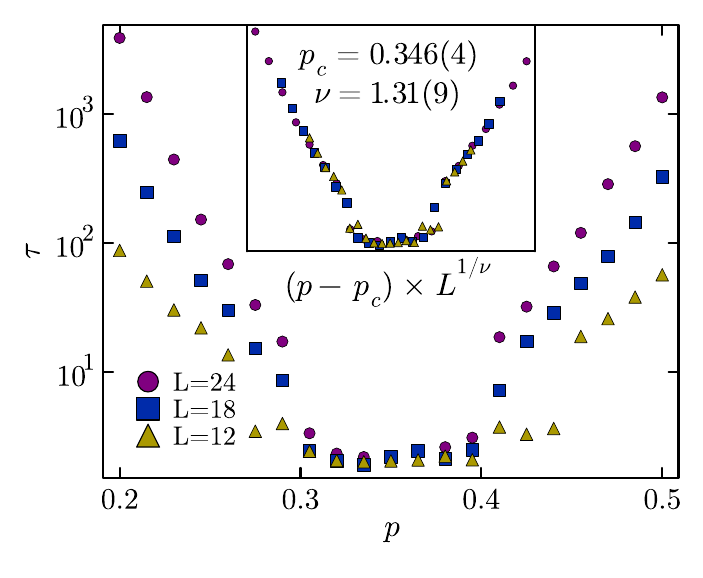}
  \end{center}
  \vskip -0.7cm
  \caption{Average purification timescale $\tau$ for a maximally-mixed logical
  state evolving under Eq.~\eqref{eq:ex1} with the trajectory $\mathbf p(p) =
  (0, p)$ near the critical point. We run the simulation at various values of
  linear system size, $L$. Inset shows the finite-size scaling collapse under
  the estimated parameters of $p_c = 0.346(4)$ and $\nu=1.31(9)$ using the
  functional form $(p-p_c)L^{1/\nu}$. These are consistent with the theoretical
  $p_c= 0.347\ldots$ and $\nu=1.\dot{3}$ for bond percolation on the triangular
  lattice.}
  \label{fig:ex1_crit}
\end{figure}

These automorphisms in Eq.~\eqref{eq:ex1_automorphisms} satisfy the conditions
in Section~\ref{sec:logically} and therefore it is possible for there to exist
a logically-protected trajectory through the parameter space. We first plot the
purification dynamics in Fig.~\ref{fig:ex1_entropy}, for illustrative values of
$\mathbf p$. As expected from Section~\ref{sec:competing}, two logical qubits
are measured out near the critical point, which we find to be around $p \sim
0.37$. Taking a parameter sweep of all $p_1, p_2 \in [0,1]$, we get
Fig.~\ref{fig:ex1_tau} that shows the average purification decay rate
$\tau=1/\Gamma$. Near the critical lines $p\sim0.35$, we get a finite decay
rate, with $\tau \rightarrow \infty$ elsewhere. To determine this critical
value more precisely, we take the trajectory $\mathbf p(p) = (0,p)$.
Figure~\ref{fig:ex1_crit} shows the average purification dynamics at different
values of linear system size $L$, and presents a scaling-collapse of the form
$(p-p_c)L^{1/\nu}$ that shows scale-invariant behaviour \cite{vuStable2024}.
Using finite size scaling methods \cite{staufferIntroduction2003}, we estimate
critical values of $p_c = 0.346(4)$ and $\nu=1.31(9)$, consistent with the
theoretical values for bond percolation on a triangular lattice
\cite{sykesExact1964}. The trajectory $\mathbf p(p) = (p,p)$ that traverses
along the diagonal of Fig.~\ref{fig:ex1_crit} will also exhibit the same
critical values.

\end{document}